\documentclass[acmsmall,screen,nonacm]{acmart}
\pdfoutput=1 

\usepackage{hyperref}
\usepackage{xcolor}
\usepackage{listings}
\usepackage{verbatim}

\usepackage{cleveref}
\usepackage{tikz}
\usepackage{amsmath}
\usepackage{pifont}
\usepackage[binary-units]{siunitx}

\newcommand{\multicrefrange}[4]{%
    \Crefrange{#1}{#2}, \labelcref{#3}-\labelcref{#4}
}

\newcommand{\codestyle}[1]{{{\fontsize{9}{10}\ttfamily #1}}}

\newcommand{\size}{\codestyle{size}}
\newcommand{\ins}{\codestyle{in\-sert}}
\newcommand{\del}{\codestyle{de\-lete}}
\newcommand{\contains}{\codestyle{con\-tains}}
\newcommand{\cas}{\texttt{CAS}}
\newcommand{\nul}{\texttt{NULL}}

\newcommand{\spsize}{\texttt{SP}}

\newcommand{\term}[1]{\textit{#1}}

\newcommand{\gnote}[1]{{\color{teal}\begin{quote}{\bf Galy's note:} #1\end{quote}}}
\newcommand{\ignore}[1]{}

\newtheorem{theorem}{Theorem}[section]
\newtheorem{claim}[theorem]{Claim}
\newtheorem{observation}[theorem]{Observation}

\acmYear{2025}\copyrightyear{2025}
\setcopyright{rightsretained}
\acmConference[SPAA '25]{37th ACM Symposium on Parallelism in Algorithms and Architectures}{July 28--August 1, 2025}{Portland, OR, USA}
\acmBooktitle{37th ACM Symposium on Parallelism in Algorithms and Architectures (SPAA '25), July 28--August 1, 2025, Portland, OR, USA}
\acmDOI{10.1145/3694906.3743316}
\acmISBN{979-8-4007-1258-6/25/07}
    
\begin{document}

\title{A Study of Synchronization Methods for Concurrent Size}
\titlenote{This work was supported by the Israel Science Foundation, Grant No.~1102/21.}

\author{Hen Kas-Sharir}
\email{hen.kas@campus.technion.ac.il}
\affiliation{
  \institution{Technion}
  \city{Haifa}
  \country{Israel}
}
\orcid{0009-0005-3274-528X}

\author{Gal Sela}
\email{gal.sela@epfl.ch}
\affiliation{
  \institution{EPFL}
  \city{Lausanne}
  \country{Switzerland}
}
\orcid{0000-0003-2342-6955}

\author{Erez Petrank}
\email{erez@cs.technion.ac.il}
\affiliation{
  \institution{Technion}
  \city{Haifa}
  \country{Israel}
}
\orcid{0000-0002-6353-956X}

\begin{abstract}
The size of collections, maps, and data structures in general, constitutes a fundamental property. An implementation of the size method is required in most programming environments. Nevertheless, in a concurrent environment, integrating a linearizable concurrent size introduces a noticeable overhead on all operations of the data structure, even when the size method is not invoked during the execution. 
In this work we present a study of synchronization methods in an attempt to improve the performance of the data structure. In particular, we study a handshake technique that is commonly used with concurrent garbage collection, an optimistic technique, and a lock-based technique. 
Evaluation against the state-of-the-art size methodology demonstrates that the overhead can be significantly reduced by selecting the appropriate synchronization approach, but there is no one-size-fits-all method. Different scenarios call for different synchronization methods, as rigorously shown in this study. Nevertheless, our findings align with general trends in concurrent computing. In scenarios characterized by low contention, optimistic and lock-based approaches work best, whereas under high contention, the most effective solutions are the handshake approach and the  wait-free approach.
\end{abstract}



\maketitle

The conference version of this paper is available at \cite{kas2025study}, 
and the code is publicly available at \cite{artifactSyncMethodsForConcurrentSize}.

\section{Introduction}\label{section:intro}
With the proliferation of multicores in modern computing architectures, the significance of concurrent programming has become acute.
Concurrent data structures 
are a key component of concurrent systems, making their correctness and performance crucial.
A fundamental property of data structures is their \size{}, which is the number of elements they contain. 
The \size{} method is widely used and its implementation is required for collections in many programming environments. For example, in Java, in order to implement the elementary \texttt{Collection} and \texttt{Map} interface classes, one is required to implement the \size{} method. 

In spite of the importance of the \size{} method, it was not known until recently how to compute the \size{} of a concurrent data structure efficiently and correctly. Traditional methods were either very slow or incorrect. For example, taking a concurrent snapshot of the data structure and traversing it to count its elements is correct but inefficient. In contrast, naively maintaining a global \size{} variable and updating it with each operation 
is not linearizable (i.e., incorrect): 
Consider, for instance, $n$ threads, each inserting an element and
then becoming preempted before updating the global counter. At this point, the counter is still 0, even though the number of nodes in the data structure is $n$. If another thread calls
\contains{}() for the $n$ inserted items, it would observe them; but if it then calls \size{}(), it would still see 0---an inconsistent result.
Yet, this approach is currently used in some Java libraries with an adequate warning that the result may be inaccurate---for example, in \texttt{ConcurrentSkipListMap} and \texttt{ConcurrentHashMap} 
within the \texttt{java.util.concurrent} package. 

Recently, Sela and Petrank~\cite{sela2021concurrentSize} proposed a new mechanism for computing a linearizable concurrent \size{} for sets and dictionaries (denoted \spsize{}), significantly improving efficiency compared to previous (linearizable) methods. While their algorithm is linearizable and relatively efficient, it still incurs non-negligible overheads. Specifically, a performance cost of 1--20\% is incurred on the standard operations of the data structure (i.e., \ins{}, \del{}, \contains{}). This overhead comes from the need for data structure operations to always cooperate with a potential concurrent \size{} method, even if a \size{} operation is not currently active. Such an overhead could discourage users from adopting it, especially for workloads that execute the \size{} operation infrequently, or not at all.

A natural question that arises is whether alternative synchronization methods could offer a reduced overhead. The goal of this paper is to provide a study of diverse synchronization methods aimed at mitigating the overhead required for supporting a linearizable concurrent \size{}. Our study covers a variety of synchronization methods, including handshake synchronization --- that is often used with concurrent garbage collection, optimistic synchronization, and lock-based synchronization. All examined methods are incorporated in a way that provides a correct (linearizable) concurrent \size{}, and our goal is to identify the best scheme for common workloads. We examine the overhead on standard data structure operations 
as well as the efficiency of the concurrent \size{} operation itself. The various synchronization methods are compared with the wait-free \spsize{}, the most efficient linearizable \size\ method currently available.

The first synchronization method that we study is handshakes, a concept initially introduced in the context of concurrent garbage collection \cite{DoligezPortable,DoligezConcurrentHS,ImplementingOnTheFlyHS,leva01}. Handshakes allow threads to adapt their behavior to different execution phases, originally, 
to enable threads to start cooperating 
when a concurrent garbage collection commences. 
In a similar spirit, we harness the handshake mechanism to allow threads to operate (almost) normally (on a fast path) when no concurrent  \size\ operation is active, while requiring them to cooperate fully (on a slow path) when a \size\ operation is invoked. This synchronization method offers a tradeoff: It reduces the overhead on all other operations when no \size\ operation is in progress. However, the \size\ operation itself suffers a reduced efficiency due to the need to invoke the handshake mechanism and wait for the cooperation of all threads before beginning to execute.  

The novel handshake mechanism that we design carefully manages transitions between the fast and slow paths. During these transitions, operations running in the slow path, which are aware of \size{}, may execute concurrently with operations in the fast path, which are not. Notably, the fast and slow paths have different linearization points. 
To ensure correctness, two handshakes are required when transitioning to the slow path, as demonstrated in our proof. Interestingly, no handshakes are needed when transitioning back to the fast path. We provide an intricate correctness proof to establish appropriate linearization points for these operations and prove their correctness.

The second synchronization method we study is optimistic. In this approach modifications to the data structure are carried out with almost no cooperation, except for announcing each modification. In tandem, the \size\ operation optimistically assumes that no updates of the data structure occur concurrently and attempts to optimistically collect local per-thread size variables in order to sum them and compute the overall size. If the collection is interrupted by a concurrent modification, the collection is retried. After several unsuccessful attempts, the method transitions to a non-optimistic approach, specially designed for this circumstance.  

Finally, we investigate the utilization of lock synchronization. In this context, the objective of the synchronization is to prevent the \size\ method from executing concurrently with data structure modifications. To improve performance, we employ a reader-writer lock, with the \size\ operation acquiring the write lock and updating operations acquiring the read lock. This setup allows concurrency between the updating operations, while preventing  concurrent execution of the \size\ operation, thereby ensuring correctness.

The approaches we examine are primarily centered on computing size for linearizable data structures implementing a set or a dictionary, but the underlying principles can potentially be applied to operations other than size and to various other data types. The optimistic approach and the lock-based one may be applied to any linearizable set or dictionary, while the handshake-based approach imposes an additional requirement detailed in \Cref{section:preliminary}.
All our methodologies produce linearizable size-supporting data structures.

Regarding progress guarantees, the handshake-based methodology preserves the original progress guarantees of \ins{}, \del{} and \contains{}. Specifically, if these operations were wait-free in the original data structure, they remain wait-free in the transformed version; similarly, lock-free or obstruction-free operations retain their original guarantees. The other methodologies preserve the original progress guarantee of \contains{} but make \ins{} and \del{} blocking. 
Notably, the added \size{} operation itself is blocking in all our methodologies, unlike the wait-free guarantee for \size{} provided by \spsize{}.



We have carefully designed and implemented a concurrent \size\ method with each of the aforementioned synchronization methods (along with additional methods that did not perform well and have thus been omitted from this report). The results appear in \Cref{section:evaluation}. It turns out that there is no one-size-fits-all solution. Different scenarios call for different synchronization methods. The observed results vary according to the chosen data structure, the levels of contention, the frequency of using the \size{} operation, and by workload (read intensive or update intensive). Nevertheless, our findings align with general trends in concurrent computing. In low contention scenarios, optimistic and lock-based approaches exhibit good performance, but their efficacy diminishes when contention increases. In such cases, \spsize{} and the handshake approach work best. 

Hash table operations execute noticeably faster than BST or skip-list operations,  resulting in a higher overhead when synchronizing a \texttt{size()} computation with hash table operations.  Specifically, the average overhead of utilizing \size{} with the BST is $2.4\%$ (with \spsize{}), while for the skip-list it is $4.4\%$ (with the handshakes approach). For the hash function, a typical workloads of $5\%$ updates and $95\%$ read operations incur a similar overhead of $4\%$ (with an optimistic approach). However, in write-oriented workloads, this overhead increases to $10\%$ (with \spsize{} performing best).

Additionally, we evaluated the performance of the \size{} operation when executed concurrently with data structure updates. Similar trends emerge in this context as well, but synchronization methods that optimize the scalability of the \size{} operation may differ from those that improve the performance of data structure operations. We expect data structure performance to be a priority for users in most cases, and thus its performance will probably guide the choice of the adopted synchronization method.


In~\Cref{section:preliminary} we define the basic terms used in this work. Previous solutions are discussed in~\Cref{sec-overview-sela}. A design of \size{} with each of the three synchronization methods is described in~\Cref{section:main-work}. The evaluation appears in~\Cref{section:evaluation} and we conclude in~\Cref{section:conclusion}. We add some measurements in~\Cref{sec:max_tries,sec:zipfian graphs}, describe some additional implementation details and optimizations in~\Cref{section:implementation-details}, and bring a correctness proof for the involved handshake mechanism in \Cref{handshakes:linearizability}.

\section{Preliminaries}\label{section:preliminary}

We consider a standard asynchronous shared memory model. The shared memory consists of a collection of shared objects supporting the following primitives, which are performed atomically:
\codestyle{read}, \codestyle{write}, \codestyle{fetch-and-add}, and \codestyle{compare-and-swap}. 
\codestyle{fetch-and-add} adds the input value to the current value and returns the previous value. \codestyle{compare-and-swap} conditionally updates the value if it matches an expected value.
These atomic instructions are commonly available in many modern processors.

An execution is said to be \term{linearizable} when each of its operations appears to be completed in an instant, occurring between its invocation and response and termed its \term{linearization point}, in accordance with the sequential specification of the data structures. A concurrent data structure is \term{linearizable} if all its executions meet this criterion~\cite{herlihy1990linearizability}.

A concurrent object is said to be \term{lock-free} if, at any point in time, at least one of the threads trying to access the object makes progress within a finite number of steps, even if other threads are paused or delayed. 
Furthermore, a concurrent object framework is \term{wait-free} when any operation by any thread can be completed within a finite number of steps, without being influenced by the operational speed of other threads~\cite{herlihy1991wait}.



All our methodologies may be applied to linearizable sets and dictionaries, while the handshake-based methodology also imposes the following additional requirement (which is the same requirement that the original wait-free mechanism imposes on the data structures it may transform to size-supporting data structures~\cite{sela2021concurrentSize}):
The \del{} operation in the original data structure must perform a marking step before physically unlinking.
Crucially, the linearization point of the \del{} operation within the original data structure must be at this marking step.


A \term{snapshot} of an object provides a captured state of that object at a specific point in time. For a set, this is essentially an atomic view of all the elements currently in it.


Finally, a \term{readers-writer lock} is a synchronization primitive allowing multiple threads to read or write to shared data without interfering with each other. It is made of two locks: a read lock and an exclusive write lock. Multiple threads can hold the read lock at the same time. If a thread has the read lock, no other thread can acquire the write lock.
Only one thread can hold the write lock at any given time. If a thread has the write lock, no other thread can acquire either the read lock or the write lock.


\section{The State-of-the-art Solution}\label{section:related-work}\label{sec-overview-sela}
Let us shortly review the work of Sela and Petrank~\cite{sela2021concurrentSize} that we extend. Their approach (denoted \spsize{}) utilizes two local variables per thread, one of which tracks the number of insertions and the other the number of deletions applied by the thread to the data structure. 
Each update of the data structure begins by executing the actual update, followed by an update to the associated count (of inserts or deletes). To obtain linearizability, an operation does not take effect until the count is properly updated. If a thread encounters an inserted node, for which the inserting thread has not yet updated its count, it helps updating the count before using this node. A similar action is taken with a node marked for deletion. The result of this helping mechanism is that operations can be regarded as being linearized at the time when the thread's local count is updated. To ensure linearizability, the \size\ method needs to take a linearizable snapshot on all local counts to obtain a consistent view of them and compute the overall size of the data structure. They utilize a specialized snapshot mechanism for that, which requires adjustment of the \ins{}, \del{} and \contains{} operations to support size snapshots. Further details appear in the original paper~\cite{sela2021concurrentSize}. 

An essential implementation component in the \spsize{} approach is the \codestyle{UpdateInfo} class. An object of this class is installed by \ins{} and \del{} operations in the nodes on which they operate, posting information about themselves, to enable other operations to observe modifications being made to these nodes and detect whether they should assist updating the associated metadata and how.
Each inserted node and each node marked for deletion is associated with an \codestyle{UpdateInfo} object. Particularly, the installment of an \codestyle{UpdateIn\-fo} object in a \codestyle{deleteInfo} field in a node signifies its marking for deletion.

\section{Memory model}
Our implementations rely on Java's memory mo\-del and use low-level atomic operations to efficiently ensure correctness. In particular, we employ the \texttt{VarHandle} API to perform atomic reads and writes with appropriate visibility and ordering guarantees across threads. 
These reads and writes are performed using the \texttt{getVola\-ti\-le} and \texttt{setVolatile} operations. To perform a \codestyle{fetch-and-add} operation we use \texttt{getAndAdd}, or \texttt{getAndIncrement} to specifically add 1. 
To perform a \codestyle{compare-and-swap} we use either \codestyle{compareAnd\-Set}, which returns a boolean to indicate success, or \codestyle{compareAndExc\-ha\-nge}, which returns the previous value.
See \Cref{subsection: Memory Model} for more information on the memory model.

\section{Synchronization Methods for \size{}}\label{section:main-work}
This paper studies three synchronization methods: handshakes, optimistic, and lock-based for \size{}. For each of these synchronization methods, we design a linearizable \size{} method focusing on improving efficiency. We then evaluate all of them against the original wait-free method of~\cite{sela2021concurrentSize}. In this section, we describe the design of size with each of these synchronization methods.  

\subsection{Handshakes}\label{section:handshakes}
\subsubsection{Overview}
We start with a handshake-based approach for evaluating the size of a data structure. The underlying idea is that many workloads perform data structure operations without requiring the data structure's size. Consequently, it appears inefficient to impose an overhead continuously for a potential \size{} operation, especially when it is infrequently executed. To address this, we aim to segment the program into distinct phases, ensuring that the overhead for the \size{} operation is only incurred when it is actually executed. When no \size{} operation is triggered, the program should execute with minimal overhead. Essentially, our goal is to establish a fast path, resembling normal data structure operations when no \size{} operation is invoked, and a slow path that manages concurrent \size{} operations, incurring the necessary overhead.

Technically, the design involves creating both a fast path and a slow path and establishing a mechanism for transitioning between these execution phases. We have opted for a straightforward approach: the fast path employs the original data structure operations without incurring overhead, while the slow path executes the implementation from \cite{sela2021concurrentSize}. However, transitioning between these phases is complex and involves some overhead on both the slow and fast paths. 

The mechanism for phase transitions, particularly upon starting to execute a \size{} operation, entails notifying all active threads that a \size{} operation is imminent, prompting them to shift to executing the slow path. This can be done by halting all operations until all threads complete their currently executing operations and switch to the slow path. However, such a cooperation can incur a performance overhead  as threads pause and remain idle while waiting for all other threads to complete their operations and switch to the required mode of operation. This cost becomes acute if phase changes occur frequently. Such a blocking cooperation also foils a progress guarantee that an individual operation may carry. 


Instead, we adopt a method commonly employed in memory management known as {\em handshakes}~\cite{DoligezPortable,DoligezConcurrentHS,ImplementingOnTheFlyHS,leva01}.
This involves an initiator (e.g., a thread that is about to execute \size{}) prompting a phase change by incrementing a global counter, thereby asking all other threads to acknowledge the phase change. Each thread checks this variable before starting to execute an operation, and responds by setting its local thread counter to match the global counter. According to the counter, a thread can tell whether a \size{} operation executes concurrently. Once all threads have responded, the handshake concludes, enabling the initiator to proceed, assured that all threads are aware of the phase change.

Notably, after responding to the handshake (by setting its local counter), a thread can continue to execute operations (on the slow path), with no need to wait for other threads, avoiding stalls required in the blocking  method, and maintaining wait-freedom of the original operations. However, this setup introduces a complication. As threads respond to the handshake individually and continue to execute operations, concurrent threads operate in different modes (fast and slow paths) simultaneously. Managing these concurrent thread executions in different modes requires careful handling.
In fact, it turns out that in order to guarantee correctness, prior to initiating the execution of the \size{} operation, we need to run two handshakes. While the necessity for two handshakes might not be intuitively evident, we will demonstrate that a single handshake does not guarantee correctness in \Cref{sec-2-handshakes} below. We prove that executing two handshakes before \size{} (and no handshakes at the end of the \size{} execution) is sufficient for ensuring correctness in \Cref{handshakes:linearizability}.

It is important to disregard idle threads that are not engaged in  operating on the data structure. These threads cannot cooperate with a handshake when they are engaged in other activity, unrelated to the data structure at hand, and  
the \size{} operation does not need to wait for them to acknowledge a phase change.  
On the one hand this is beneficial due to the time saved by not waiting for all threads. On the other hand, this requires any thread that starts an operation to notify all other threads that it is not idle anymore. Such notifications (including a memory fence) incur a non-negligible cost on very fast operations such as operations on the hash table.

\subsubsection{A \size{} design with handshakes}
The process of implementing handshakes for concurrent size entails designing the fast and slow paths, along with the mechanism to transition between different phases of execution. We will run two handshakes before starting the size calculation. During the first handshake, each thread will respond and start using the slow path. In the second handshake threads will only respond, acknowledging the handshake, without taking any further action. A thread that starts an operation after the first handshake was initiated, leaving an idle state, will use the slow path, maintaining the invariant that all threads execute only in the slow path after the first handshake completes, and until the completion of the \size{} execution.  
Once the size calculation is completed, a corresponding announcement by the \size{} suffices to let the threads revert to using the fast path. In \Cref{handshakes:linearizability} we will explain why a single handshake is insufficient at the beginning of a \size{} execution. We then argue that two handshakes are enough, and that no handshakes are needed at the end of the \size{} execution.

\subsubsection{Data-structure transformation}\label{handshakes:data-structure transformation}

To implement the fast and slow paths we employ two metadata arrays: \codestyle{fastMetadataCounte\-rs} designated for fast path operations and \codestyle{metadataCounters} for slow path operations. They are both fields of a \codestyle{HandshakeSizeCalc\-ulator} object we will later describe.
When an operation of a thread $T$ executes the fast path, it is guaranteed that no \size{} operation can execute concurrently. Therefore, it is guaranteed that, during the fast path execution of the operation, the fast path metadata associated with the thread $T$ is accessed only by $T$. Consequently, no synchronization is required for updating this part of the metadata, enhancing the efficiency of the fast path.
Conversely, slow path operations are executed in a manner similar to the wait-free \spsize{}, which takes extra care to allow data structure operations and updates to the slow path metadata to run concurrently with a \size{} operation.

Our data structure transformation scheme, detailed next, is illustrated in \Cref{fig: transformed data structure with handshakes}.
For each \ins{} or \del{} operation performed on the data structure, we have it first announce starting an operation (leaving the idle mode), then execute one of two new operations---\codestyle{slow\_op} or \codestyle{fast\_op}, and eventually announce returning to idle mode.
The \codestyle{fast\_op} operation executes the original operation and then updates the metadata associated with fast operations on a successful \ins{} or \del{}.
The \codestyle{slow\_op} executes the code of the \spsize{} algorithm. 

The slow path and fast path use different methods to mark an object as deleted. The slow path (following the \spsize{} algorithm) installs an \codestyle{UpdateInfo} object in a \codestyle{deleteInfo} field in order to mark the node, whereas the fast path, following the code of the original data structure, uses a simpler mark (typically setting the value field to \nul{} or setting the next field to point to a marker node). 

Since slow and fast operations might run concurrently during the execution of the first handshake and at the end of executing \size{}, we further adjust both slow and fast operations of the data structure to treat a node as marked when either the associated \codestyle{deleteInfo} field is not \nul{} or when the node is marked according to the original data structure's marking scheme. 

This change allows both slow and fast operations on the same key to execute concurrently. To complete the adjustment, we need to address the following issue: in the \spsize{} methodology, the data structure's operations call \codestyle{this.sizeCalculator.updateMetada\-ta} to help concurrent \del{} operations whenever they encounter marked nodes that they need to unlink. Slow operations in our transformation follow the same behavior, but only for nodes that slow operations marked (by installing an \codestyle{UpdateInfo} object in a \codestyle{deleteInfo} field in the node).

We do not modify the \contains{} operation to use two different paths; it always runs in a ``slow'' mode, following the design of the \spsize{} algorithm.  

We did explore building slow and fast paths for \contains{}, but an evaluation indicated no performance enhancement.  
By exclusively employing the slow path, threads executing \contains{} can bypass the necessity to engage in handshakes, which are typically utilized for transitioning between fast and slow paths.  Consequently, this eliminates the requirement to announce their departure from the idle phase at the beginning of a \contains{} operation.  It turned out that employing a slow path for \contains{} did not exhibit a noticeable slowdown compared to using a fast path that announces non-idle status at the beginning of the operation. 

\begin{figure*}[htb]
    \centering
\begin{lstlisting}
@\underline{\textbf{class} HandshakeSizeCalculator}@:
    @\underline{HandshakeSizeCalculator()}@:@\label{code:HandshakeSizeCalculator ctor}@
        this.sizePhase = 4
        // The 3 following arrays are implicitly initialized to zeros
        this.fastMetadataCounters = new long[numThreads]
        this.metadataCounters = new long[numThreads][2]
        this.opPhase = new int[numThreads]
        this.countersSnapshot = new HandshakeCountersSnapshot()
        this.countersSnapshot.collecting.setVolatile(false)
    @\underline{fastUpdateMetadata(opKind)}@:
        tid = ThreadID.threadID.get()
        if opKind == INSERT: @\label{code:HandshakeSizeCalculator opkind check if}@
            this.fastMetadataCounters[tid].setVolatile(1+this.fastMetadataCounters[tid].getVolatile())
        else:
            this.fastMetadataCounters[tid].setVolatile(-1+this.fastMetadataCounters[tid].getVolatile())
    @\underline{setOpPhase(opPhase)}@:
        tid = ThreadID.threadID.get()
        this.opPhase[tid] = opPhase
    @\underline{setOpPhaseVolatile(opPhase)}@:
        tid = ThreadID.threadID.get()
        this.opPhase[tid].setVolatile(opPhase)
    @\underline{getSizePhase()}@:
        return this.sizePhase.getVolatile()
    @\underline{compute()}@:
        currCountersSnapshot = this.countersSnapshot.getVolatile() @\label{code: HandshakeSizeCalculator read counterSnapshot}@
        if not currCountersSnapshot.collecting.getVolatile(): @\label{code: HandshakeSizeCalculator if collecting}@
            newCountersSnapshot = new HandshakeCountersSnapshot() @\label{code:HandshakeSizeCalculator init countersSnapshots}@
            witnessedCountersSnapshot = this.countersSnapshot.compareAndExchange(currCountersSnapshot, newCountersSnapshot) @\label{code:HandshakeSizeCalculator compareAndExchange}@
            if witnessedCountersSnapshot == currCountersSnapshot: @\label{code:HandshakeSizeCalculator if compareAndExchange was successful}@
                currSizePhase = _doFirstAndSecondHandshakes() @\label{code:HandshakeSizeCalculator compute doFirstAndSecondHandshakes}@
                _collect(newCountersSnapshot) @\label{code:HandshakeSizeCalculator compute collect}@
                fastSize = this._computeFastSize()
                newCountersSnapshot.collecting.setVolatile(false) @\label{code: HandshakeSizeCalculator set collecting false}@
                c = newCountersSnapshot.computeSize(fastSize) @\label{code:HandshakeSizeCalculator compute computeSize}@
                this.sizePhase.setVolatile(currSizePhase + 2) @\label{code: HandshakeSizeCalculator compute finish handshake}@
                return c @\label{code: HandshakeSizeCalculator return size}@
            currCountersSnapshot = witnessedCountersSnapshot
        return _waitForComputing(currCountersSnapshot)
\end{lstlisting}
\caption{\codestyle{HandshakeSizeCalculator} interface methods}\label{fig:HandshakeSizeCalculator1}
\end{figure*}

\begin{figure}
\begin{lstlisting}
    @\underline{\_performHandshake(sizePhase)}@:
        for each tid:
            wait until this.opPhase[tid].getVolatile() == IDLE_PHASE or this.opPhase[tid].getVolatile() $\geq$ sizePhase @\label{perform_hanshake: wait for size phase}@
    @\underline{\_doFirstAndSecondHandshakes()}@:
        wait until (currSizePhase = sizePhase.getVolatile()) % 4 == 0 @\label{code: HandshakeSizeCalculator doFirstAndSecondHandshakes wait until}@
        this.sizePhase.setVolatile(currSizePhase+1) @\label{code: HandshakeSizeCalculator doFirstAndSecondHandshakes first inc}@ 
        _performHandshake(currSizePhase+1) @\label{code: HandshakeSizeCalculator first handshake}@
        this.sizePhase.setVolatile(currSizePhase+2) @\label{code: HandshakeSizeCalculator doFirstAndSecondHandshakes second inc}@ 
        _performHandshake(currSizePhase+2) @\label{code: HandshakeSizeCalculator second handshake}@
        return currSizePhase+2 @\label{code: HandshakeSizeCalculator doFirstAndSecondHandshakes return}@
    @\underline{\_computeFastSize()}@:
        fastSize = 0
        for each tid:
            fastSize += this.fastMetadataCounters[tid].getVolatile()
        return fastSize
    @\underline{\_waitForComputing(currCountersSnapshot)}@:
        while true:
            currentSize = currCountersSnapshot.size@\label{code: waitForComputing retrieveSize}@
            if currentSize != INVALID_SIZE:
                return currentSize
\end{lstlisting}
\caption{\codestyle{HandshakeSizeCalculator} auxiliary methods}\label{fig:HandshakeSizeCalculator2}
\end{figure}

For the size calculation we employ the \codestyle{HandshakeSizeCalcula\-tor} and \codestyle{HandshakeCountersSnapshot} objects, which include both methods from \codestyle{SizeCalculator} and \codestyle{CountersSnapshot} of the \spsize{} methodology (\codestyle{\_collect}, \codestyle{updateMetadata} and \codestyle{createUpdateIn\-fo} of \codestyle{SizeCalculator}, and \codestyle{add} and \codestyle{forward} of \codestyle{CountersSnapsh\-ot}) which we do not repeat here, and new methods that appear in \Cref{fig:HandshakeSizeCalculator1,fig:HandshakeSizeCalculator2,fig:HandshakeCountersSnapshot}.

In the \spsize{} methodology, a \size{} operation calculates the size by calling \codestyle{SizeCalculator.com\-pute}, which performs the calculation using a \codestyle{CountersSnapshot} object. It first obtains a \codestyle{Counters\-Snap\-shot} instance that has the value \codestyle{true} in its \codestyle{collecting} field (because otherwise the instance is associated with a \size{} operation whose linearization point has already passed). It then performs some collection process of the metadata values into that \codestyle{CountersS\-napshot} instance. At this point it sets the \codestyle{collecting} field of this instance to \codestyle{false} and then computes the size using this instance. The moment of setting \codestyle{collecting} to \codestyle{false} is the linearization point of the \size{} operation. Further details appear in \cite{sela2021concurrentSize}. Our \codestyle{HandshakeSizeCalculator.compute} adds transitioning from the slow path to the fast path and back as well as handling the fast metadata, as follows.   

After obtaining a \codestyle{HandshakeCountersSnapshot} instance with \codestyle{collecting=true}, if the \size{} operation was the one to install this instance in the \codestyle{HandshakeSizeCalculator}, it proceeds to initiate a sequence involving two handshakes with other threads concurrently accessing the data structure for insertion or deletion (\Cref{code:HandshakeSizeCalculator compute doFirstAndSecondHandshakes}).
In practical terms, these handshakes are facilitated through the use of two fundamental components: a \codestyle{sizePhase} field, to which \size{} operations write the phase with which they wish the other threads will synchronize, and an \codestyle{opPhase} array field, sized to accommodate all running threads, where they publish their current phase for \size{} operations to inspect.

When a thread is not actively engaged in either an \ins{} or a \del{} operation on the data structure, it is in an \codestyle{IDLE\_PHASE} status. This state essentially informs the \size{}-performing threads that they can disregard the phase of the thread in question.
Conversely, during an \ins{} or a \del{} operation, a thread indicates its activity to the \size{}-performing threads by setting its corresponding cell in the \codestyle{opPhase} array to the appropriate value: If it identifies no concurrent \size{} operation (according to a \codestyle{sizePhase} value that reflects no ongoing \size{} operations), it takes the fast path after setting its cell to \codestyle{FAST\_PHASE} using a volatile write (to make sure \size{} operations---which accordingly perform volatile reads of the \codestyle{opPhase} values---see it). Else (if it  identifies a concurrent \size{} operation), it takes the slow path after setting its cell to the phase value that was published by a concurrent \size{} operation in the \codestyle{sizePhase} field, in order to communicate its acknowledgment of the \codestyle{sizePhase}.
To identify its per-thread cells, each thread is assigned a unique, stable thread ID upon initiation, stored in the \texttt{ThreadLocal} variable \texttt{ThreadID.threadID}.
This avoids reliance on system thread IDs.
A full description of the identifier allocation mechanism appears in \Cref{subsection: Thread registration}.

Back to the \size{}-performing thread, to initiate a handshake it increments the \codestyle{sizePhase} field by 1.
Subsequently, it awaits the synchronization of all other threads with this \codestyle{sizePhase} by inspecting their respective cells in the \codestyle{opPhase} array.
Synchronization is achieved when the value in each thread's cell is either equal to \codestyle{IDLE\_PHASE} or it is greater than or equal to \codestyle{sizePhase}.

The \size{}-performing thread carries out two such handshakes one after another.

After all threads have successfully synchronized with it in the second handshake, the \size{}-performing thread can safely perform the size computation: It first collects slow metadata values into the obtained \codestyle{HandshakeCountersSnapshot} instance according to the \spsize{} scheme by calling \codestyle{\_collect}, then sums the fast metadata values by calling \codestyle{\_computeFastSize}. At this point it sets the \codestyle{collecting} field to \codestyle{false} (which constitutes the linearization point of the operation similarly to the \spsize{} method), and then completes the size calculation---computes the size derived from slow operations (by summing slow metadata snapshot values following the snapshot mechanism of the \spsize{} methodology) and adds it to the previously-computed size derived from fast operations.

Upon completion of the computation process, the \size{}-perform\-ing thread increments the \codestyle{sizePhase} field by 2, to inform the other threads that the \size{} operation has finished, and they can return to a fast path of execution in their following operations.\footnote{
The decision to increment the {\fontsize{8}{10}\ttfamily sizePhase} by 2 is a straightforward yet practical optimization. It allows the determination of the {\fontsize{8}{10}\ttfamily size} operation phase using ({\fontsize{8}{10}\ttfamily sizePhase} mod 4), instead of ({\fontsize{8}{10}\ttfamily sizePhase} mod 3). Computing the remainder of a number after division by 4 is highly efficient in hardware, involving a bitwise {\fontsize{8}{10}\ttfamily and} operation with the constant 3.}

When a \size{} operation observes another ongoing \size{} operation, it does not follow the execution scheme described thus far for \size{} operations; instead, it calls the \codestyle{\_waitForComputing} method to wait for the other \size{} to complete its computation and adopt its computed size value once it becomes available.

The handshake-based methodology preserves the original prog\-ress guarantees of the \ins{}, \del{}, and \contains{} operations, as it adds only a constant number of non-blocking instructions to their execution paths.
In contrast to the method of~\cite{sela2021concurrentSize}, however, the \size{} operation in the handshake-based approach is not wait-free. It involves blocking wait instructions during the handshakes.


\subsubsection{Optimization: size operations join the previous handshake}\label{subsubsection: handshakes optimization}
We could further improve the performance of \size{} operations in this methodology by allowing concurrent \size{} operations to join the previous handshake when possible. This improvement is possible when the previous \size{} operation has already set the collecting field of its \codestyle{HandshakeCountersSnapshot} object to false in \Cref{code: HandshakeSizeCalculator set collecting false} and has yet to exit the second handshake by incrementing the \codestyle{sizePhase} value by $2$ in \Cref{code: HandshakeSizeCalculator compute finish handshake}. In that case, a later concurrent \size{} that has successfully installed its \codestyle{HandshakeCountersSnapsh\-ot} object in \Cref{code:HandshakeSizeCalculator compareAndExchange} and has observed a \codestyle{sizePhase} value $\equiv_{4}{2}$ in \Cref{code: HandshakeSizeCalculator doFirstAndSecondHandshakes wait until} can attempt to increase the \codestyle{sizePhase} by $4$ using a \codestyle{compareAn\-dSet} operation. This modification to the \codestyle{sizePhase} value together with the altering of the write in  \Cref{code: HandshakeSizeCalculator compute finish handshake} to use a \codestyle{compareAn\-dSet} operation with an expected value of \codestyle{currSizePhase} ensures that if the \size{} that wants to join the previous handshake successfully modifies the \codestyle{sizePhase} then the previous \size{} (who was yet to execute \Cref{code: HandshakeSizeCalculator compute finish handshake}) will not modify \codestyle{sizePhase}'s value again. By successfully adding $4$ to the \codestyle{sizePhase} value, the \size{} that wants to join the previous handshake maintains a \codestyle{sizePhase} $\equiv_{4}{2}$, thus maintaining correctness as the only \ins{} and \del{} operations that may run are ones on the slow path which did not run concurrently with fast operations. This optimization will be implemented such that instead of executing \Cref{code: HandshakeSizeCalculator doFirstAndSecondHandshakes wait until}, a \size{} operation will obtain \codestyle{sizePhase} only once and then check if the obtained value is $\equiv_{4}{2}$. If the condition is not met, it will proceed as before, performing two handshakes one after the other. Otherwise, it will try to write \codestyle{sizePhase}$+4$ to \codestyle{sizePhase} using a \codestyle{compareAndSet} call with an expected value of \codestyle{sizePhase}. If successful, there is now no need for two handshakes and it can proceed directly to \Cref{code:HandshakeSizeCalculator compute collect}. If the \codestyle{compareAndSet} call fails, then the previous \size{} managed  to execute \Cref{code: HandshakeSizeCalculator compute finish handshake} and once again, \size{} has to perform two handshakes from the beginning as before. Along with these modifications, the value of the variable \codestyle{currSizePhase} in the \codestyle{\_doFirstAndSecond\-Handshakes} function should be maintained to uphold the current value of the \codestyle{sizePhase} field.

\subsubsection{Two handshake rationale}\label{sec-2-handshakes}

In the slow path of our handshake-based methodology, dependent operations help the (update) operations they rely on to update the metadata before executing their own operations. In contrast, operations in the fast path do not assist in updating metadata on behalf of the operations they depend on. The use of handshakes allows a slow path operation to execute concurrently with a fast path operation. In such cases, the slow path operation may linearize first, with the fast path operation depending on it. Since the fast operation does not help the slow one, this could lead to a situation where the metadata is updated for the second (fast) operation before being updated for the first (slow) operation.

With a single handshake, a concurrent \size{} operation might then execute concurrently with the slow path operation and linearize before the slow path updates the metadata. In this scenario, the \size{} operation might only account for the fast (second) operation while missing the (first) slow operation it depends on. 
This scenario is demonstrated in \Cref{fig:concurrent-execution-with-single-handshake-example}, where an \ins{} operation acknowledges the handshake initiated by the \size{} operation and starts running in slow mode, while a concurrent \del{} operation that started operating in fast mode before the beginning of the handshake is still running.


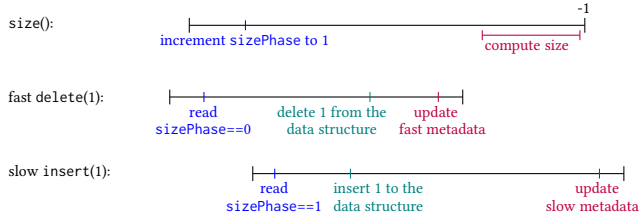
\begin{figure}[h]
	\begingroup
	\renewcommand{\codestyle}[1]{{\fontsize{6}{8.5}\selectfont\texttt{#1}}} 
	\fontsize{6}{9.5}\selectfont
	\setlength{\unitlength}{0.065mm}
	\begin{picture}(800,230)
		\put(-250,150){\size{}():}
		\put(-250,0){fast \del(1):}
		\put(-250,-150){slow \ins(1):}
		
		\put(120,170){\line(0,-1){30}} 
		\put(120,155){\line(1,0){810}} 
		\put(930,170){\line(0,-1){30}} 
		\put(915,180){-1}
		\put(235,165){\color{blue} \line(0,-1){20}} 
		\put(60,115){\color{blue} increment \codestyle{sizePhase} to 1}
		\put(720,145){\color{purple} \line(0,-1){20}}
		\put(720,135){\color{purple} \line(1,0){200}}
		\put(920,145){\color{purple} \line(0,-1){20}}
		\put(725,105){\color{purple} compute size}
		
		\put(80,25){\line(0,-1){30}} 
		\put(80,10){\line(1,0){600}} 
		\put(680,25){\line(0,-1){30}} 
		\put(150,20){\color{blue} \line(0,-1){20}}
		\put(120,-30){\color{blue} read}
		\put(50,-65){\color{blue} \codestyle{sizePhase}==0}
		\put(490,20){\color{teal} \line(0,-1){20}} 
		\put(300,-30){\color{teal} delete 1 from the}
		\put(320,-65){\color{teal} data structure}
		\put(630,20){\color{purple} \line(0,-1){20}} 
		\put(580,-30){\color{purple} update}
		\put(550,-65){\color{purple} fast metadata}
		
		\put(250,-130){\line(0,-1){30}} 
		\put(250,-145){\line(1,0){760}} 
		\put(1010,-130){\line(0,-1){30}} 
		\put(450,-135){\color{teal} \line(0,-1){20}} 
		\put(415,-185){\color{teal} insert 1 to the}
		\put(415,-220){\color{teal} data structure}
		\put(960,-135){\color{purple} \line(0,-1){20}}
		\put(910,-185){\color{purple} update}
		\put(850,-220){\color{purple} slow metadata}
		\put(295,-135){\color{blue} \line(0,-1){20}}
		\put(265,-185){\color{blue} read}
		\put(195,-220){\color{blue} \codestyle{sizePhase}==1} 
		
	\end{picture}
	\vspace{40pt}
	\caption{An execution with a single handshake in which the size computation is concurrent with a slow \texttt{insert} that ran concurrently with a fast dependent \texttt{delete}.}
	\label{fig:concurrent-execution-with-single-handshake-example}
	\endgroup
\end{figure}

To ensure linearizability, we ensure that when a \size{} operation computes the size, all concurrent data structure operations have their metadata adequately updated in an order that respects operation dependencies. This is achieved by having the \size{} operation execute concurrently only with slow path operations that have not previously executed concurrently with fast path operations.

After the first handshake completes, we know that all threads are operating in the slow path. However, at this point, some slow path operations may have already executed concurrently with fast path operations. Therefore, a second handshake is performed to wait for all threads to complete their previous operations. Once the second handshake completes, we know that any currently executing operation is in the slow path and has never executed concurrently with a fast path operation.

\subsection{Optimistic Approach}\label{section:optimistic}

We now move to presenting an optimistic scheme for evaluating the size of a data structure.
The main idea of this approach is to allow the \size{} operations to execute optimistically assuming that there are no concurrent updates to the data structure, with no locks and minimal interference to other operations. 
Each thread maintains one local counter that represents the number of \ins{}s it has executed minus the number of \del{}s. This counter effectively captures the impact of this thread on the size of the data structure. 
The objective of the optimistic synchronization is to read all these local counters and collectively sum them up at a time when no concurrent update is executing. Achieving this enables the computation of a linearizable size.

To support the optimistic \size{} operation, the other threads cooperate in order to allow detecting whether the \size{} operation executes without any concurrent update operations. Each thread maintains a local flag indicating whether it is presently involved in updating the data structure, and a local counter of the overall number of update operations the thread has executed. This information is compactly stored within a single word, denoted the activity counter. 
When a thread initiates an update operation, it increments its activity counter. Upon completion of the update operation, the thread increments its activity counter again. Therefore, an odd number in the activity counter indicates that the thread is currently performing an update operation, while the value of the activity counter divided by 2 reveals the count of overall number of updates executed by that thread.

A \size{} operation begins by reading the activity counters of all threads (non-atomically) to ensure that none are currently engaged in updating the data structure—that is, it verifies that all activity counters are even. If any activity counter is found to be odd, indicating that the corresponding thread is actively modifying the data structure, the \size{} operation waits until it becomes even. Once all activity counters are read as even, indicating no ongoing updates, the \size{} operation proceeds to read the metadata counters of all threads (non-atomically), sums them, and then re-reads the activity counters (non-atomically) to confirm that they have not changed. If any change is detected, the operation restarts. The fact that all activity counters were even and remained unchanged ensures that no concurrent update occurred during the metadata counter reads. Consequently, the sum of the metadata counters yields a linearizable size computation.

The issue with the method described above, in its raw form, is the potential for it to restart endlessly if the second read of the activity counters does not match the first. To address this, after multiple restarts, the \size{} operation sets a flag to temporarily suspend update operations and requests that the updating threads assist in computing the size. Once the size computation is complete, all operations resume. Although this approach effectively handles exceptional cases, it may degrade performance, particularly under high contention when update operations are frequent.


\subsubsection{Data-structure transformation}\label{optimistic: algorithm in detail}
The data structure transformation for the optimistic approach uses an \codestyle{OptimisticSizeCalcu\-lator} object whose methods appear in \Cref{fig:OptimisticSizeCalculator1,fig:OptimisticSizeCalculator2} to calculate the size. Next we bring the transformation details (the full pseudocode appears in \Cref{fig: transformed data structure with an optimistic scheme}).
Like in the other methodologies, an array named \codestyle{metadataCoun\-ters} with per-thread size metadata is maintained and updated upon a successful \ins{} or \del{} operation.
The time gap between updating the data structure and updating the \size{} metadata in \ins{} or \del{} operations can lead to non-linearizable size results for \size{} operations that observe the \size{} metadata during this period.
To prevent this, we maintain an activity counter per thread in an array named \codestyle{activityCounters}. Each thread performing an \ins{} or \del{} increments its cell in the \codestyle{activityCounters} array before making any changes to the data structure, and increments it again after updating the \size{} metadata regardless of whether the operation was successful or not. Using this activity counter array, a \size{} operation can determine whether the metadata was updated during its execution, and if so, it can retry the operation, as follows.

\begin{figure}

\begin{lstlisting}
@\underline{\textbf{class} OptimisticSizeCalculator}@:
    @\underline{OptimisticSizeCalculator()}@:@\label{code:OptimisticSizeCalculator ctor}@
        MAX_TRIES = 3
        this.metadataCounters = new long[numThreads]
        this.activityCounters = new long[numThreads]
        this.awaitingSizes = 0
        this.sizeInfo = new SizeInfo()
    @\underline{incrementActivityCounter()}@:
        tid = ThreadID.threadID.get()
        this.activityCounters[tid].setVolatile(1+this.activityCounters[tid].getVolatile())
    @\underline{helpSize()}@:
        if this.awaitingSizes.getVolatile() == 0: return@\label{code:check awaitingSizes}@
        currentSizeInfo = this.sizeInfo.getVolatile()
        while true:
            if currentSizeInfo.size.getVolatile() != INVALID_SIZE:
                break
            size = _tryComputeSize()
            if size != INVALID_SIZE:
                activeSizeInfo.size.compareAndSet(INVALID_SIZE, size)
                break
    @\underline{updateMetadata(opKind)}@:
        tid = ThreadID.threadID.get()
        if opKind == INSERT: @\label{code:OptimisticSizeCalculator opkind check if}@
            this.metadataCounters[tid].setVolatile(1+this.metadataCounters[tid].getVolatile())
        else:
            this.metadataCounters[tid].setVolatile(-1+this.metadataCounters[tid].getVolatile())
    @\underline{computeSize()}@:
        count = 0
        <activeSizeInfo, isReturnableSizeInfo> = _obtainActiveSizeInfo()
        while true:
            if (size = activeSizeInfo.size.getVolatile()) != INVALID_SIZE:
                if isReturnableSizeInfo: break
                else:
                    <activeSizeInfo, _> = _obtainActiveSizeInfo()
                    isReturnableSizeInfo = true
            if count == MAX_TRIES:
                this.awaitingSizes.getAndAdd(1)@\label{code:awaitingSizes inc}@
            if count <= MAX_TRIES:
                count++
            size = _tryComputeSize()
            if size != INVALID_SIZE:
                activeSizeInfo.size.compareAndSet(INVALID_SIZE, size)
                break
        if count == MAX_TRIES + 1:
            this.awaitingSizes.getAndAdd(-1)@\label{code:awaitingSizes dec}@
        return size
\end{lstlisting}
\caption{\codestyle{OptimisticSizeCalculator} interface methods}\label{fig:OptimisticSizeCalculator1}
\end{figure}

\begin{figure}
\begin{lstlisting}
		@\underline{\_readActivityCounters()}@:
			status = new long[numThreads]
			for each tid:
				wait until ((status[tid] = this.activityCounters[tid].getVolatile())%2 == 0) @\label{code: optimistic wait for even numbers}@
			return status
		@\underline{\_retryActivityCounters(status)}@:
			for each tid:
				if status[tid] != this.activityCounters[tid].getVolatile():
					return false
			return true
	    @\underline{\_tryComputeSize()}@:
	        status = _readActivityCounters()@\label{code:read activityCounters first time}@
	        sum = 0@\label{code:sum counters start}@
	        for each tid:
            	sum += this.metadataCounters[tid].getVolatile()@\label{code:sum counters end}@
        	if _retryActivityCounters(status):@\label{code:read activityCounters second time}@
            	return sum
        	return INVALID_SIZE
		@\underline{\_obtainActiveSizeInfo()}@:
			currentSizeInfo = this.sizeInfo.getVolatile()
			if currentSizeInfo.size.getVolatile() == INVALID_SIZE:
				activeSizeInfo = currentSizeInfo
				isNewlyInstalledSizeInfo = false
			else:
   				isNewlyInstalledSizeInfo = true
   				newSizeInfo = new SizeInfo()
				witnessedSizeInfo = this.sizeInfo.compareAndExchange(currentSizeInfo, newSizeInfo) @\label{code: install size object}@
				if witnessedSizeInfo == currentSizeInfo:
					activeSizeInfo = newSizeInfo
				else:
					activeSizeInfo = witnessedSizeInfo
			return <activeSizeInfo, isNewlyInstalledSizeInfo>
    \end{lstlisting}
    \caption{\codestyle{OptimisticSizeCalculator} auxiliary methods}\label{fig:OptimisticSizeCalculator2}
\end{figure}

To calculate the size, a \size{} operation calls the \codestyle{\_tryComputeSize} method, which starts by making a copy named \codestyle{status} of the \codestyle{activ\-ityCounters} array (\Cref{code:read activityCounters first time}). It is important to note that this copy is not obtained using a snapshot mechanism. For any obtained odd value, which means that the corresponding thread is executing an \ins{} or a \del{} operation, the cell in the \codestyle{status} array is re-read until obtaining an even activity counter value. Once obtaining a \codestyle{status} array with no odd values, the \size{} operation proceeds to calculating the size by summing up the values in the \codestyle{metadataCounters} array (\Cref{code:sum counters start,code:sum counters end}).
Then the \size{} operation accesses the \codestyle{activityCounters} array again and compares its values with the values previously obtained in the \codestyle{status} array (\Cref{code:read activityCounters second time}). If they do not match, it restarts. Otherwise, the \size{} operation finishes and returns the computed size.

To prevent the \size{} operation from restarting indefinitely, we set a limit named \codestyle{MAX\_TRIES} on its number of retries, which determines the maximal number of attempts the \size{} will go through before making concurrent \ins{} and \del{} operations assist it. Once this limit is reached, the \size{} operation increments a counter called \codestyle{awaitingSizes} (\Cref{code:awaitingSizes inc}), which it will later decrement before it returns (\Cref{code:awaitingSizes dec}).
The \ins{} and \del{} operations check this counter before they start operating. In case its value is positive, they help the \size{} operation by trying to compute the size themselves in a similar fashion to \size{} operations---by obtaining \codestyle{activityCounters} values before and after the computation (see the \codestyle{helpSize} method in \Cref{fig:OptimisticSizeCalculator1}).
The \codestyle{MAX\_TRIES} variable has a big effect on the transformed data structure's performance. If it is too small, \ins{} and \del{} operations may be interrupted frequently by \size{} operations requiring them to help before performing their operation and therefore harming their performance. If it is too big, \size{} operations may take a long time to complete, deteriorating the performance of \size{} operations.

Helping a \size{} operation compute the size (both by \ins{} and \del{} operations and by other \size{} operations) is coordinated using a shared object named \codestyle{SizeInfo}, which has a single field named \codestyle{size} initialized to \codestyle{INVALID\_SIZE} and intended to hold the result of a \size{} operation. \size{} operations install such an instance in \codestyle{OptimisticSizeCalculator.sizeInfo}, and concurrent \size{}, \ins{} and \del{} operations that observe an installed instance with \codestyle{INVALID\_SIZE} size value attempt to compute the size and write the obtained size onto the size field.
The reason a \size{} operation needs to obtain a \codestyle{SizeInfo} instance, installed in \codestyle{OptimisticSizeCalculator.sizeInfo} by itself or by a concurrent \size{}, is to be able to retrieve from it a size value computed by another thread, as the \size{} operation might keep failing to obtain two identical copies of even activity counters and compute a correct size on its own. After obtaining a \codestyle{SizeInfo} instance, the \size{} operation keeps attempting to obtain two such activity counters copies and compute the size in between. On a successful attempt, it returns the computed value while also writing it to the obtained \codestyle{SizeInfo} instance for helping others. On a failing attempt, if another thread succeeded and wrote its computed size to the \codestyle{SizeInfo} instance, it returns this computed size.
However, a size written by another thread to the first \codestyle{SizeInfo} instance obtained by \size{} may not be returned, since it might have been computed (by summing the \codestyle{countersMetadata} values) before this \size{}'s interval (and so the \size{} operation would have been linearized outside its interval).
Once observing that the \size{} field in the first obtained \codestyle{SizeInfo} instance is set, a new instance should be installed and obtained, and the size value---that will be later computed and written to it---may be legally returned.

The optimistic methodology does not maintain the progress guarantees of \ins{} and \del{} due to the blocking wait in \Cref{code: optimistic wait for even numbers}. It does maintain them for the \contains{} operation as it does not modify it. 

\subsection{Locks} \label{section:locks}
In this section, we describe the final synchronization method studied: lock-based synchronization. As detailed in \Cref{section:preliminary}, locks provide mutual exclusion by allowing only one thread to hold the lock at any given time. However, simply blocking all update operations to let only one execute at a time could be detrimental to scalability and performance. To address this, we employ advanced locks known as read-write locks, which enable multiple reader threads to execute concurrently, whereas a writer thread executes alone, preventing any other thread (reader or writer) from acquiring the lock concurrently with a writer thread that holds the writer lock.

Using these locks, all update operations acquire the reader lock, enabling them to execute concurrently. The \size{} operation acquires the writer lock, ensuring it executes alone without concurrent updates. To facilitate quick execution of \size{} and clear the path quickly for other operations, the updating threads maintain their local metadata counters. The metadata allows the \size{} operation to quickly execute, as it only needs to read the metadata counters of all threads (rather than traversing the whole data structure to count elements),
and it can then release the write lock, allowing all threads to resume execution. The \size{} operation returns the sum of all metadata counters. Moreover, concurrent \size{} executions cooperate, allowing one execution to perform the size computation and others to utilize the result of this computation.

For lack of space, the details of this approach along with pseudo-code are relegated to \Cref{app-locks-details}.

\section{Evaluation}\label{section:evaluation}
We implemented all of the presented methodologies for computing size in Java, closely corresponding to the algorithms and pseudocode described across \Cref{section:handshakes,section:optimistic,section:locks}. This implementation includes all described optimizations in \Cref{subsubsection: handshakes optimization,subsection: general optimizations}. The code for the data structures and the measurements is available at \cite{artifactSyncMethodsForConcurrentSize}.
In this section, we present the evaluation of all methodologies compared to the methodology from \cite{sela2021concurrentSize}. Two primary aspects were chosen for testing: (1) the additional overhead each methodology incurs on the operations of the original data structure and (2) the performance of the \size{} operation, evaluated by testing its scalability. Finally, we try to give recommendations on which methodology should be used in each scenario.

\paragraph{Platform.} 
All experiments were executed on a system operating on Linux (Ubuntu 20.04.5 LTS), powered by two Intel(R) Xeon(R) Gold 6338 CPUs @2.00GHz, each with 64 threads, summing up to a total of 128 threads. The system is equipped with 32GB of RAM. The methodologies were implemented using Java, employing OpenJDK version 21. We used the @Contended annotation on some shared fields to reduce false sharing, and enabled it with -XX:-RestrictContended -XX:ContendedPaddingWidth=64. As in \cite{sela2021concurrentSize}, the G1 garbage collector was deployed and the flags -server, and -Xms31G were utilized to improve performance and minimize disruption of Java's garbage collection.

\paragraph{Data structures.} We evaluated the methodologies on three different data structures: SkipList, Binary Search Tree (BST) and Hash\-Table. The implementations for the baseline data structures are taken from the public implementation of \cite{sela2021concurrentSize}, available in~\cite{artifactConcurrentSize}. These implementations in turn are based on prior work. The SkipList builds on Java's \texttt{ConcurrentSkipListMap} from the \texttt{java.util.con\-current} package in Java SE 18. The BST builds on Brown's implementation \cite{brownJava} of the lock-free binary search tree of \cite{ellen2010non} that places elements in leaf nodes. The HashTable was implemented in \cite{sela2021concurrentSize} based on the linked list in the base level of Java's \texttt{Concurrent\-SkipList\-Map}.
Since the BST implementation in \cite{brownJava} does not linearize the \del{} operation at the marking step, to comply with the restrictions of the handshake-based methodology, we used a variant of BST that linearizes the deletion at the marking step as in \cite{sela2021concurrentSize}. Since the lock-based and optimistic methodologies do not pose that restriction, we used the unmodified BST implementation from \cite{brownJava} for these transformations.

\paragraph{Methodology.} For the most part, we use the same testing methodology as in \cite{sela2021concurrentSize}. This methodology involves initializing the data structure with 1M items prior to each experiment. Subsequently, two distinct workloads are executed: an update-heavy workload, comprising 30\% \ins{} operations, 20\% \del{} operations, and 50\% \contains{} operations; and a read-heavy workload, consisting of 3\% \ins{} operations, 2\% \del{} operations, and 95\% \contains{} operations. These workloads align with the recommended read rates outlined in the \textit{Yahoo! Cloud Serving Benchmark} (YCSB)~\cite{cooper2010benchmarking}. YCSB also proposes a $100\%$-read workload; however, this scenario is less pertinent to our case as the likelihood of \size{} calls on a data structure that remains unchanged is negligible.
The results that correspond to the read-heavy workload are displayed on the left side of \Cref{fig:SkipList overhead,fig:BST overhead,fig:HT overhead,fig:size scalability SL,fig:size scalability BST,fig:size scalability HT,fig:MAX_TRIES SL,fig:MAX_TRIES BST,fig:MAX_TRIES HT}, while those related to the update-heavy workload are presented on the right side.

The keys utilized for operations during the experiment as well as for the initialization of the data structure, are selected uniformly at random from a specified range $[1, r]$ like in \cite{sela2021concurrentSize}. The value of $r$ is determined to ensure the target size of the data structure is maintained. Furthermore, in all experiments, the type of the subsequent operation is determined iteratively based on the specified update-heavy or read-heavy workload proportions.
Each experiment involves the concurrent execution of $w$ workload threads, which engage in \ins{}, \del{}, and \contains{} operations according to the characteristics of the workload, alongside $s$ \size{} threads, which repeatedly invoke the \size{} operation with a delay between each two invocations. We set the delay time between \size{} operations in overhead measurements to either $0$ or $700 \si{\micro\second}$ (microseconds) to represent continuous or occasional invocations of \size{}. In the rest of the measurements we kept the delay at $0$. 
We chose $700 \si{\micro\second}$ to represent an execution of \size{} at about $10\%$ of the clock time, depending on the methodology used.
For baseline algorithms, only $w$ workload threads are employed.
The values of $w$ and $s$ vary across experiments, with the constraint that $w+s$ is predominantly chosen as a power of $2$. Experiments are run for $5$ seconds. Each reported data point in the graphical representations is the average outcome of 10 runs, following an initial warm-up phase consisting of 5 preliminary runs to stabilize the Java virtual machine. 
To reduce the variance between experiments, we disabled hyper-threading, leaving us with $64$ threads to utilize. Additionally, in experiments involving up to $32$ threads, we employed the "taskset --cpu-list 0-31" command to ensure that the entire experiment was executed on a single CPU node to reduce variability.

\subsection{Overhead of \size{}}\label{sec-overheads}
The graphs in \Cref{fig:SkipList overhead,fig:BST overhead,fig:HT overhead} present throughput measurements of the various synchronization methods with the skip list (\Cref{fig:SkipList overhead}), the BST (\Cref{fig:BST overhead}), and the hash table (\Cref{fig:HT overhead}). 
Specifically, for each data structure we evaluate the original version of the data structure, the \spsize{} method of~\cite{sela2021concurrentSize} (denoted {\em SP}), the optimistic method (denoted {\em optimistic}), the handshake method (denoted {\em handshake}), and the lock-based method (denoted {\em stampedLock}). Each figure contains three rows. The first row presents an execution with no concurrent \size{} execution, representing just the overhead for having size available. The second row presents an execution with a continuous execution of \size{} by a concurrent thread, representing the overhead for always cooperating with the execution of \size{}. Lastly the third row presents an execution of operations when the \size{} operation occasionally runs concurrently. As stated above, we let the thread running \size{} execute a delay of $700 \si{\micro\second}$ after each \size{} execution, to represent a cooperation with a \size{} method at approximately $10\%$ of the time, depending on the efficiency of \size{} with the specific synchronization method. 

\begin{figure*}[htbp]
    \centering
    \medskip
    \textit{Read heavy}\quad
    \includegraphics[height=.02\textwidth]{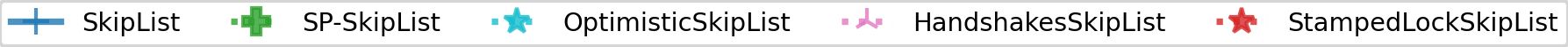}\quad
    \textit{Update heavy}\par
    \medskip
    \text{Without a concurrent \size{} thread}\par
    \smallskip
    \includegraphics[width=.45\textwidth,trim={0 0 0 .1cm}]{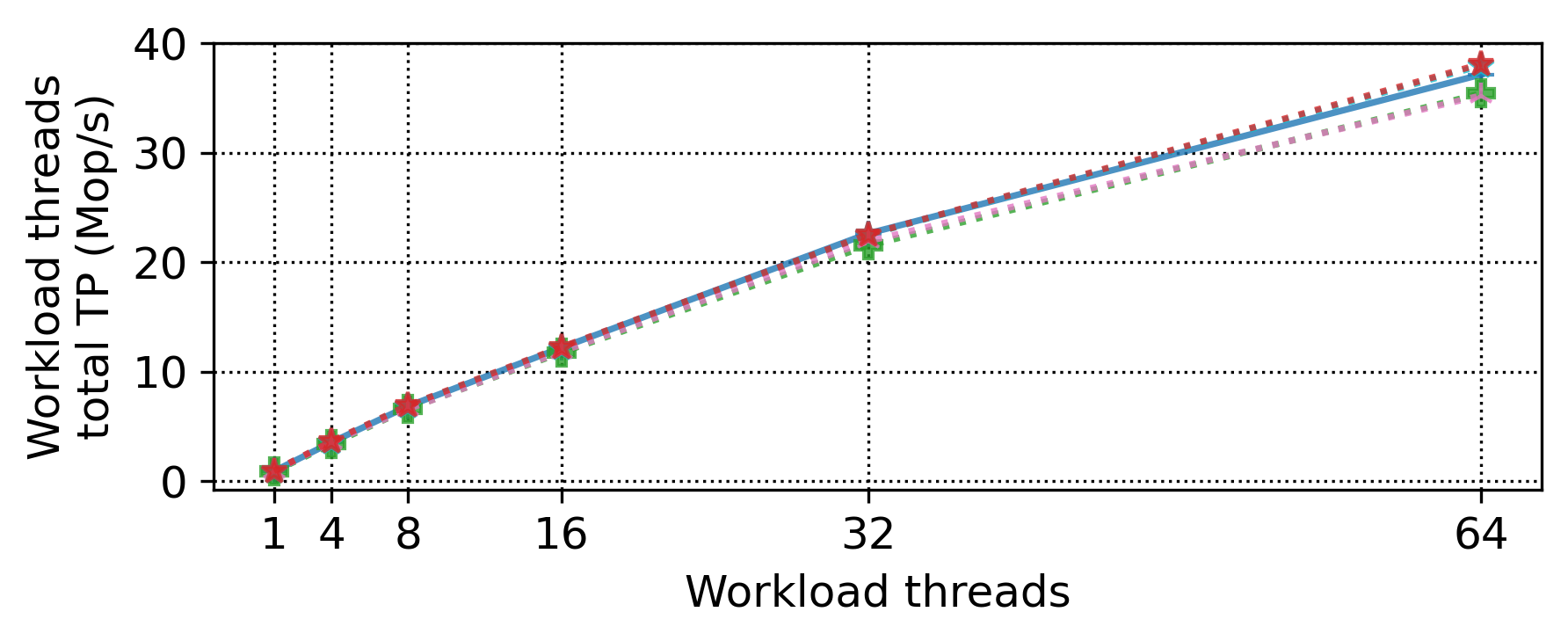}\hspace{2.5em}
    \includegraphics[width=.45\textwidth,trim={0 0 0 .1cm}]{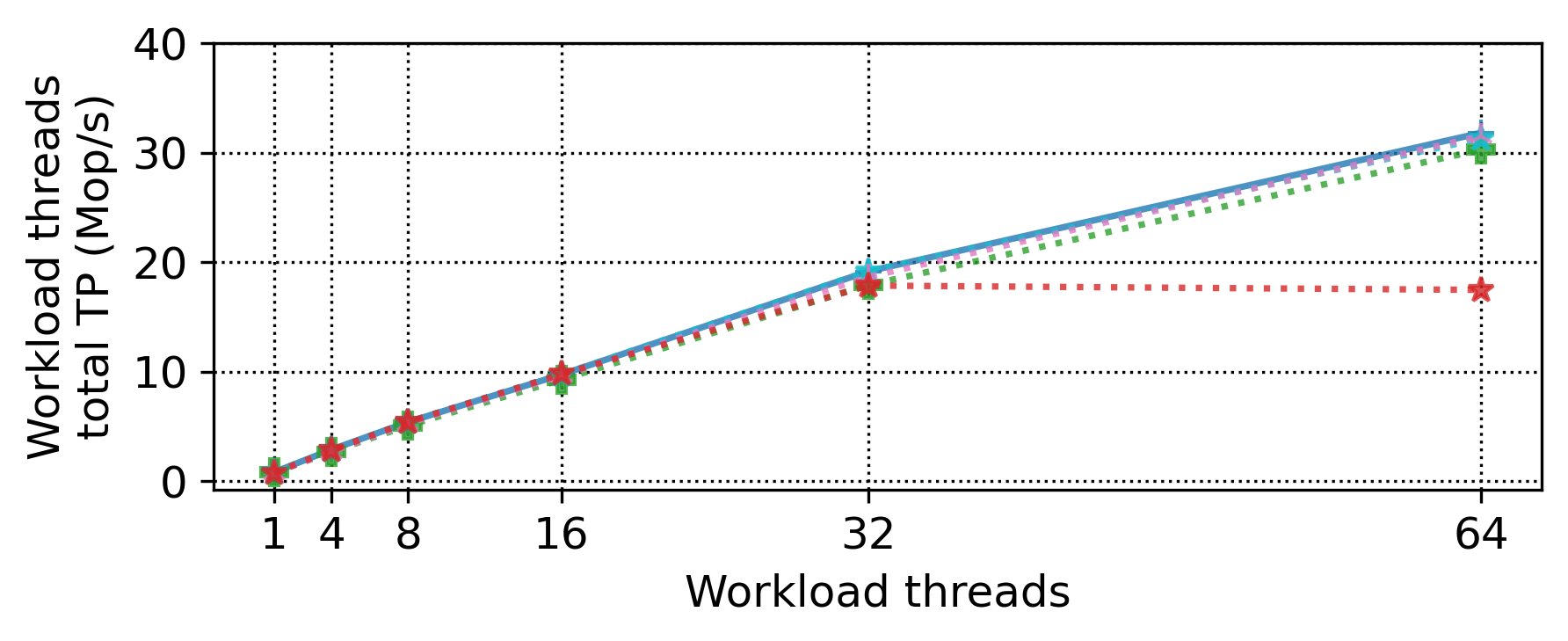}\par
    \includegraphics[width=.45\textwidth,trim={0 0 0 .2cm}]{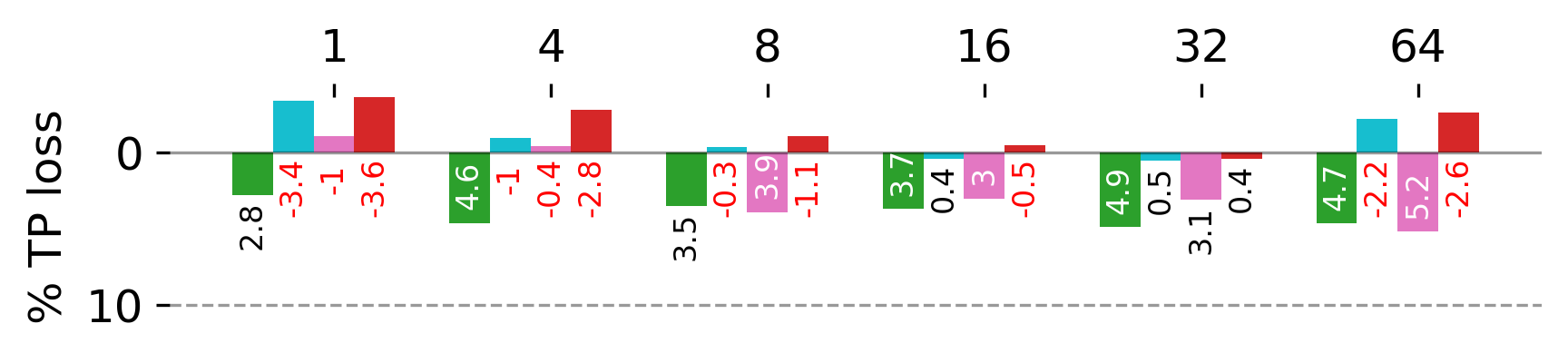}\hspace{2.5em}
    \includegraphics[width=.45\textwidth,trim={0 0 0 .2cm}]{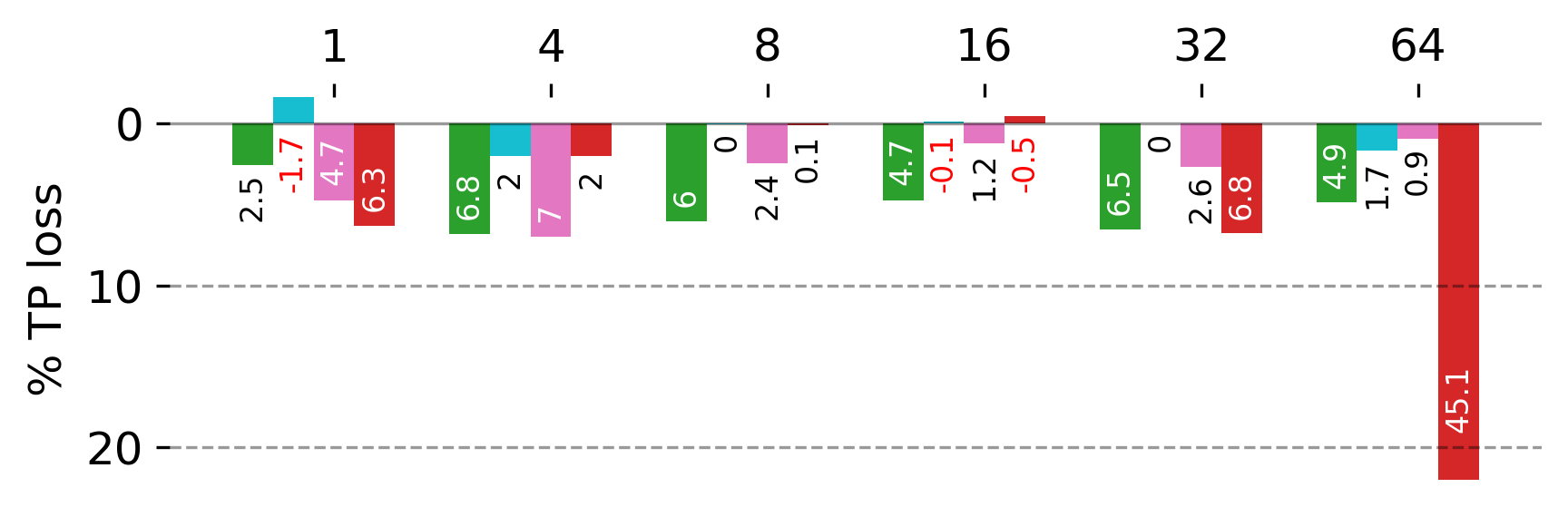}\par
    \medskip
    \text{With a concurrent \size{} thread and no delay}\par
    \includegraphics[width=.45\textwidth,trim={0 0 0 .1cm}]{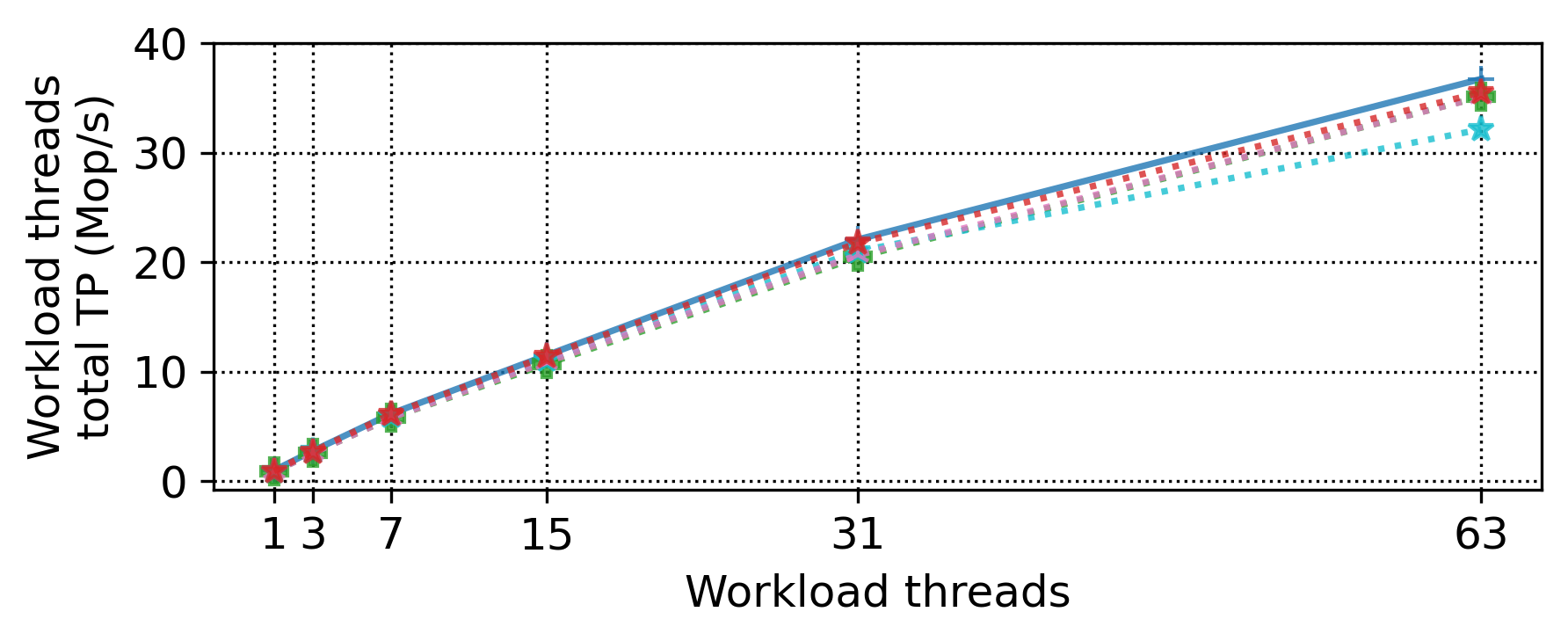}\hspace{2.5em}
    \includegraphics[width=.45\textwidth,trim={0 0 0 .1cm}]{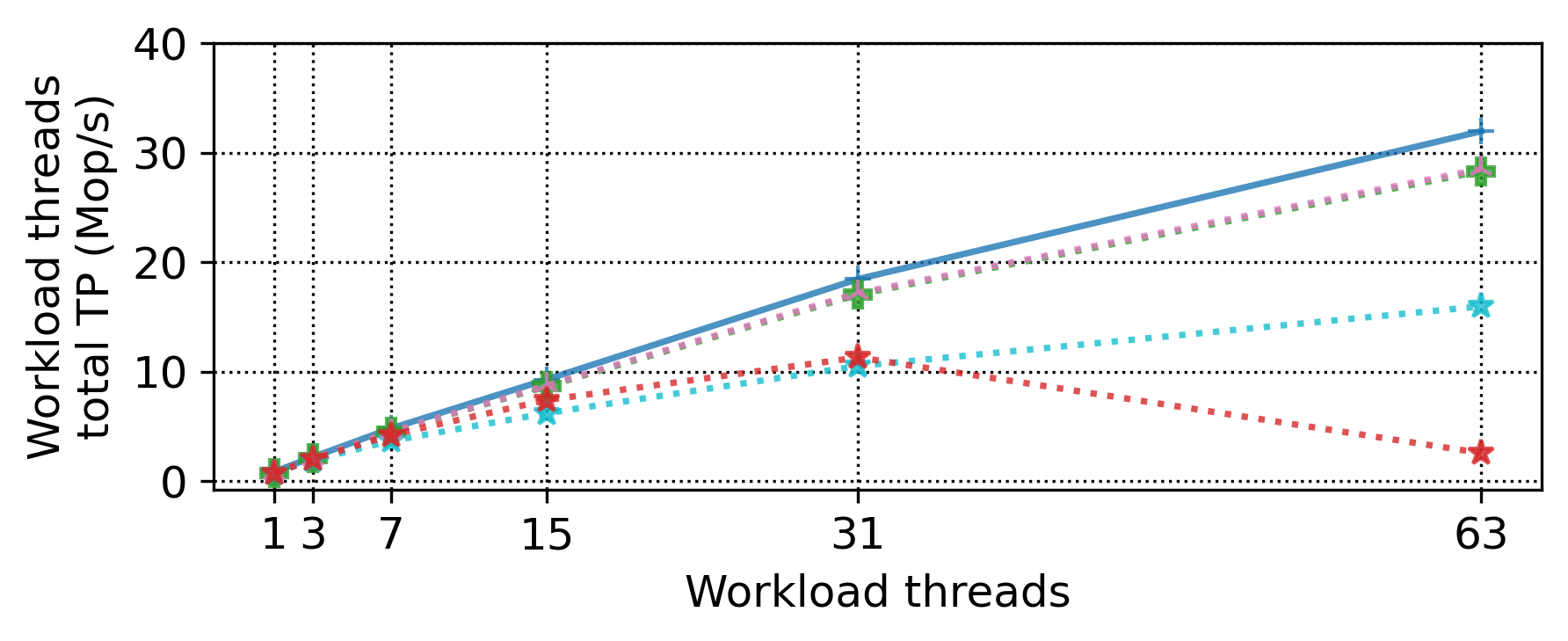}\par
    \includegraphics[width=.45\textwidth,trim={0 0 0 .2cm}]{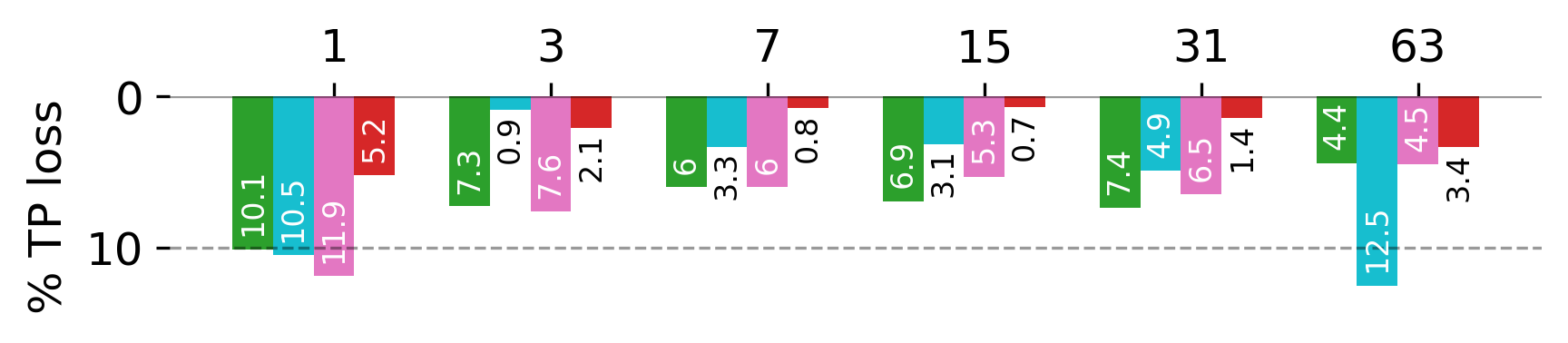}\hspace{2.5em}
    \includegraphics[width=.45\textwidth,trim={0 0 0 .2cm}]{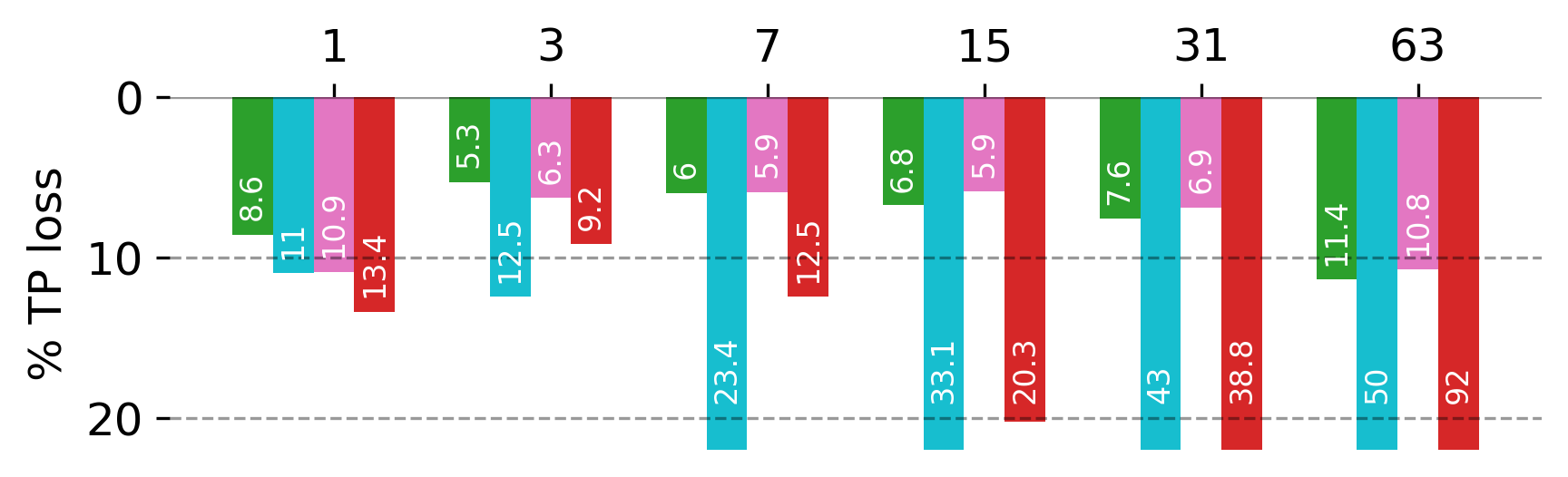}\par
    \medskip
    \text{With a concurrent \size{} thread and 700 \si{\micro\second} delay}\par
    \includegraphics[width=.45\textwidth,trim={0 0 0 .1cm}]{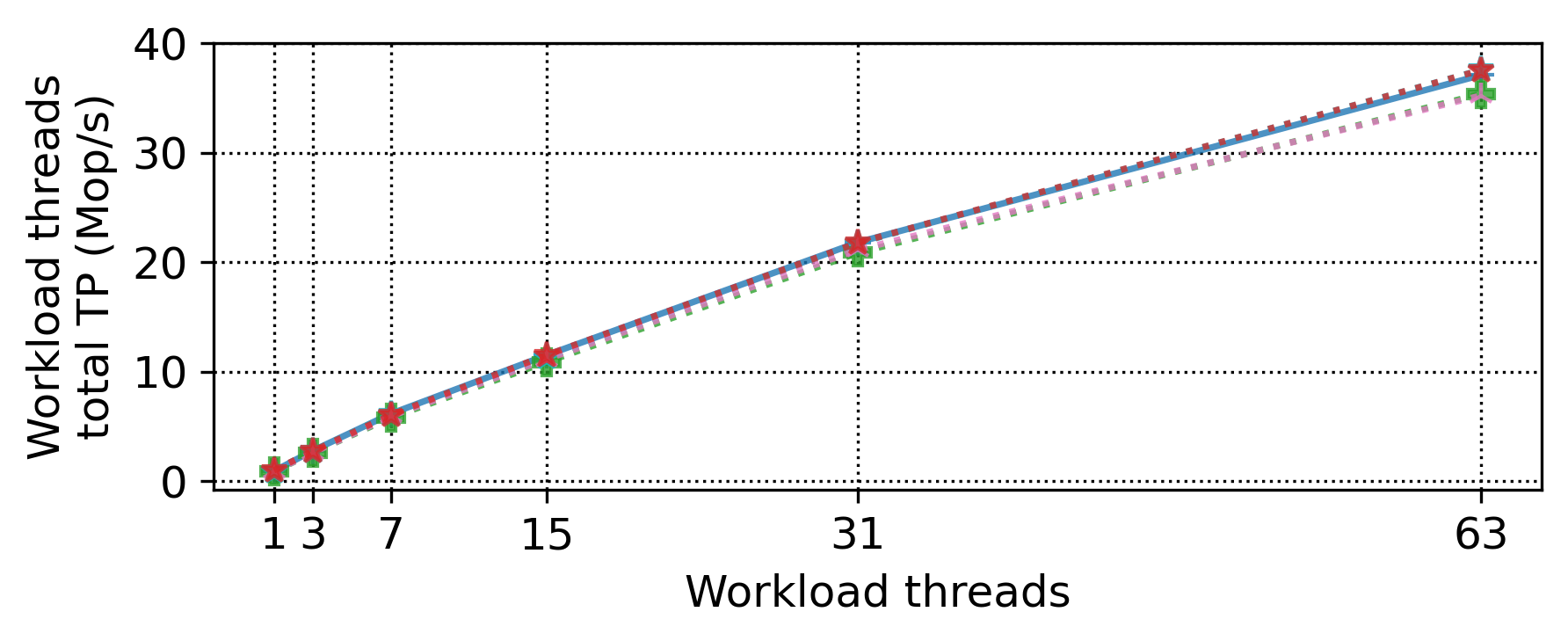}\hspace{2.5em}
    \includegraphics[width=.45\textwidth,trim={0 0 0 .1cm}]{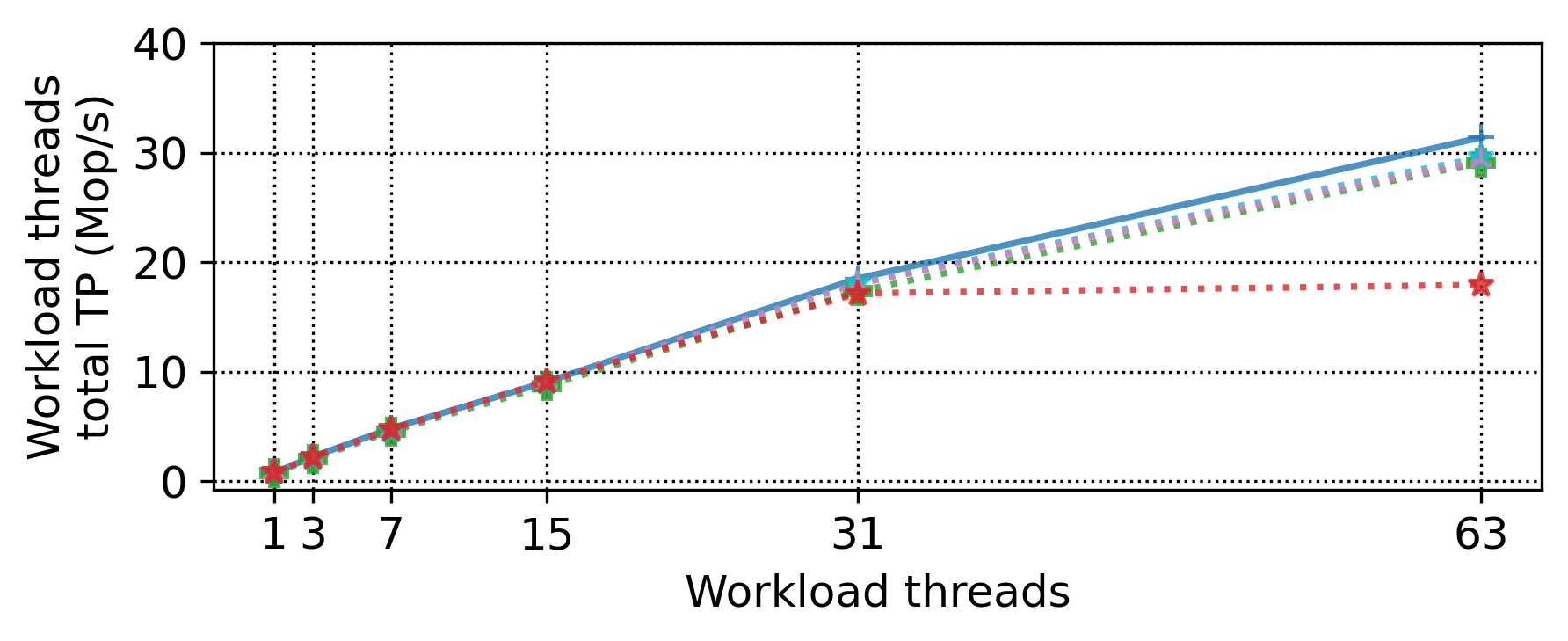}\par
    \includegraphics[width=.45\textwidth,trim={0 0 0 .2cm}]{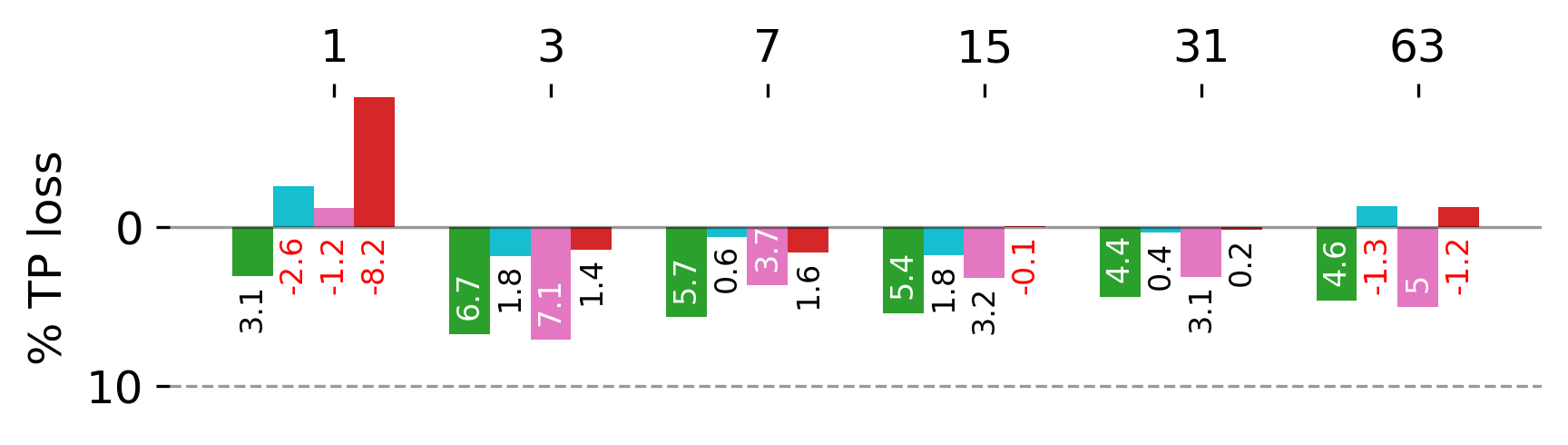}\hspace{2.5em}
    \includegraphics[width=.45\textwidth,trim={0 0 0 .2cm}]{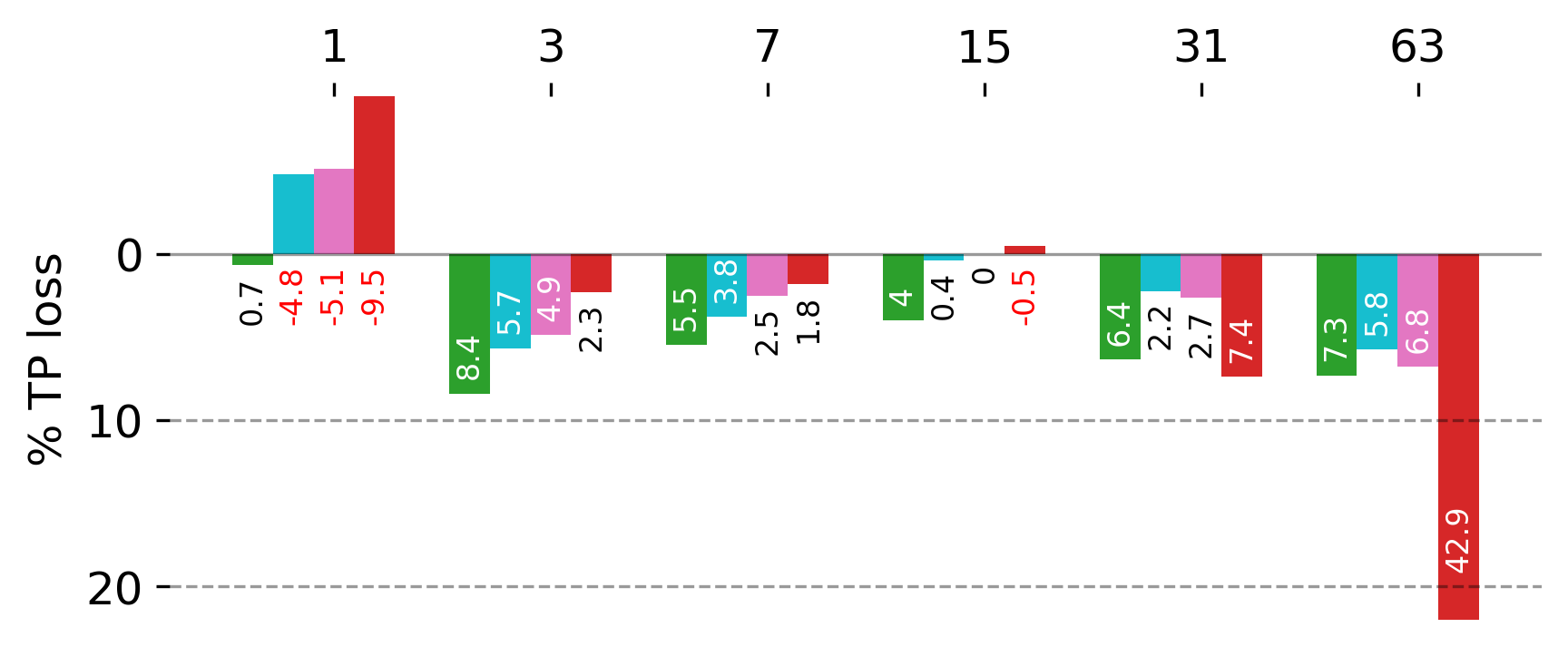}\par
    \caption{Overhead on skip list operations}
    \label{fig:SkipList overhead}
\end{figure*}

\begin{figure*}[htbp]
	\centering
	\medskip
	\textit{Read heavy}\quad\quad\quad
	\includegraphics[height=.02\textwidth]{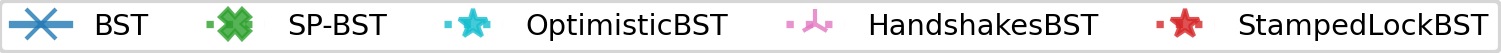}\quad\quad\quad
	\textit{Update heavy}\par
	\medskip
	\text{Without a concurrent \size{} thread}\par
        \smallskip
	\includegraphics[width=.45\textwidth,trim={0 0 0 .1cm}]{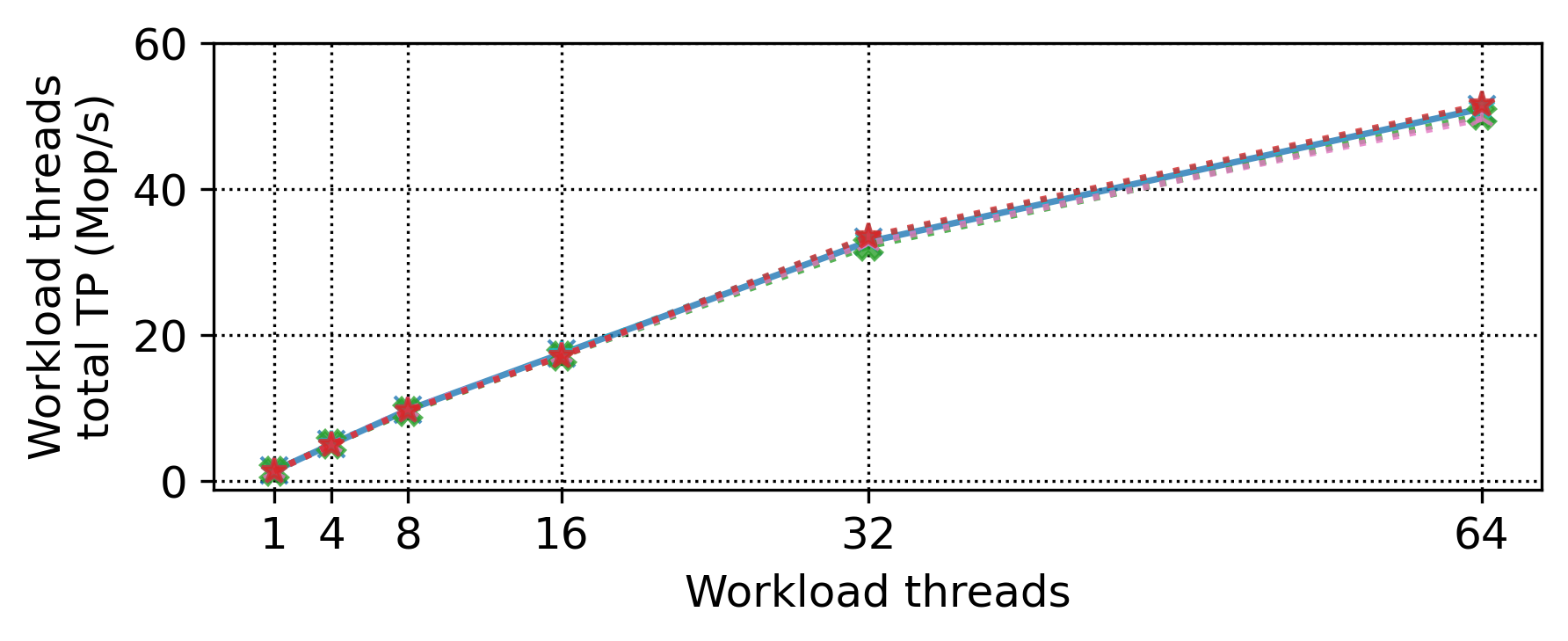}\hspace{2.5em}
	\includegraphics[width=.45\textwidth,trim={0 0 0 .1cm}]{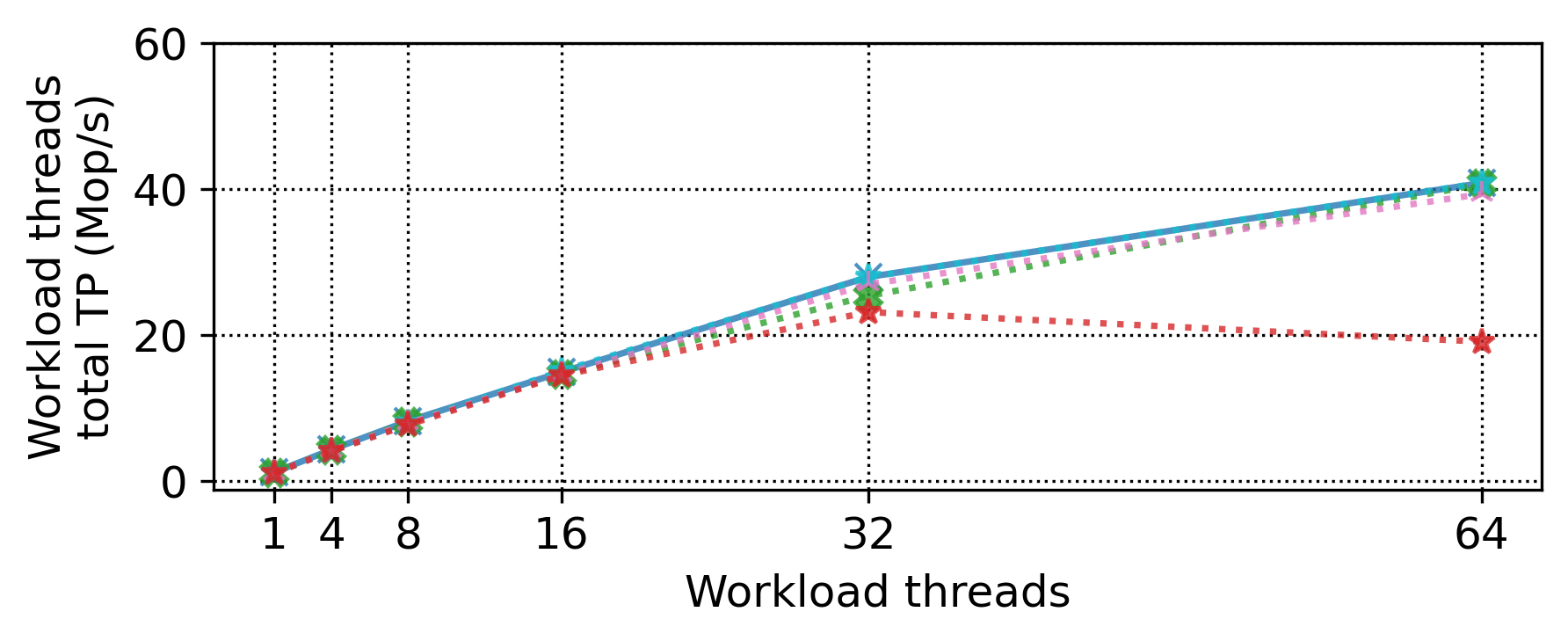}\par
	\includegraphics[width=.45\textwidth,trim={0 0 0 .2cm}]{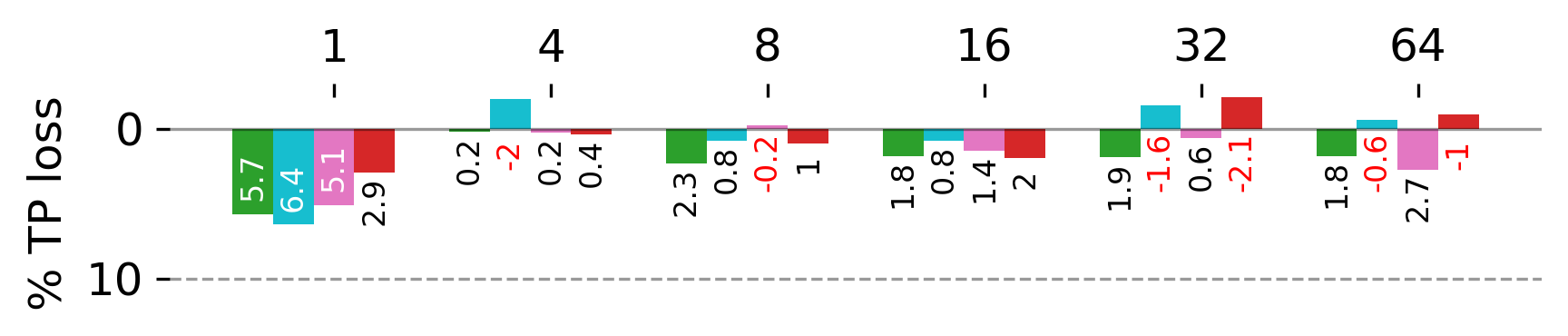}\hspace{2.5em}
	\includegraphics[width=.45\textwidth,trim={0 0 0 .2cm}]{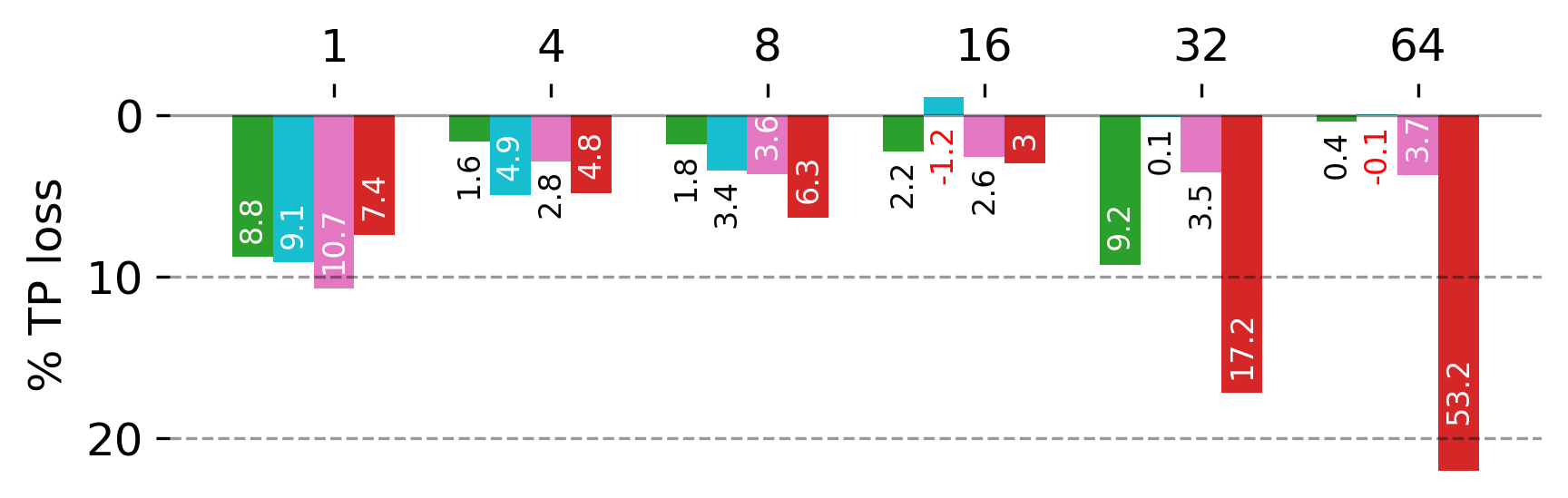}\par
	\medskip
	\text{With a concurrent \size{} thread and no delay}\par
	\includegraphics[width=.45\textwidth,trim={0 0 0 .1cm}]{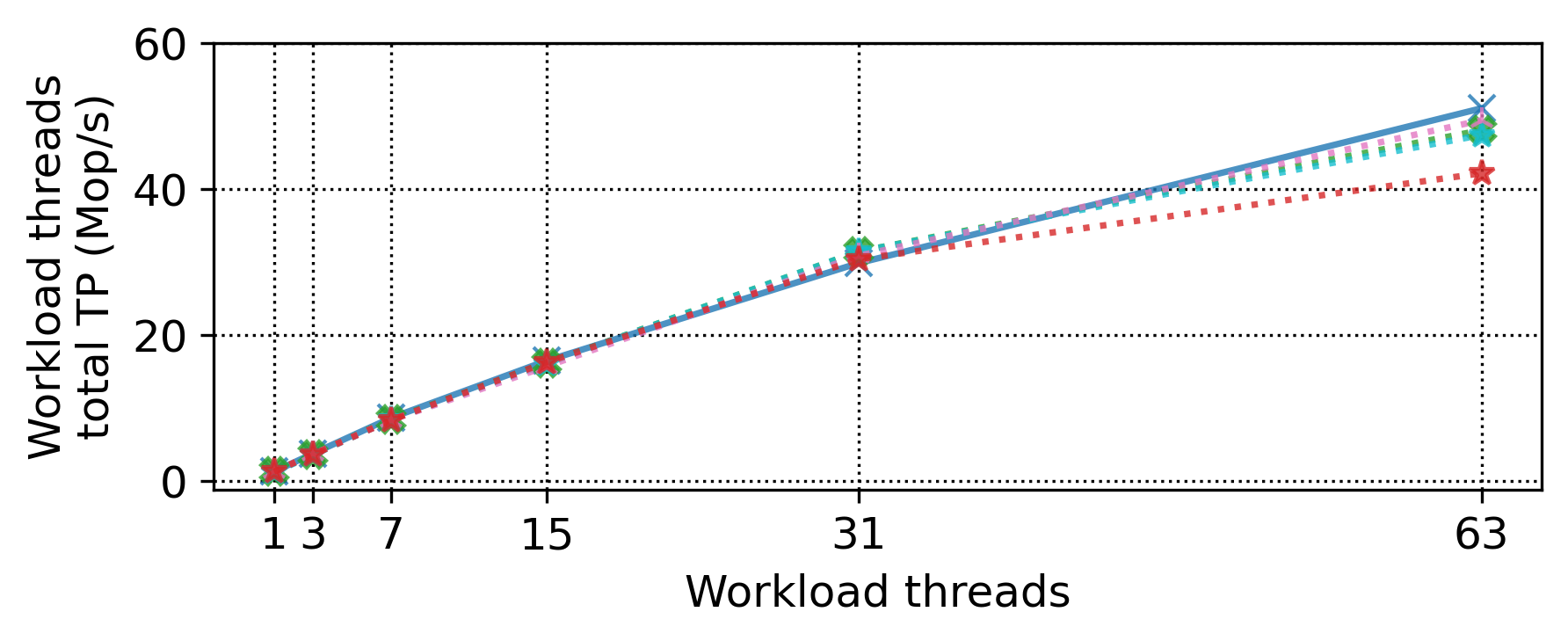}\hspace{2.5em}
	\includegraphics[width=.45\textwidth,trim={0 0 0 .1cm}]{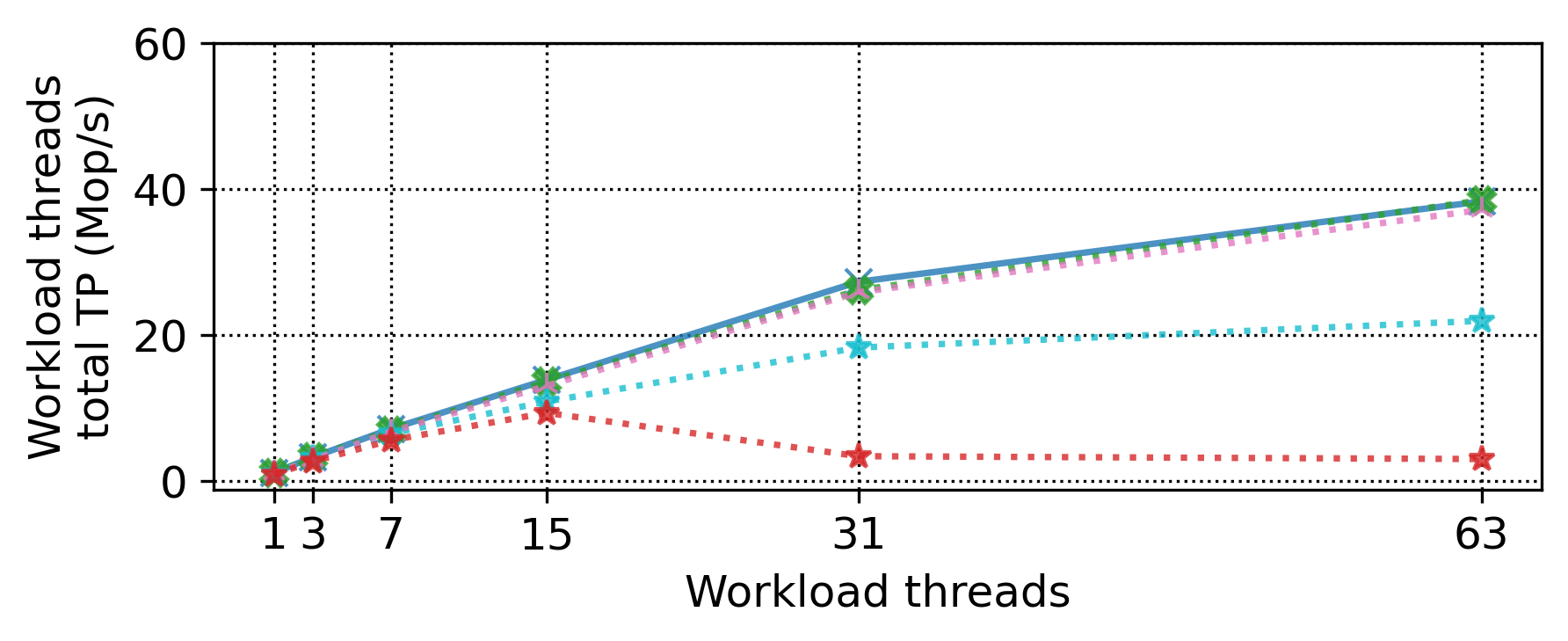}\par
	\includegraphics[width=.45\textwidth,trim={0 0 0 .2cm}]{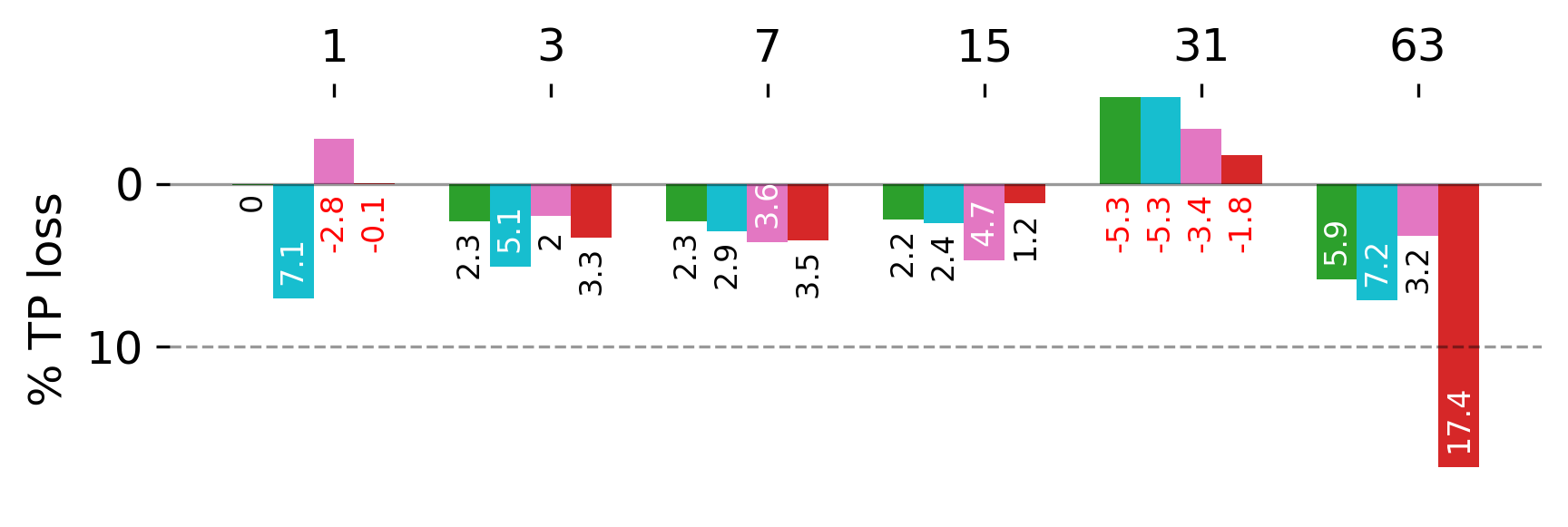}\hspace{2.5em}
	\includegraphics[width=.45\textwidth,trim={0 0 0 .2cm}]{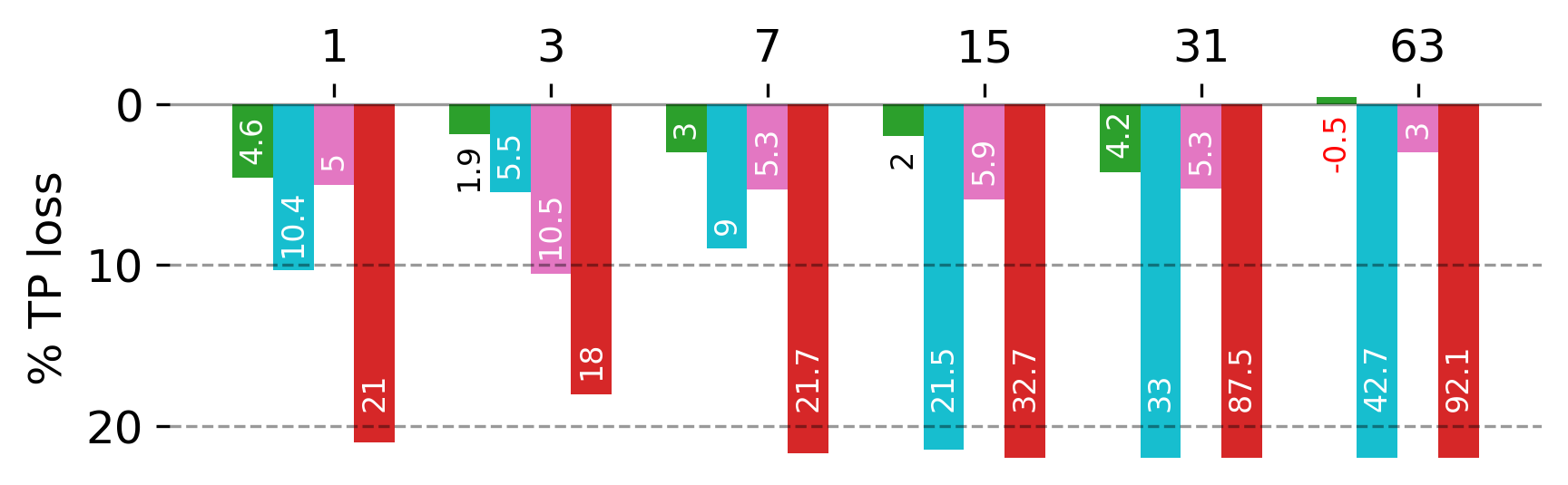}\par
	\medskip
	\text{With a concurrent \size{} thread and 700 \si{\micro\second} delay}\par
	\includegraphics[width=.45\textwidth,trim={0 0 0 .1cm}]{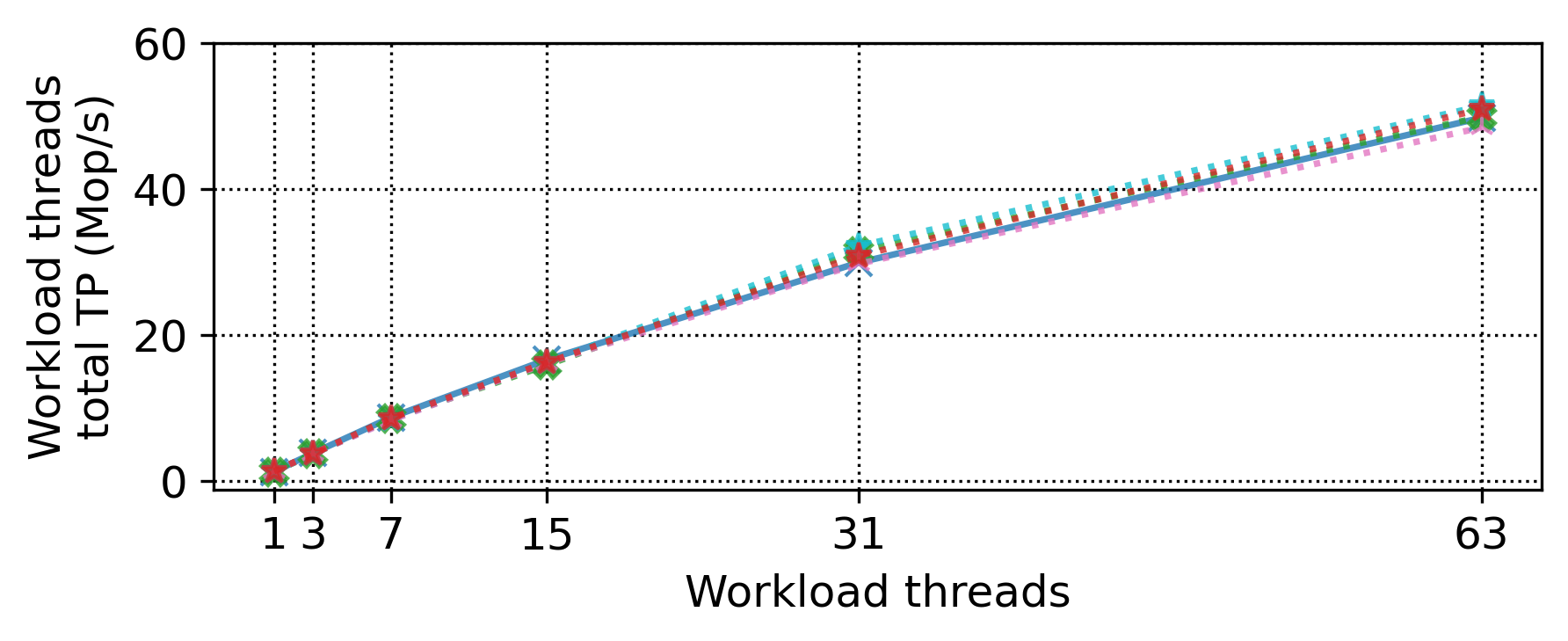}\hspace{2.5em}
	\includegraphics[width=.45\textwidth,trim={0 0 0 .1cm}]{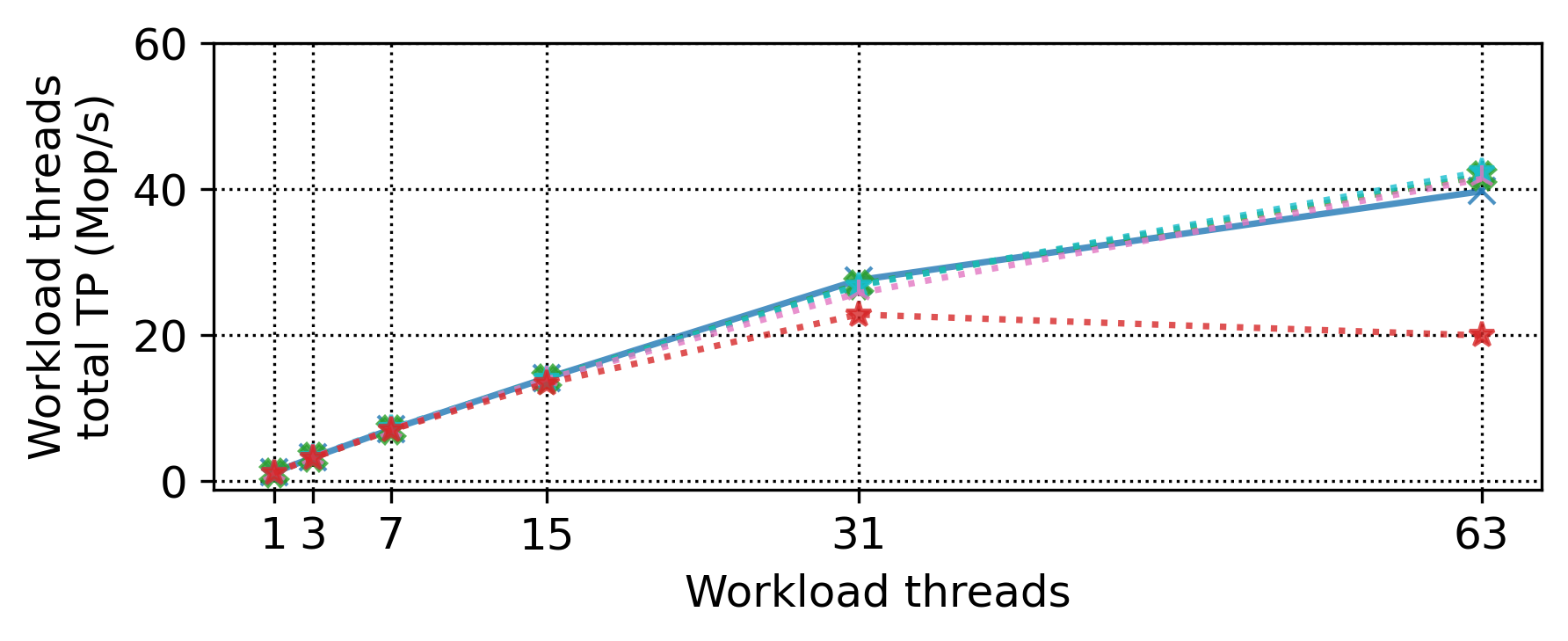}\par
	\includegraphics[width=.45\textwidth,trim={0 0 0 .2cm}]{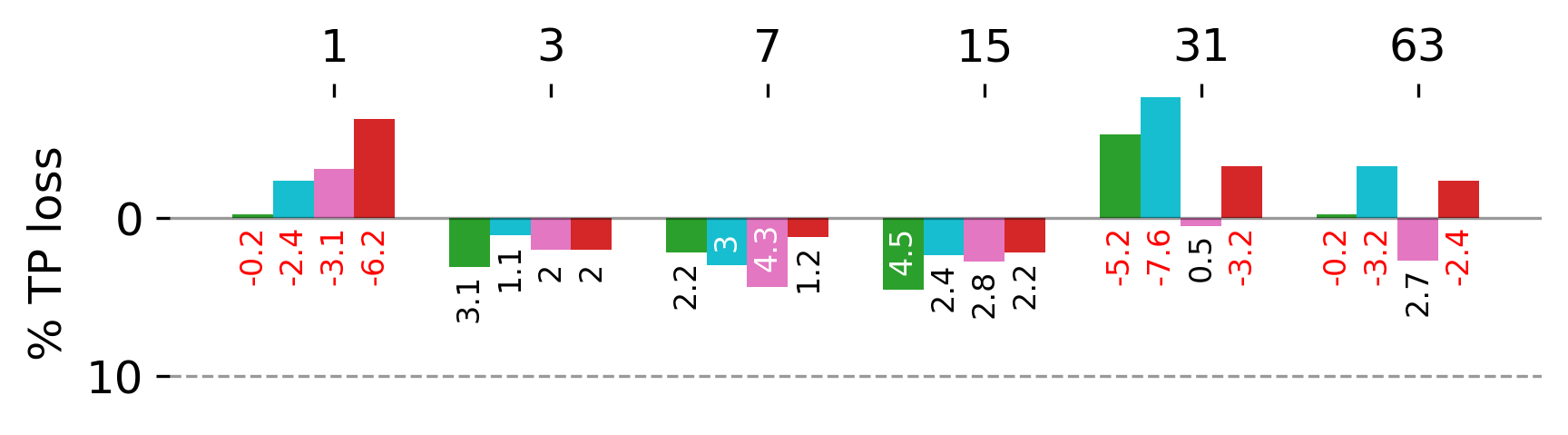}\hspace{2.5em}
	\includegraphics[width=.45\textwidth,trim={0 0 0 .2cm}]{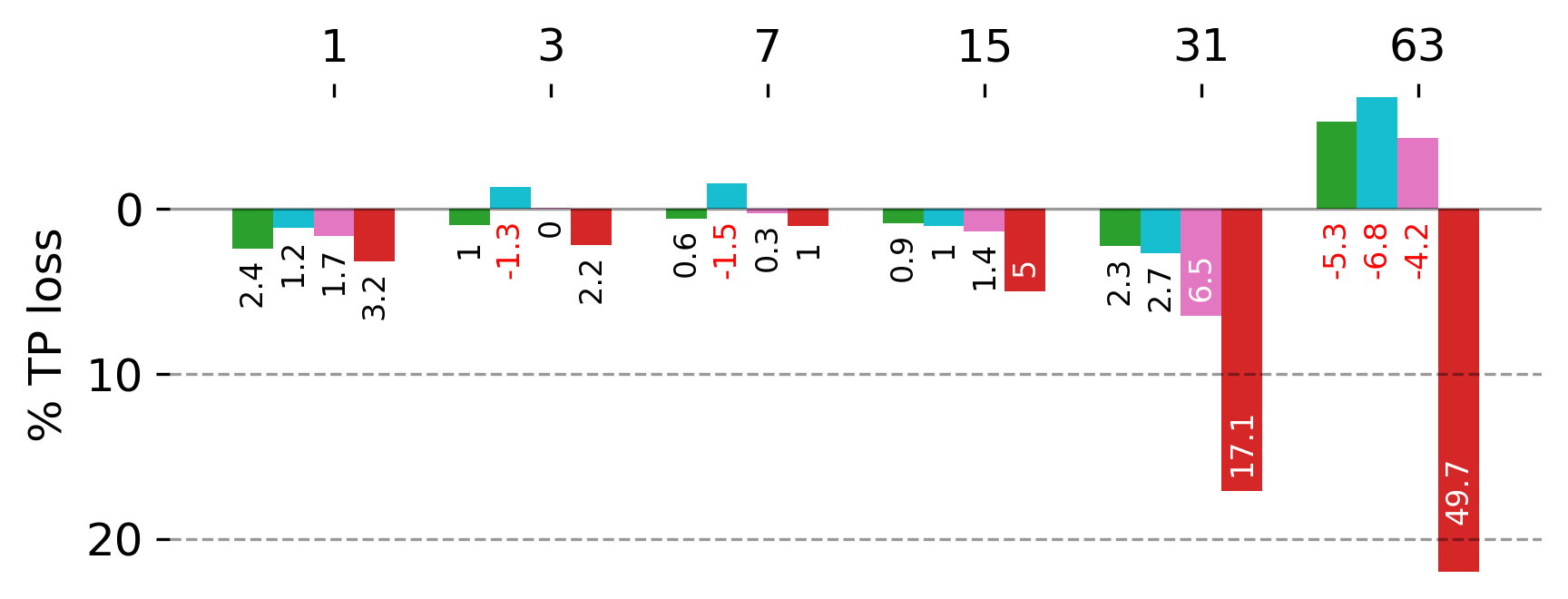}\par
	\caption{Overhead on BST operations}
	\label{fig:BST overhead}
\end{figure*}

\begin{figure*}[htbp]
	\centering
	\medskip
	\textit{Read heavy}\hspace{0.4em}
	\includegraphics[height=.018\textwidth]{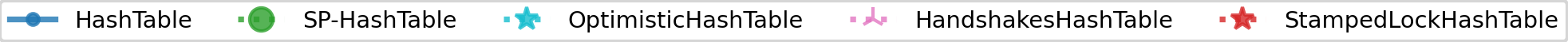}\hspace{0.4em}
	\textit{Update heavy}\par
	\medskip
	\text{Without a concurrent \size{} thread}\par
        \smallskip
	\includegraphics[width=.45\textwidth,trim={0 0 0 .1cm}]{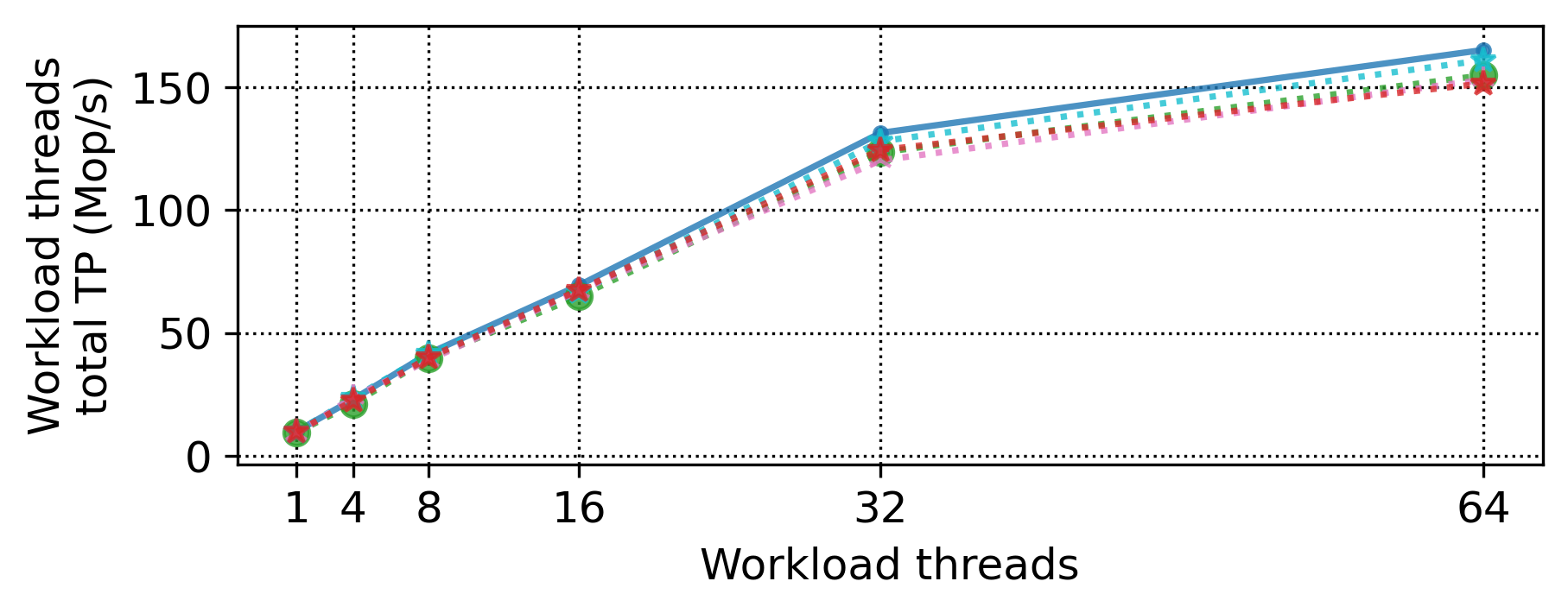}\hspace{2.5em}
	\includegraphics[width=.45\textwidth,trim={0 0 0 .1cm}]{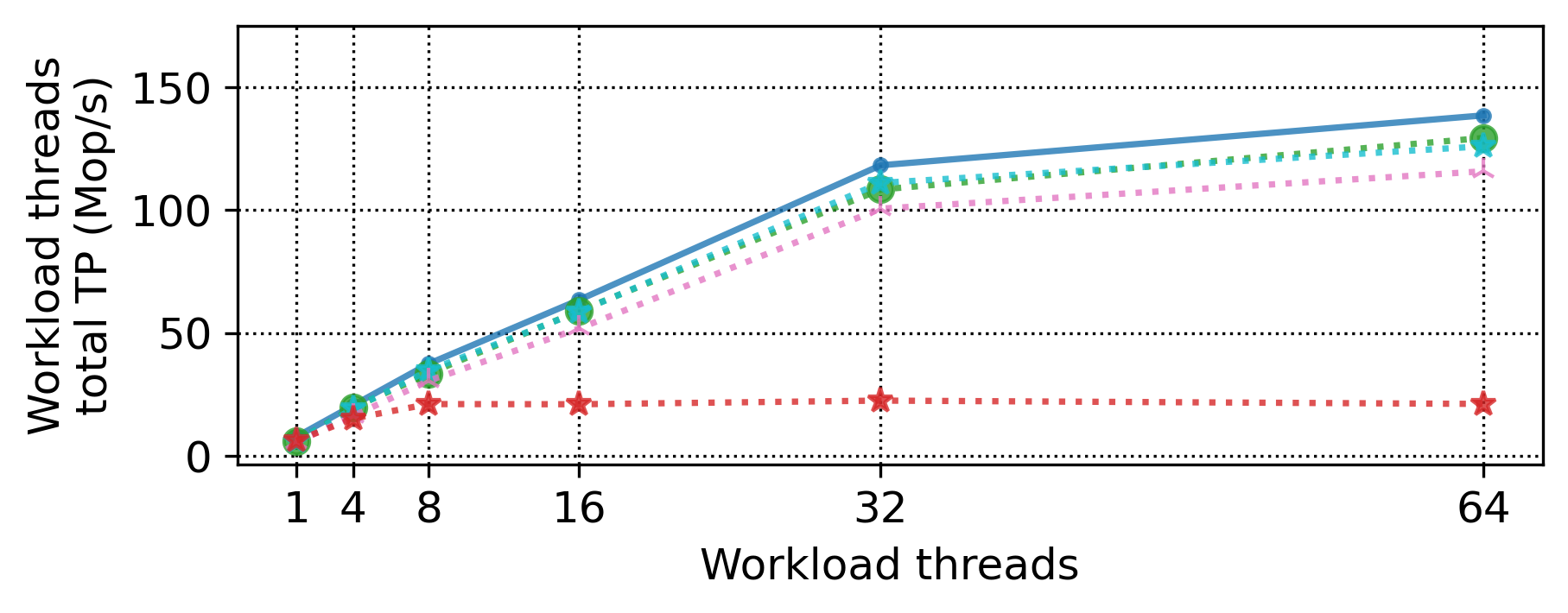}\par
	\includegraphics[width=.45\textwidth,trim={0 0 0 .2cm}]{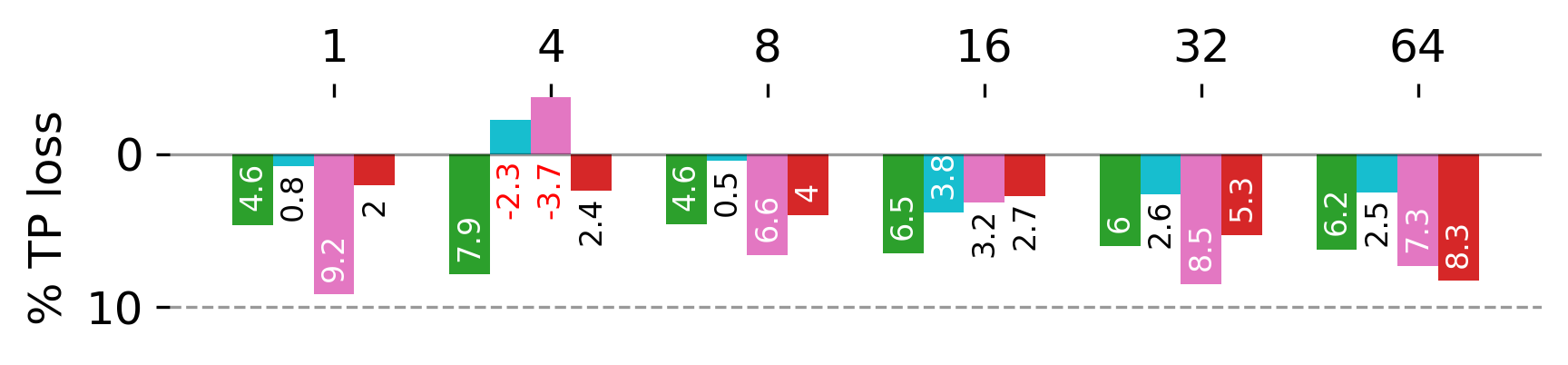}\hspace{2.5em}
	\includegraphics[width=.45\textwidth,trim={0 0 0 .2cm}]{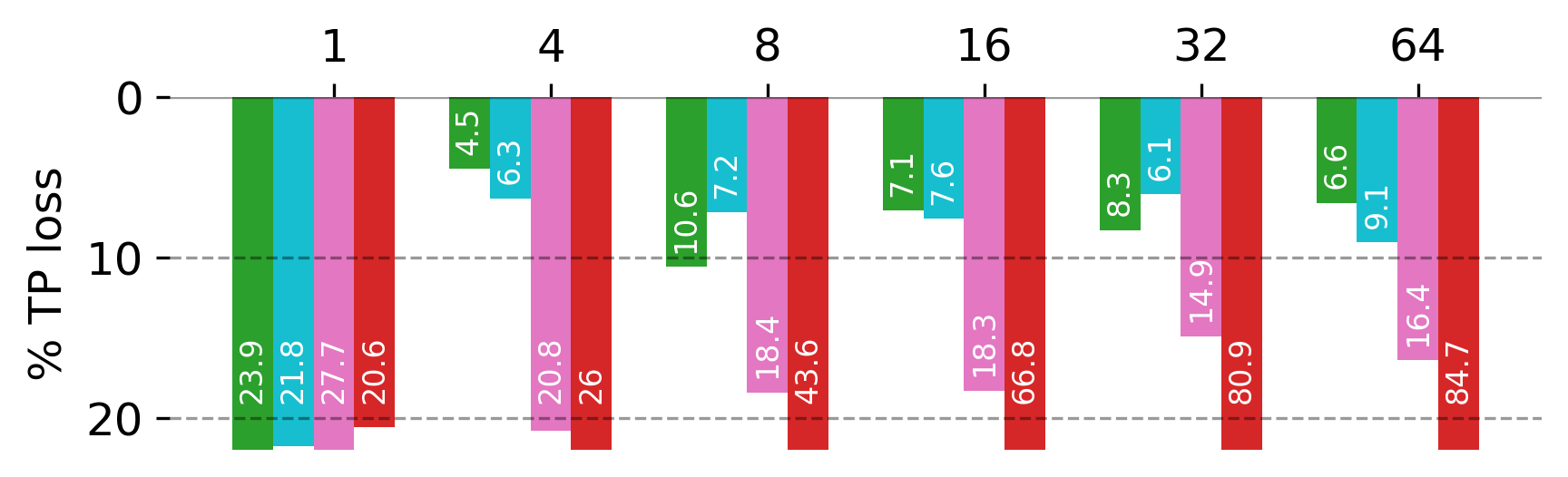}\par
	\medskip
	\text{With a concurrent \size{} thread and no delay}\par
	\includegraphics[width=.45\textwidth,trim={0 0 0 .1cm}]{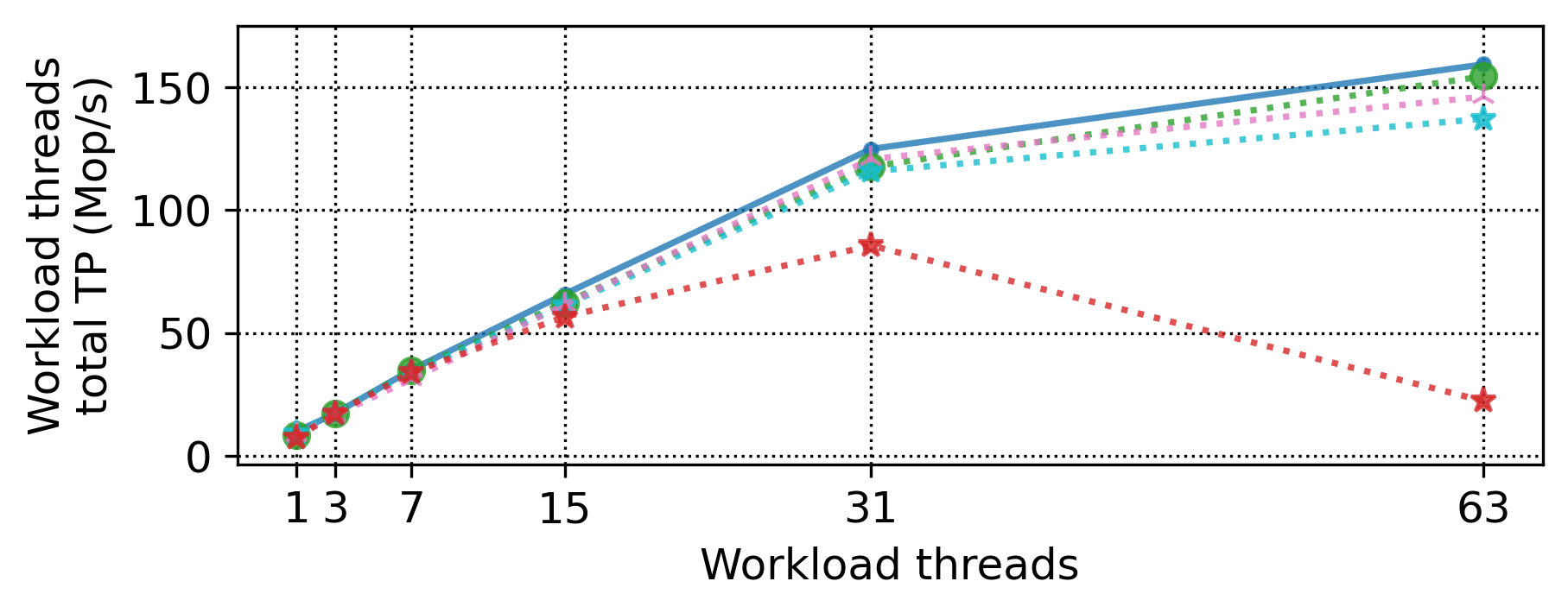}\hspace{2.5em}
	\includegraphics[width=.45\textwidth,trim={0 0 0 .1cm}]{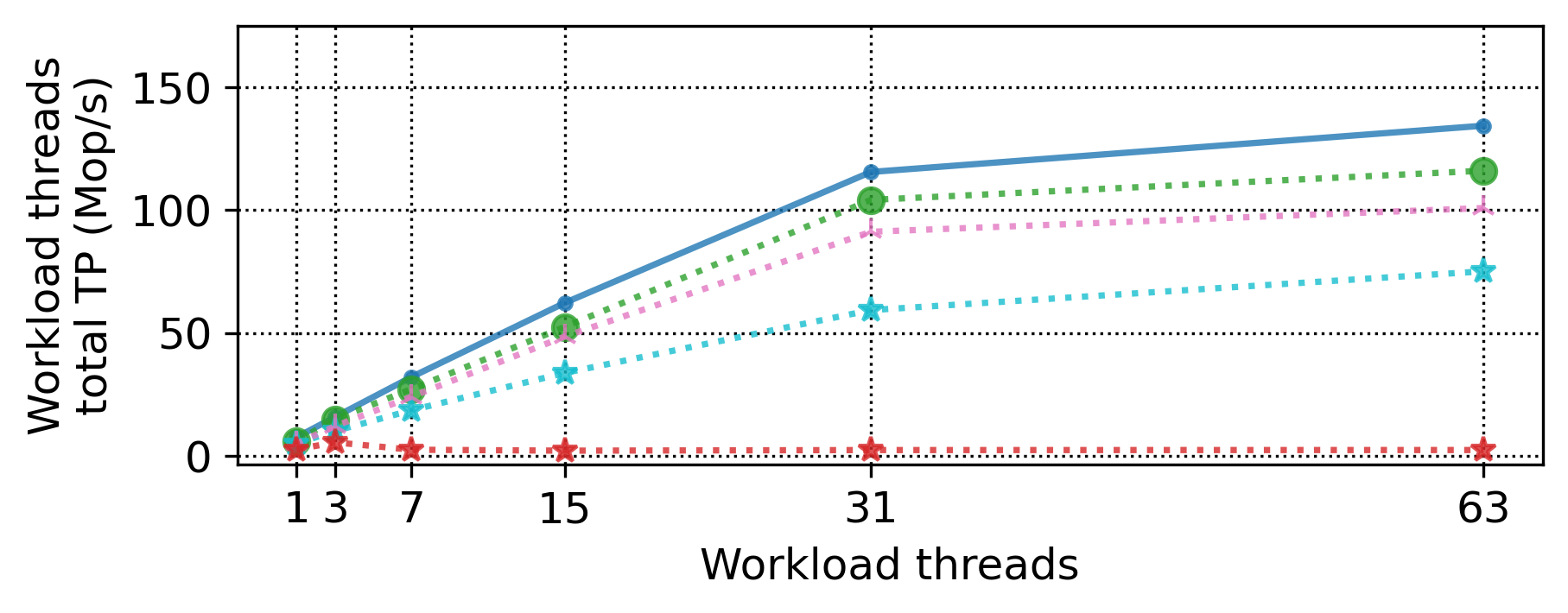}\par
	\includegraphics[width=.45\textwidth,trim={0 0 0 .2cm}]{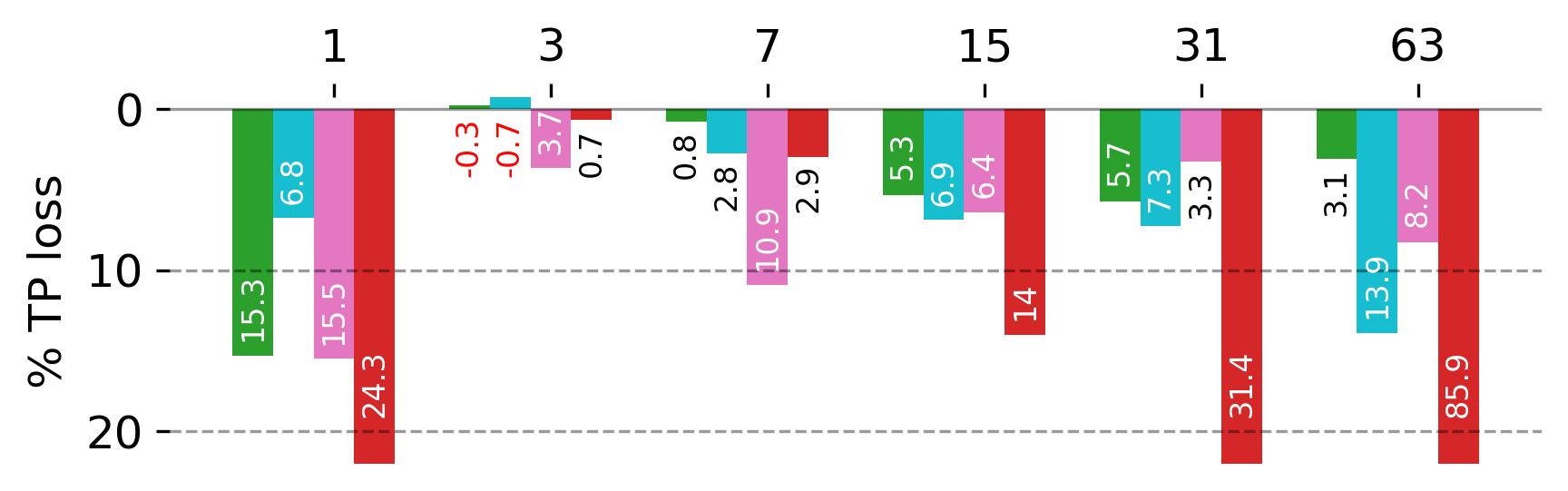}\hspace{2.5em}
	\includegraphics[width=.45\textwidth,trim={0 0 0 .2cm}]{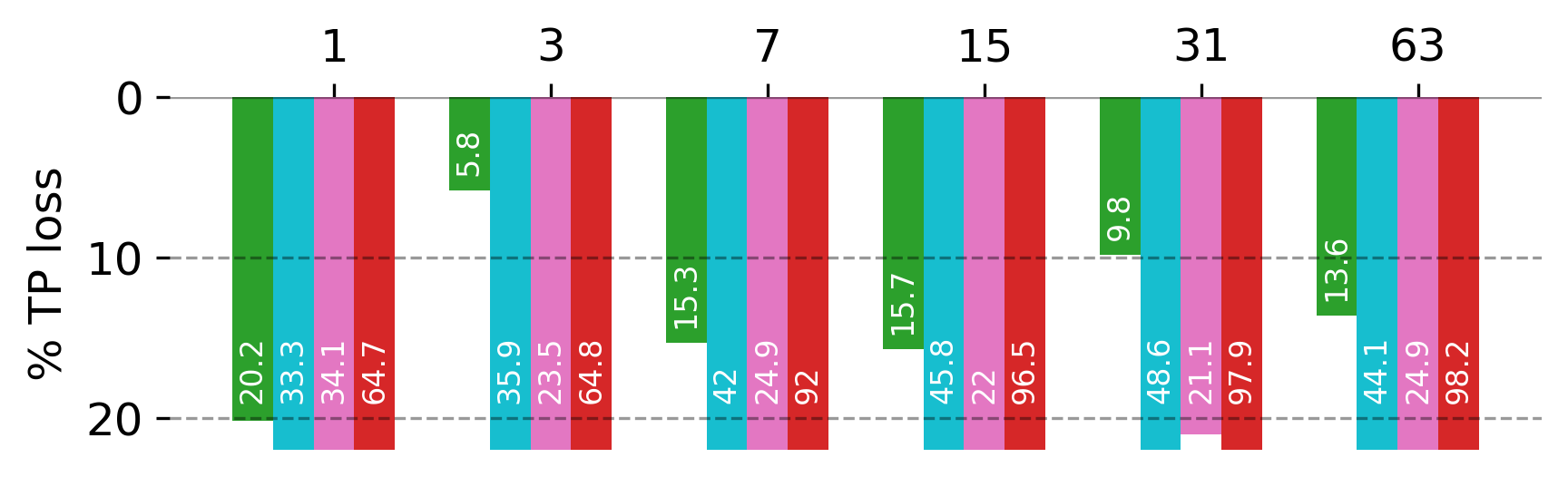}\par
	\medskip
	\text{With a concurrent \size{} thread and 700 \si{\micro\second} delay}\par
	\includegraphics[width=.45\textwidth,trim={0 0 0 .1cm}]{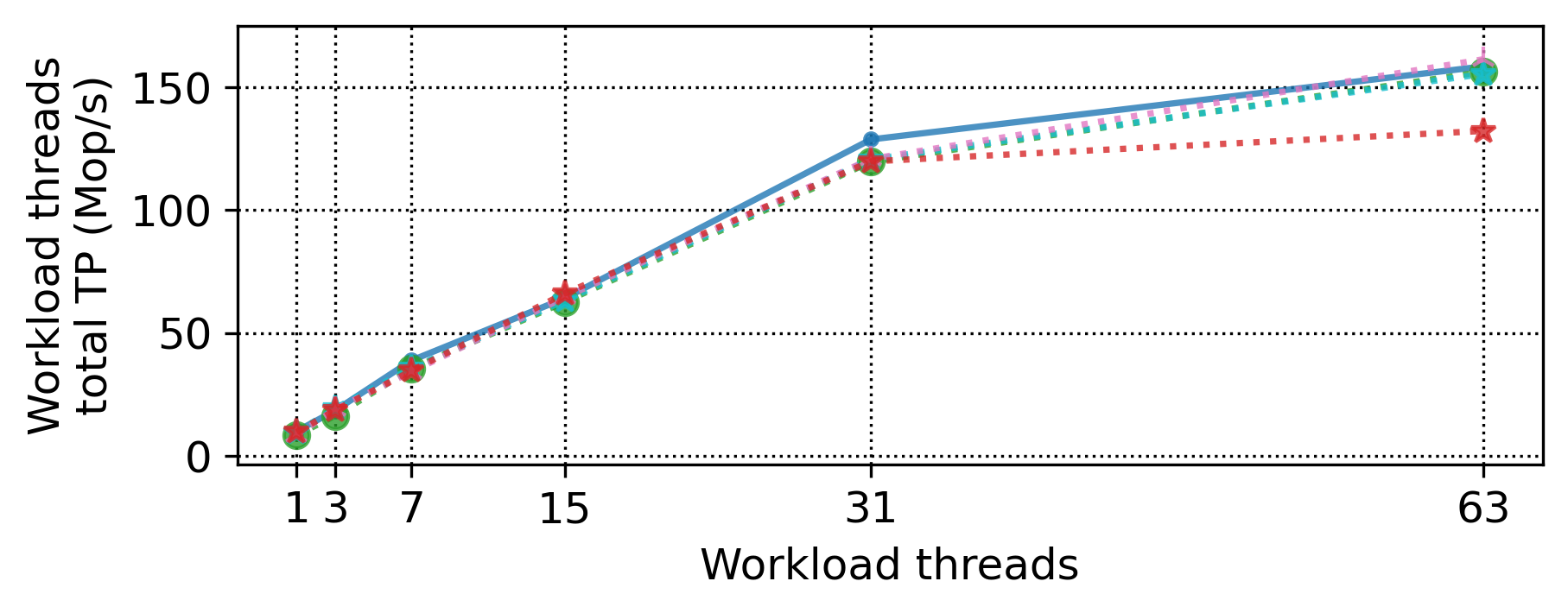}\hspace{2.5em}
	\includegraphics[width=.45\textwidth,trim={0 0 0 .1cm}]{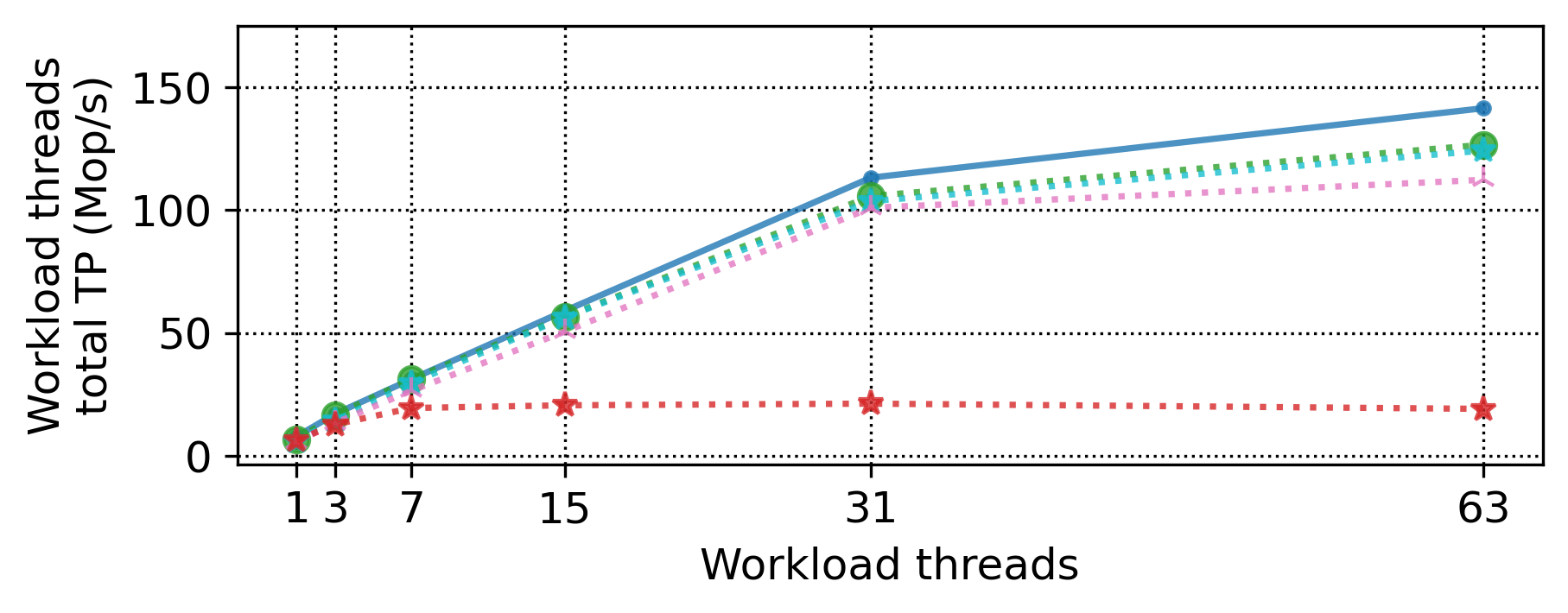}\par
	\includegraphics[width=.45\textwidth,trim={0 0 0 .2cm}]{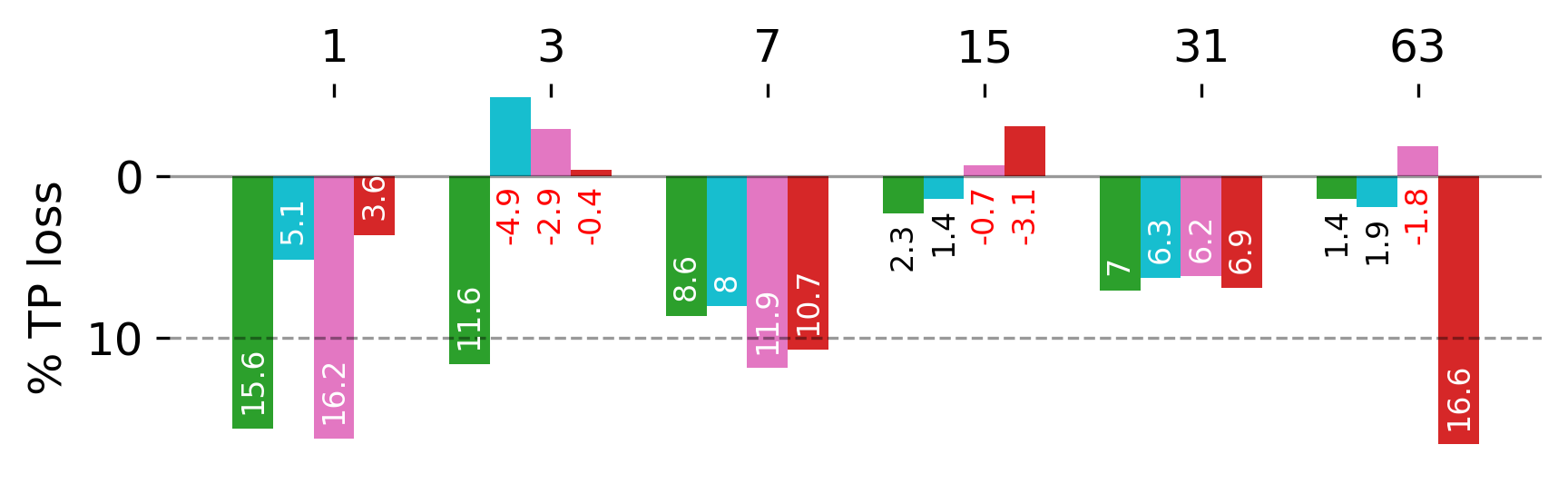}\hspace{2.5em}
	\includegraphics[width=.45\textwidth,trim={0 0 0 .2cm}]{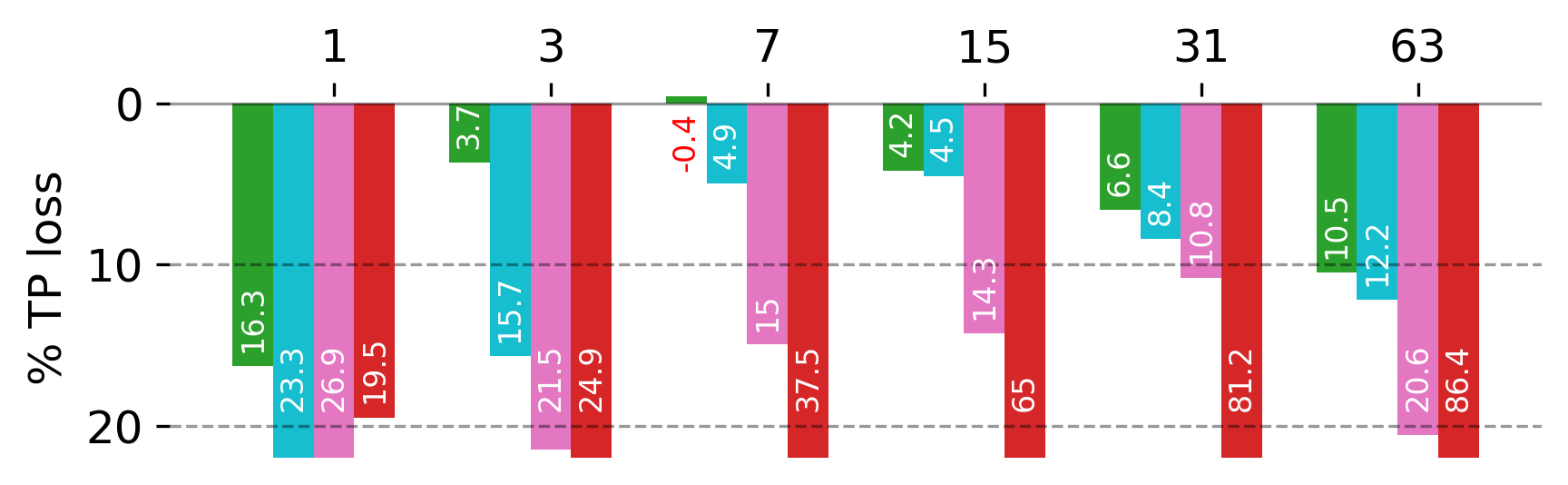}\par
	\caption{Overhead on hash table operations}
	\label{fig:HT overhead}
\end{figure*}

The X-axis represents the number of threads running operations concurrently (1, 4, 8, 16, 32, and 64 when no concurrent \size{} executes, and 1, 3, 7, 15, 31, 63 when one concurrent thread executes \size{} on a separate thread).
Each graph has an upper part and a lower part. The throughput is depicted in the upper part, with the Y-axis showing the number of million operations executed per second. On the lower part of each graph we depict the overhead percents compared to the throughput of the original data structure (that does not support a \size{} execution). The Y-axis shows $0\%$ when no throughput loss is demonstrated, and bigger percentage when overhead is witnessed. To keep the picture clear for small bars, we cut any long overhead bar at ~$20\%$ and write the lost percentage on the specific bar. 

It turns out that there is no one-size-fit-all method. Different scenarios call for different synchronization methods. The observed results vary by the chosen data structure, the contention levels, the frequency of utilizing the \size{} operation, and the workload (read intensive or update intensive). Notably, the hash function exhibits exceptionally fast operations, making the relative overhead on cooperation with size much higher when compared to both the skip list and the BST. In scenarios characterized by a write-heavy workload, the original \spsize{} approach emerges as the recommended strategy for the hash table. Its overhead averages around $10\%$, and generally within the range of $0-20\%$. Conversely, for read-intensive workloads, the optimistic approach proves optimal, showcasing an average overhead of around $4\%$ and fluctuating between $0-14\%$. 

The skip list and BST perform comparably, demonstrating commendable performance with both the lock-based and optimistic methods, particularly in read-heavy workloads. However, their efficacy diminishes significantly under write-heavy workloads, especially during heightened contention and when the \size{} operation is actively employed. The optimistic approach is somewhat less harmful in this scenario, but still performs poorly when employed with repeated \size{} calls. Consequently, when anticipating write-heavy workloads or when usage patterns are uncertain, both the \spsize{} and handshake approaches exhibit superior performance. Notably, the handshake approach on the skip list maintains an average overhead of approximately $4.4\%$, surpassing the \spsize{} approach by about $1\%$, and ranging between $0-12.9\%$. On the other hand, for the BST, the \spsize{} approach consistently outperforms the other approaches, showcasing an average overhead of $2.4\%$ and ranging between $0-7.1\%$ across all scenarios.

The handshake approach does not offer a performance advantage over \spsize{} in many scenarios due to several overheads inherent in its design. For instance, even when no \size{} operation is active, update operations on the fast path still perform some conditional checks, which are not performed in \spsize{}, as well as write to the shared \codestyle{opPhase} array. Additionally, \codestyle{contains{}} operations are always executed on the slow path, introducing additional checks (in comparison to the baseline data structures) to ensure correctness under concurrent, slow updates. The cumulative overheads reduce the expected performance gains of the handshake approach over \spsize{}, particularly for short operations common in hash tables.



\subsection{Size Scalability}

Next, we study the scalability of the \size{} operation across the studied synchronization methods. We measured the total throughput of threads executing the \size{} operation in each data structure, both in a read-oriented and a write-oriented workload. We ran 32 workload threads concurrently with \size{}-executing threads, whose number varies between 1 to 32 (so the overall number of running threads was bounded by 64). The results are depicted in~\Cref{fig:size scalability SL,fig:size scalability BST,fig:size scalability HT}.
\begin{figure*}[htbp]
    \centering
    \medskip
    \textit{Read heavy}\quad\quad
    \includegraphics[height=.0205\textwidth]{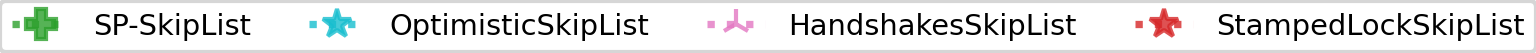}\quad\quad
    \textit{Update heavy}\par
    \includegraphics[width=.45\textwidth,trim={0 0 0 .1cm}]{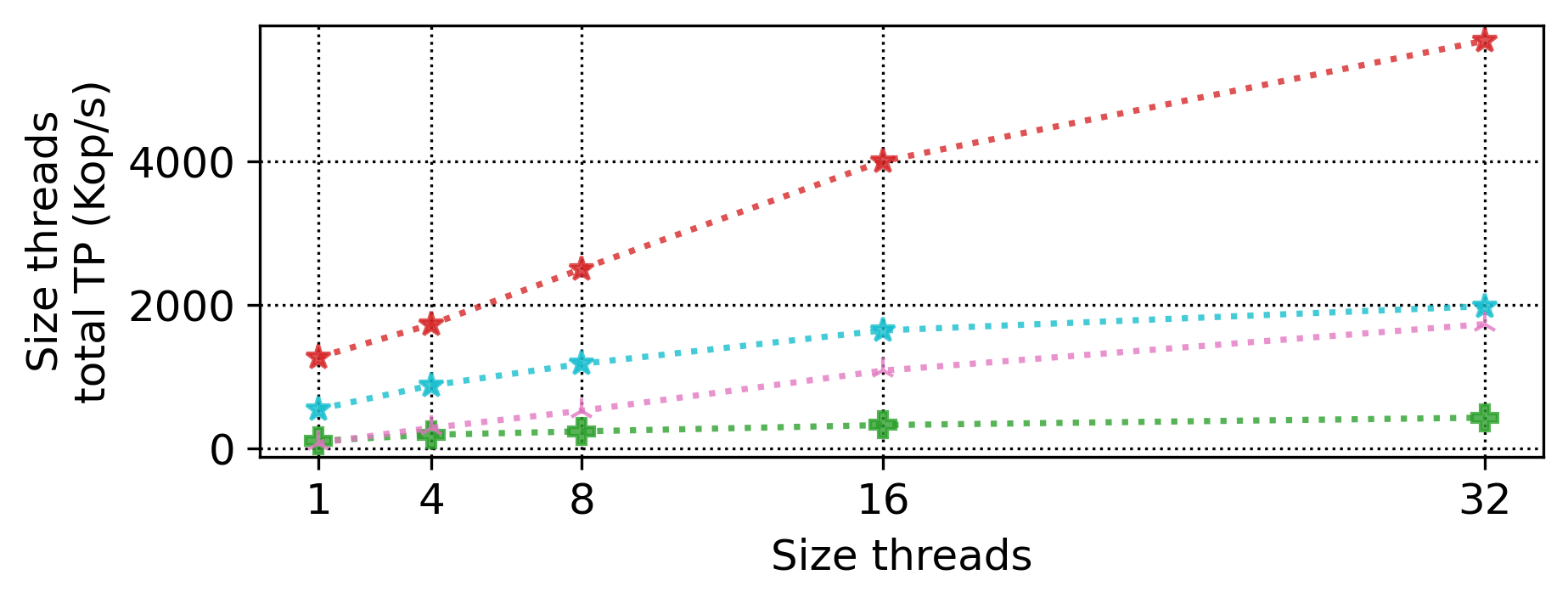}\hspace{2.5em}
    \includegraphics[width=.45\textwidth,trim={0 0 0 .1cm}]{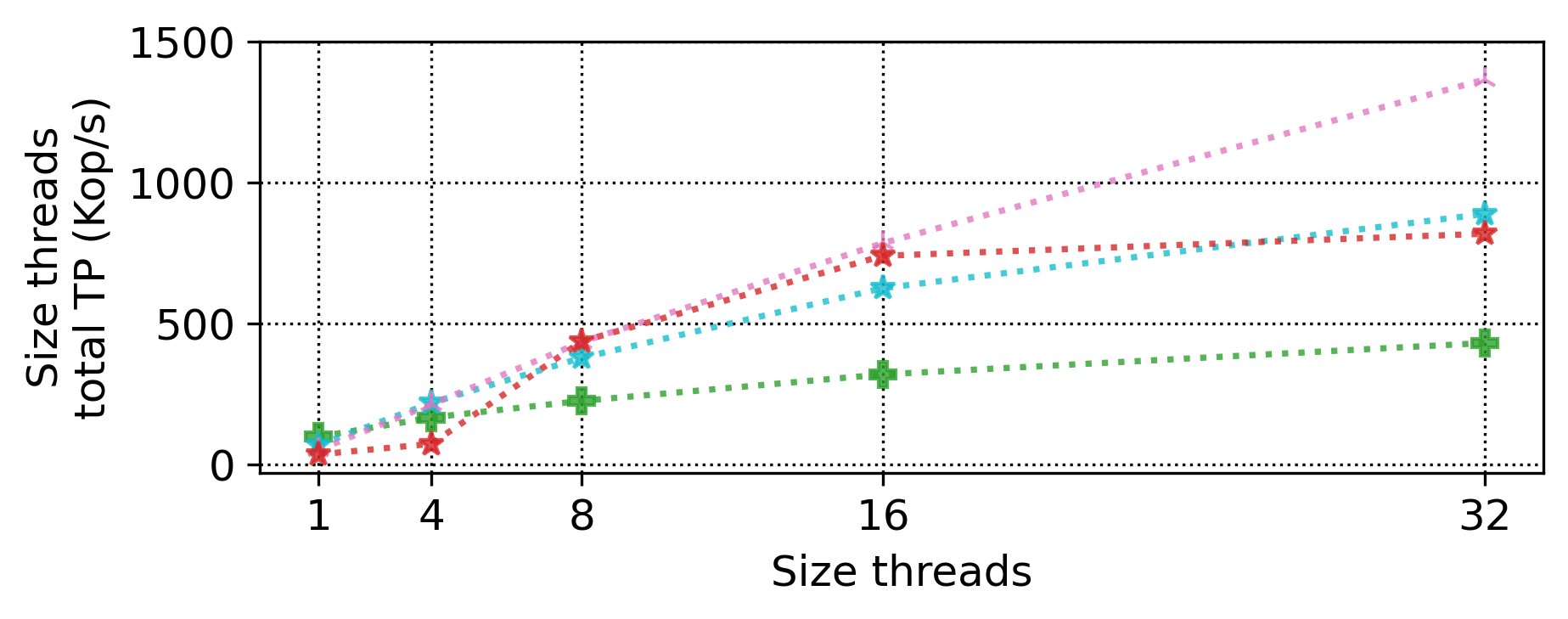}\vspace{-1.2em}
    \caption{size scalability in skip list}
    \label{fig:size scalability SL}
    \medskip
    \medskip
    \medskip
    \textit{Read heavy}\hspace{3.3em}
    \includegraphics[height=.021\textwidth]{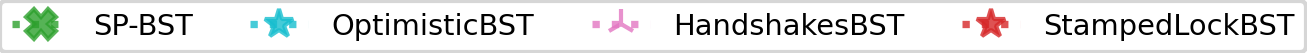}\hspace{3.3em}
    \textit{Update heavy}
    \includegraphics[width=.45\textwidth,trim={0 0 0 .1cm}]{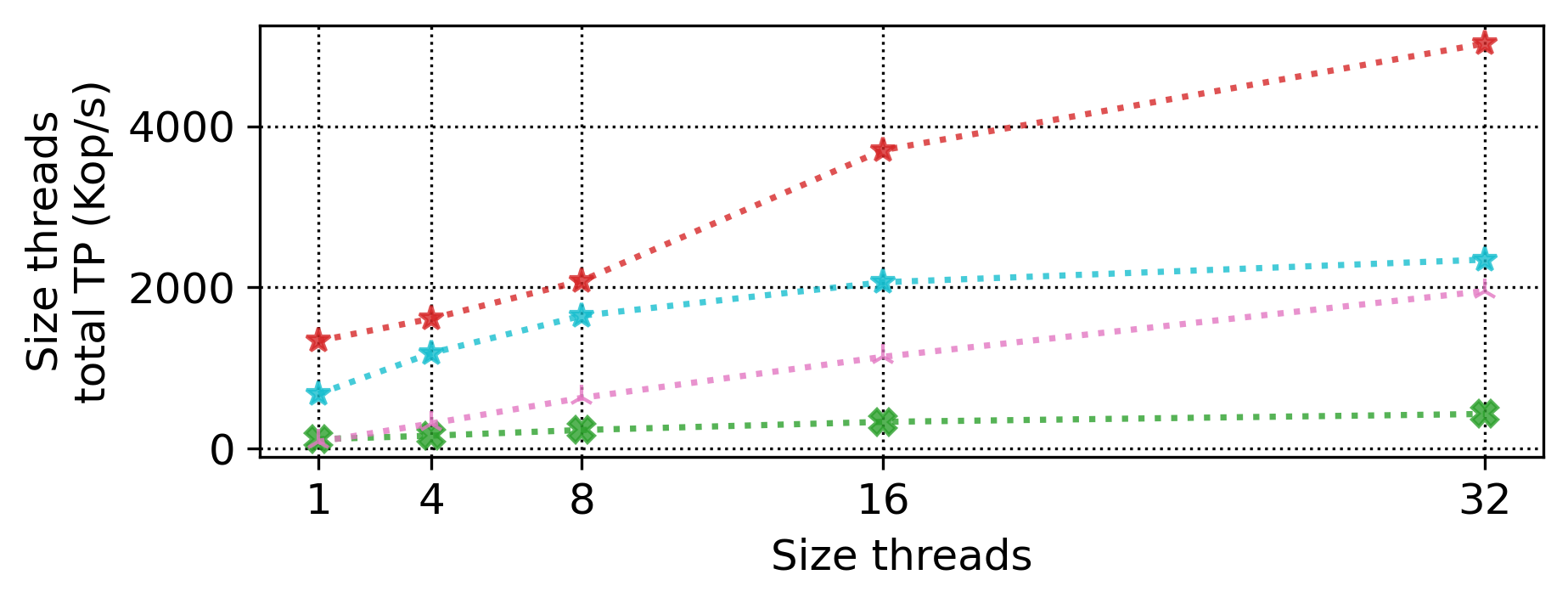}\hspace{2.5em}
    \includegraphics[width=.45\textwidth,trim={0 0 0 .1cm}]{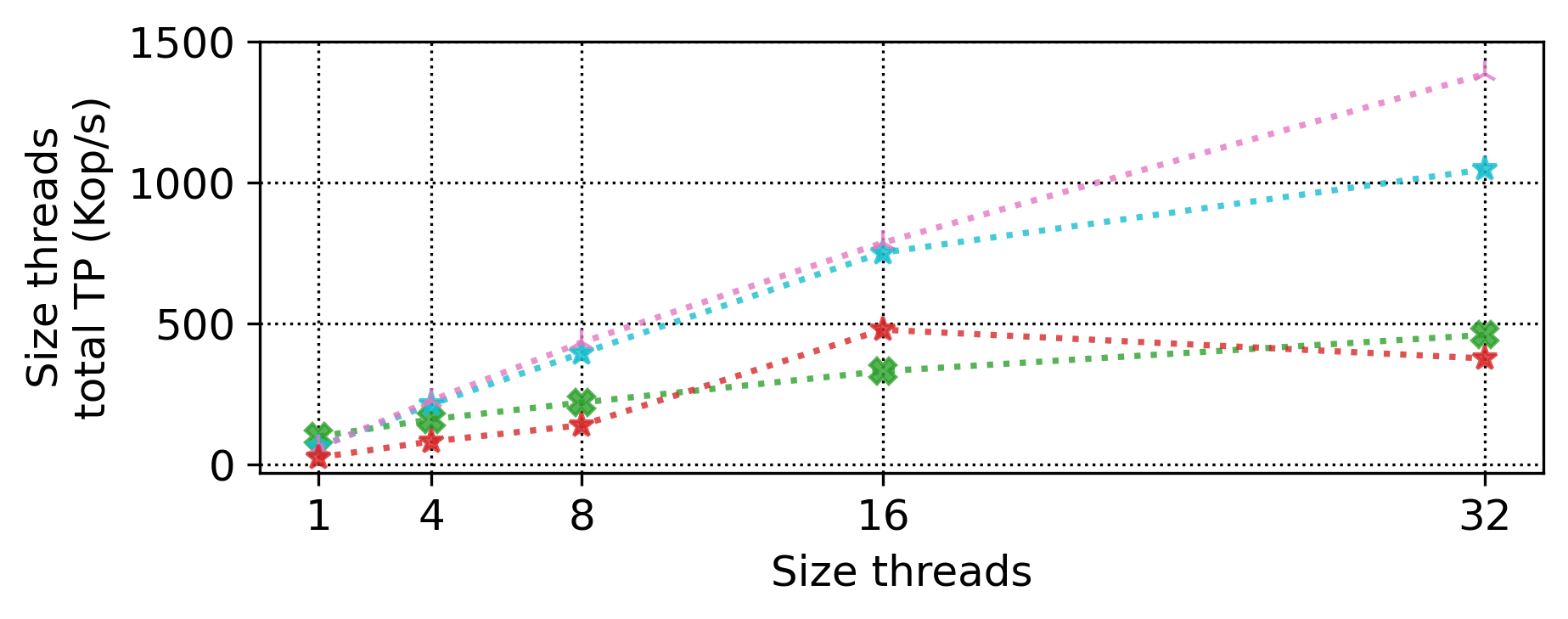}\vspace{-1.2em}
    \caption{size scalability in BST}
    \label{fig:size scalability BST}
    \medskip
    \medskip
    \medskip
    \textit{Read heavy}\hspace{0.2em}
    \includegraphics[height=.021\textwidth]{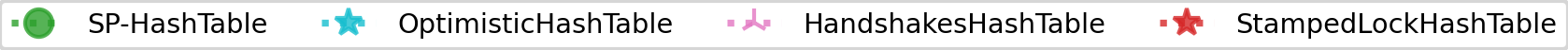}\hspace{0.2em}
    \textit{Update heavy}
    \includegraphics[width=.45\textwidth,trim={0 0 0 .1cm}]{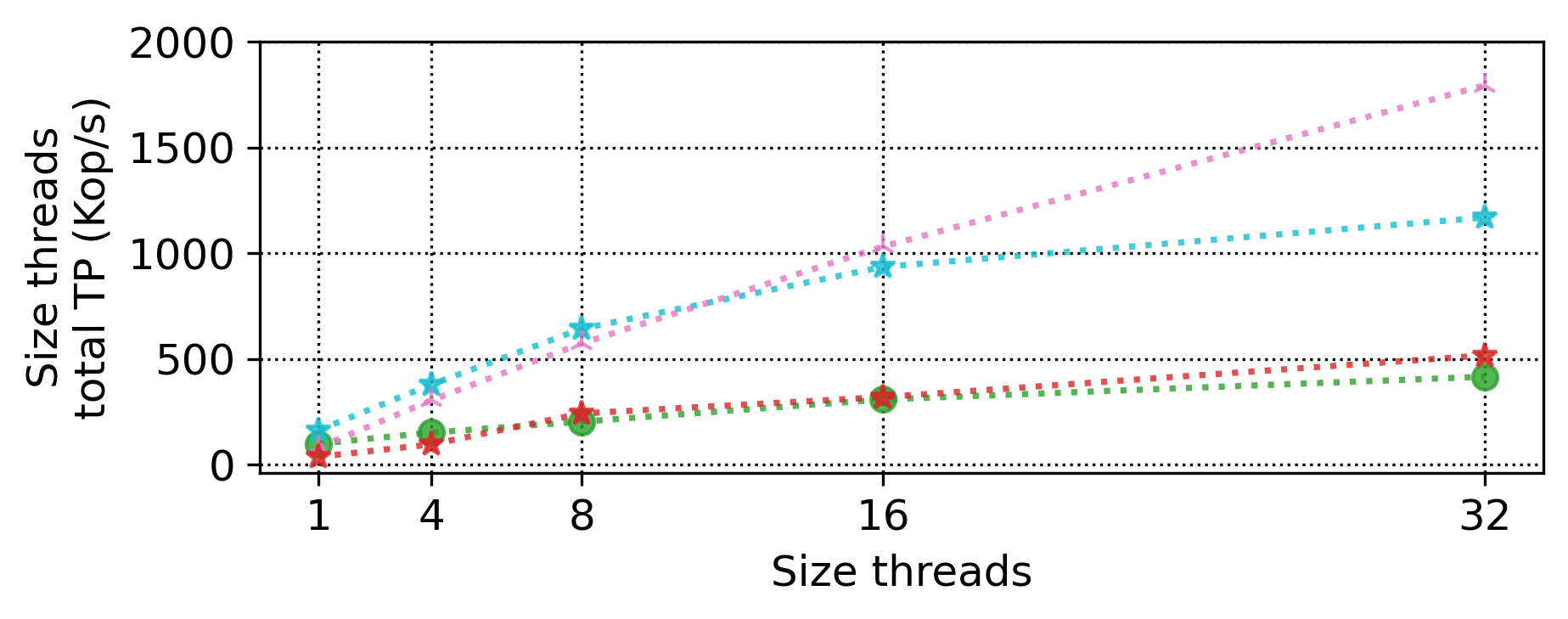}\hspace{2.5em}
    \includegraphics[width=.45\textwidth,trim={0 0 0 .1cm}]{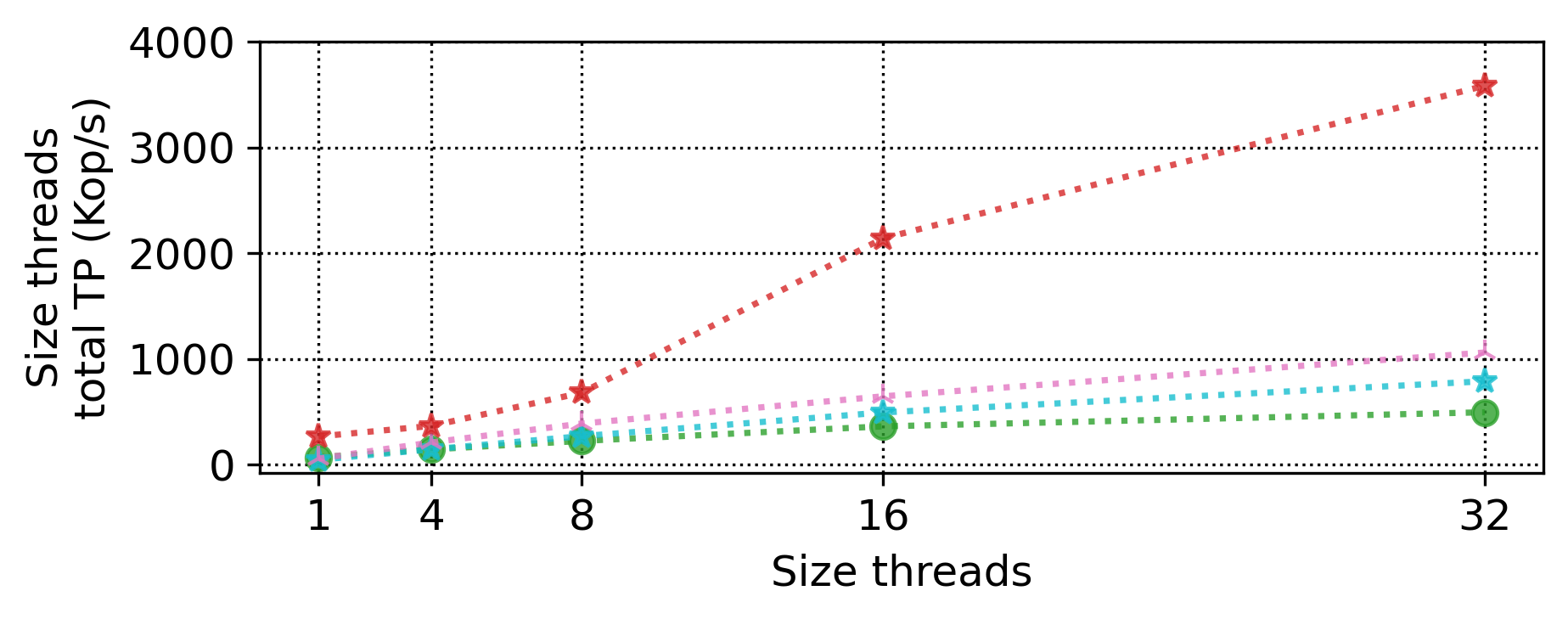}\vspace{-1.2em}
    \caption{size scalability in hash table}
    \label{fig:size scalability HT}
\end{figure*}

Again the results differ between the hash table and the other two data structures. For the BST and the skip list on a read-oriented workload, the lock-based method does significantly better. A \size{} operation that grabs the write lock can execute very fast with no interference from the data structure operations.

In contrast, when the data structure is updated frequently, the wait for acquiring the lock becomes dominant and the handshake synchronization wins, albeit not with such a significant advantage as locks have on a read-heavy workload.

Hence, if the scalability of the \size{} is of high importance, then the lock-based method should be favored with read-oriented workloads, and for write-heavy workloads the handshake synchronization does better than the lock-based synchronization. Normally, we expect the overhead on the operations themselves to be the more important consideration, as they typically occur more frequently.

Similarly to the investigation of overheads, the hash table behaves completely different. The handshake synchronization wins on the read-heavy workload and the lock-based synchronization wins on the write-heavy workload. Lock-based synchronization is not recommended for hash tables as the \size{} operations take over the lock and allow almost no data structure operations to execute, as can be seen in the middle row, right side, of \Cref{fig:HT overhead}. The handshake approach or the \spsize{} synchronization method may thus be the best overall choice for this case.

The demonstrated superior scalability of the handshake approach over \spsize{} is due to its relaxed progress guarantee. The \spsize{} method guarantees wait-freedom. To achieve a wait-free \size{}, a \size{}() operation in \spsize{} that detects a concurrent \size{}() execution, attempts to help it compute the size and competes with it on performing a compare-and-swap on the same snapshot cells. As more \size{} threads are added, the contention grows, resulting in degradation of performance and hindering scalability. 
In contrast, in the handshake method, the first concurrent \size{} caller becomes a leader and carries out the operation, while the rest of the concurrent \size{} calls simply block and wait for the result. 
Hence, the handshake method scales better in workloads with frequent concurrent \size{} operations: adding more \size{} threads introduces much less contention among them, resulting in better performance.

\subsection{MAX\_TRIES measurements}
In \Cref{sec:max_tries} of the supplementary material we assess the impact of the \codestyle{MAX\_TRIES} variable on the optimistic method. We conducted additional measurements comparing its overhead on the original data structure's operations and on the performance of the \size{} operation with \codestyle{MAX\_TRIES} values ranging from 2 to 16. Since the results show a clear trade-off between the performance of the original operations and the performance of the \size{} operation, it is recommended to select a \codestyle{MAX\_TRIES} value that is well calibrated to the user preference and  system. The chosen value should consider the workload characteristics and the desired balance between minimizing overhead on standard operations while ensuring accurate and timely \size{} calculations.

\subsection{Zipfian measurements}\label{sec-zipfian}
To evaluate the robustness of our results under skewed access patterns, we conducted an additional set of overhead experiments using a Zipfian distribution for key selection in \contains{} operations. In contrast, \ins{} and \del{} operations continued to use keys drawn uniformly at random from the global key space which remained the same. We employed the \texttt{ScrambledZipfianGenerator} from the YCSB benchmark suite~\cite{cooper2010benchmarking}, which ensures a randomized mapping of hot keys across the key space. The skew parameter $\theta$ was set to $0.99$, a standard choice for simulating strong locality~\cite{cooper2010benchmarking}. Our findings, appearing in \Cref{sec:zipfian graphs}, reveal no major differences from the results obtained in \Cref{sec-overheads} using uniformly distributed workloads, suggesting that the methods evaluated maintain consistent performance characteristics under non-uniform read patterns as well.

\section{Conclusion}\label{section:conclusion}
This paper investigated the use of various synchronization methods for the important \size{} property of concurrent data structures. Three designs—employing handshake, optimistic, and lock-based synchronization techniques—were proposed to implement a linearizable \size{} operation for sets and dictionaries. Their performance was evaluated against each other and compared to the existing method of~\cite{sela2021concurrentSize} using various workloads.
The evaluation showed that while no single scheme offers the best performance in all scenarios, the findings align with general trends in concurrent computing.
In low-contention scenarios, optimistic and lock-based approaches perform best, whereas, under high contention, the handshake and wait-free SP methods provide the most effective solutions.

An interesting direction for future work is to design an adaptive solution that monitors the workload and dynamically switches between the handshake, optimistic, and lock-based synchronization methods. Such a system would require lightweight workload identification and efficient mode transitions, and it remains to be seen whether the potential performance gains justify the added complexity.
\bibliographystyle{ACM-Reference-Format}
\bibliography{refs}
\newpage

\appendix

\section{MAX\_TRIES measurements}\label{sec:max_tries}
In this section we evaluate the impact of the \codestyle{MAX\_TRIES} parameter on the optimistic method in the different data structures. Presented in  \Cref{fig:MAX_TRIES SL,fig:MAX_TRIES BST,fig:MAX_TRIES HT}, we conducted additional measurements to compare both the overhead imposed on the original data structure and on the performance of the \size{} operation across \codestyle{MAX\_TRIES} values ranging from 2 to 16. Each experiment was performed under two workload conditions: a read-heavy scenario and a update-heavy scenario.

\begin{figure*}[b!]
	\centering
	\medskip
	\textit{Read heavy}\hspace{10cm}
	\textit{Update heavy}\par
        \medskip
        \includegraphics[height=.025\textwidth]{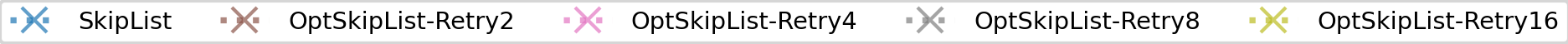}\vspace{0.5em}
	\text{MAX\_TRIES effect on non-\size{} operations' performance in skip list}\par
        \medskip
	\includegraphics[width=.45\textwidth,trim={0 0 0 .1cm}]{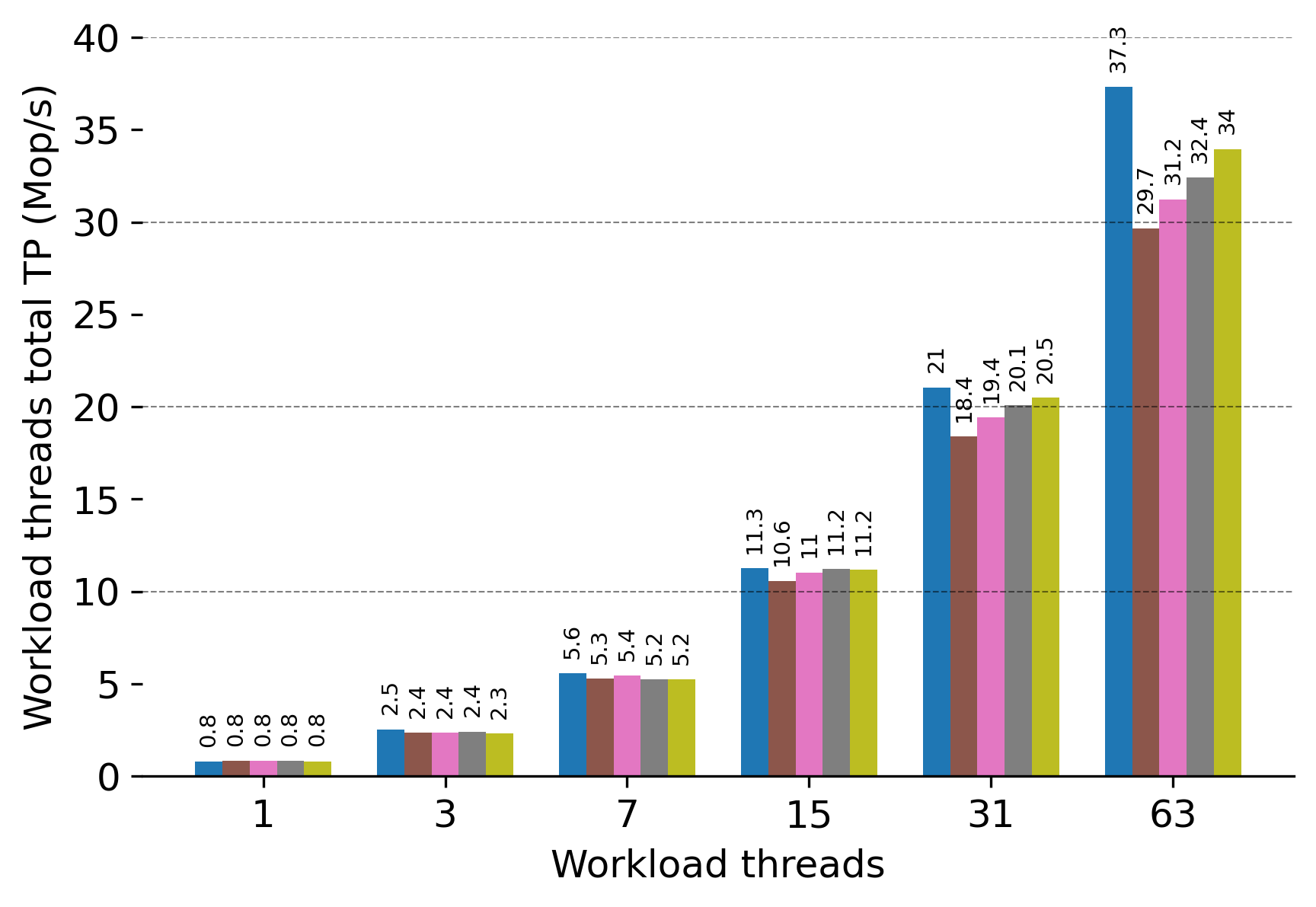}\hspace{2.5em}
	\includegraphics[width=.45\textwidth,trim={0 0 0 .1cm}]{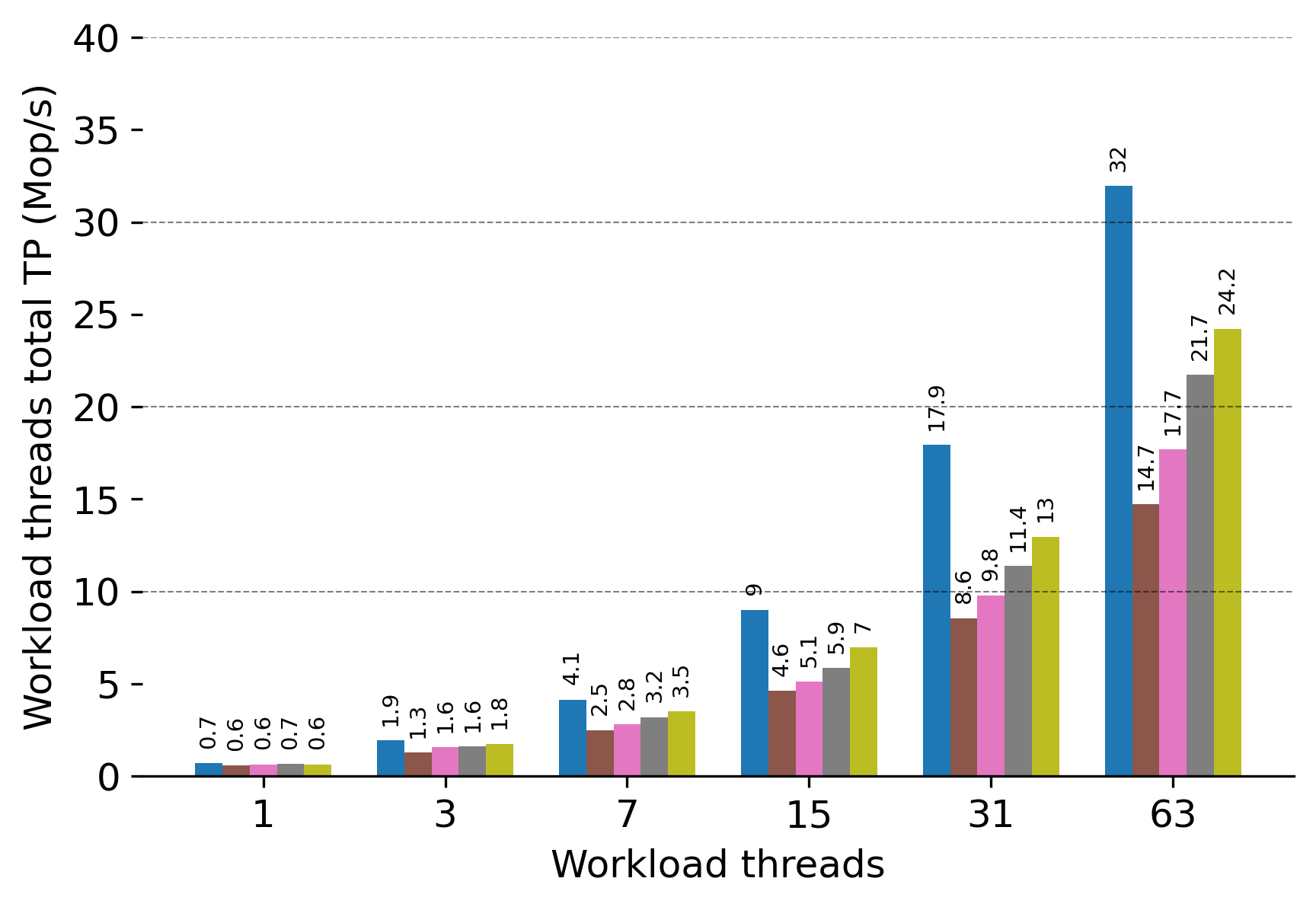}\vspace{-1.2em}

        \vspace{1.2em}
        \medskip
	\text{MAX\_TRIES effect on \size{} operations' performance in skip list}\par
        \medskip
	\includegraphics[width=.45\textwidth,trim={0 0 0 .1cm}]{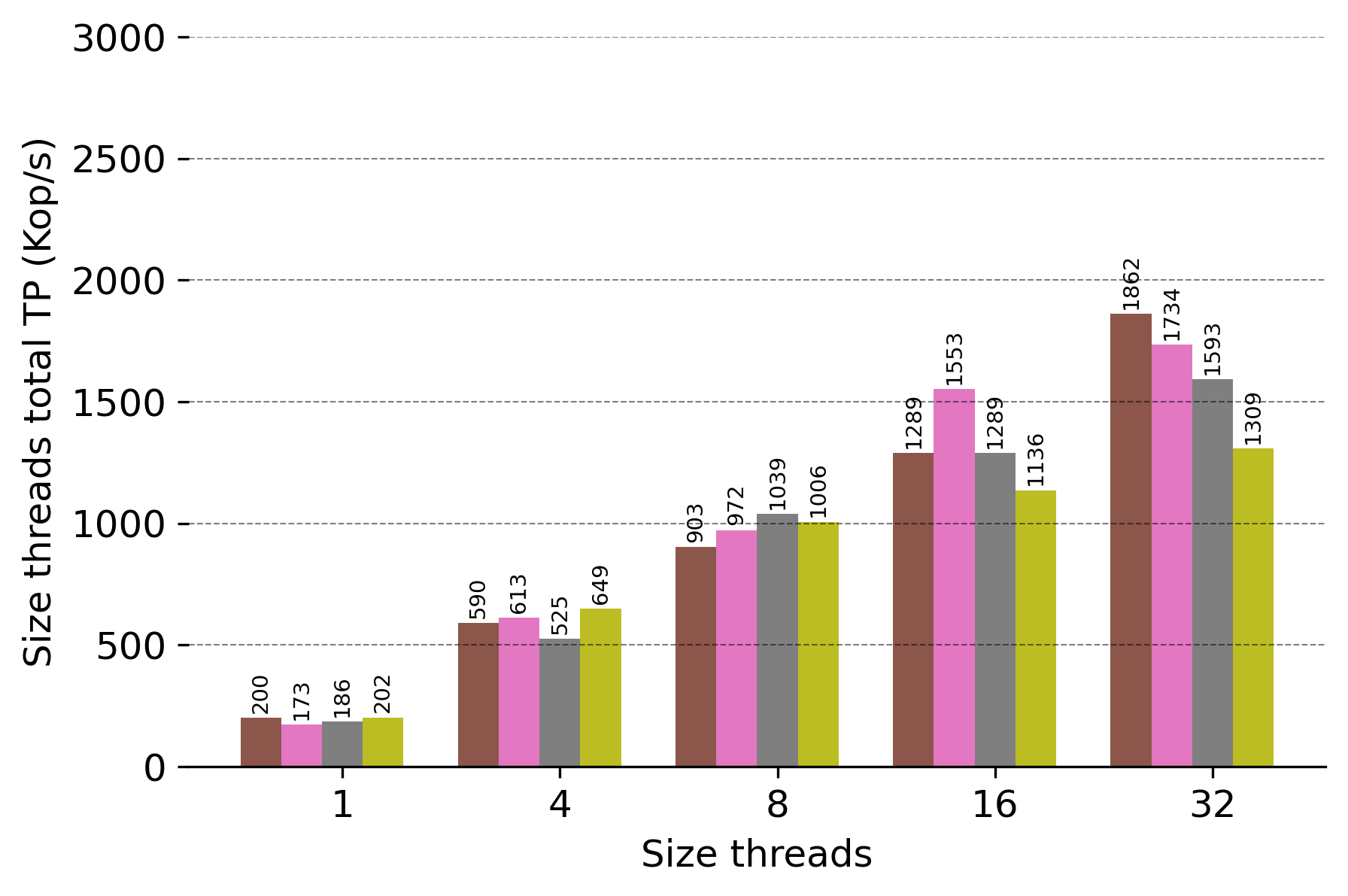}\hspace{2.5em}
	\includegraphics[width=.45\textwidth,trim={0 0 0 .1cm}]{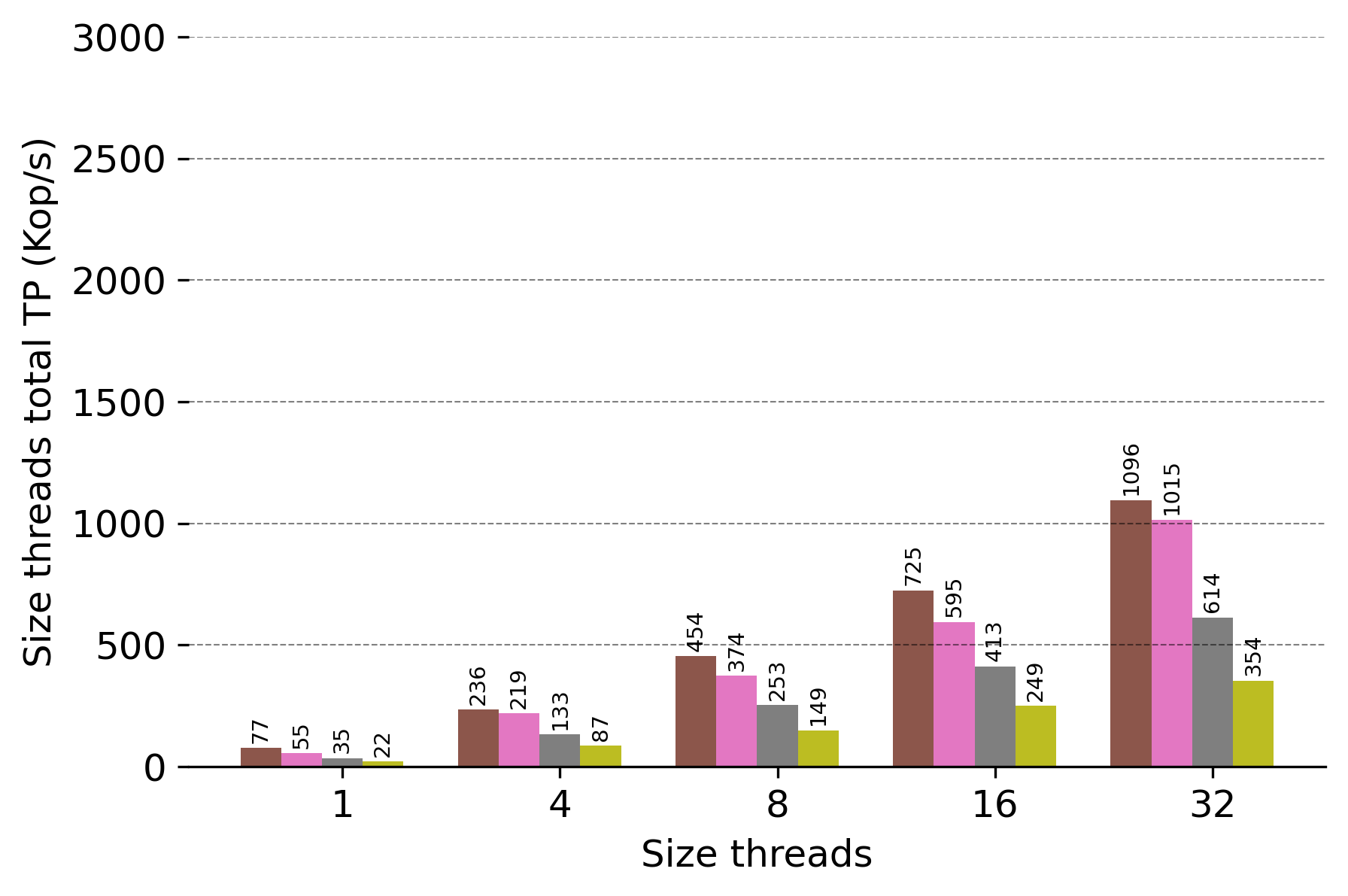}\vspace{-1.2em}
        \medskip
        \medskip
        \caption{MAX\_TRIES measurements in skip list}\par
	\label{fig:MAX_TRIES SL}
\end{figure*}

\begin{figure*}[htbp]
	\centering
	\medskip
	\textit{Read heavy}\hspace{10cm}
	\textit{Update heavy}\par
        \medskip
    \includegraphics[height=.025\textwidth]{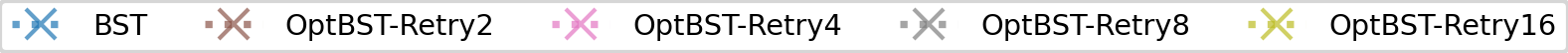}\vspace{0.5em}
\text{MAX\_TRIES effect on non-\size{} operations' performance in BST}\par
        \medskip
	\includegraphics[width=.45\textwidth,trim={0 0 0 .1cm}]{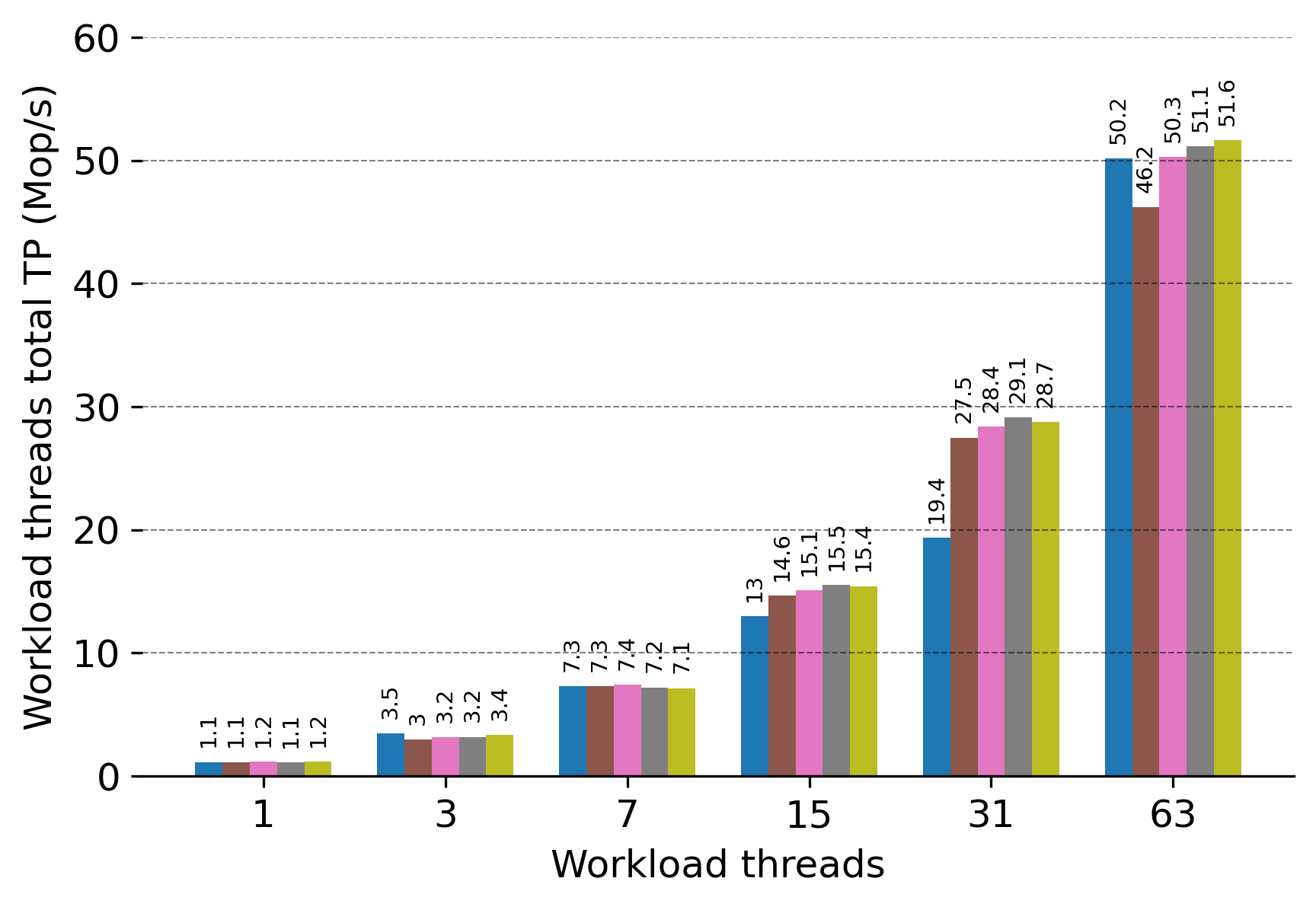}\hspace{2.5em}
	\includegraphics[width=.45\textwidth,trim={0 0 0 .1cm}]{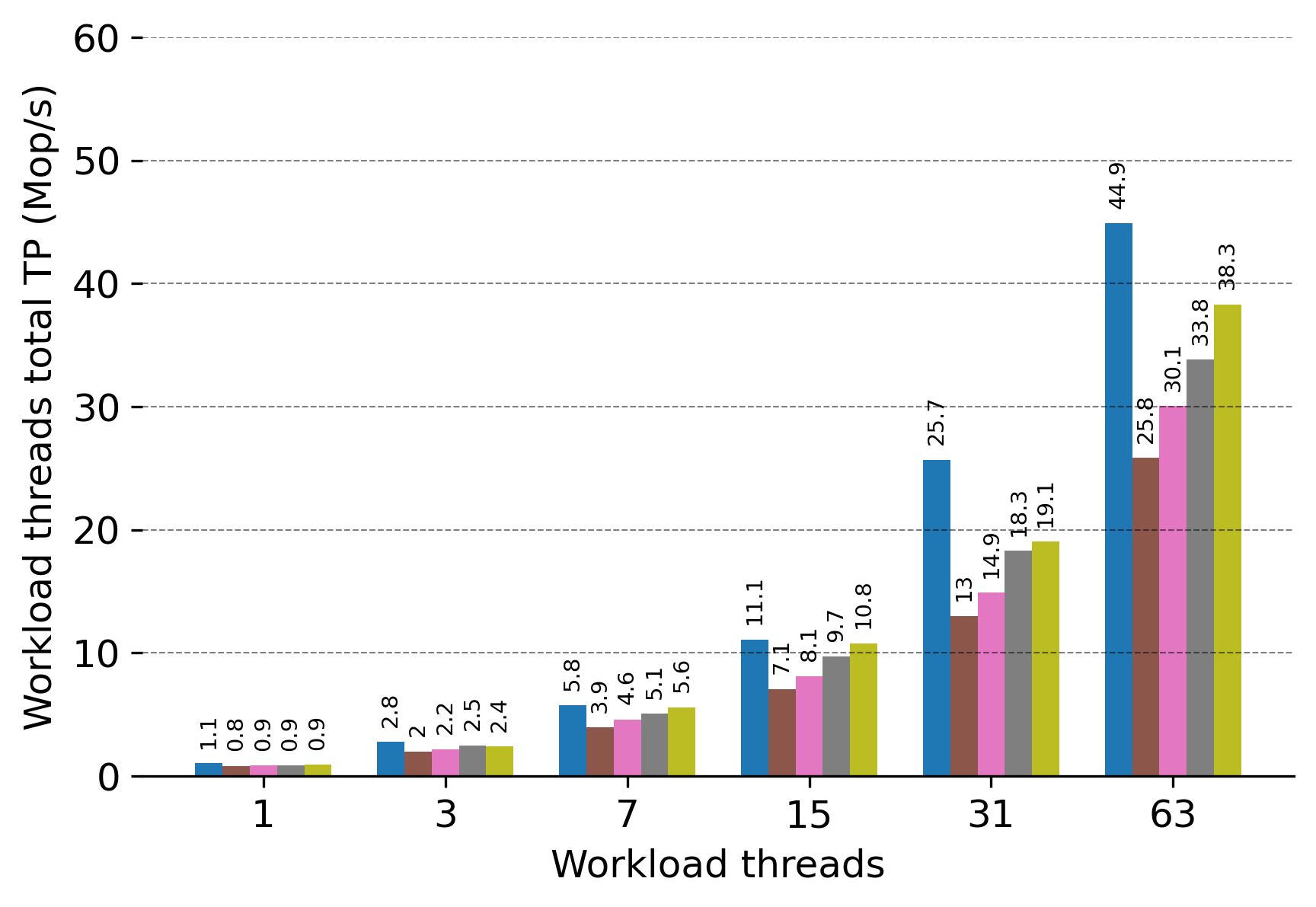}\vspace{-1.2em}

        \vspace{1.2em}
        \medskip
	\text{MAX\_TRIES effect on \size{} operations' performance in BST}\par
        \medskip
	\includegraphics[width=.45\textwidth,trim={0 0 0 .1cm}]{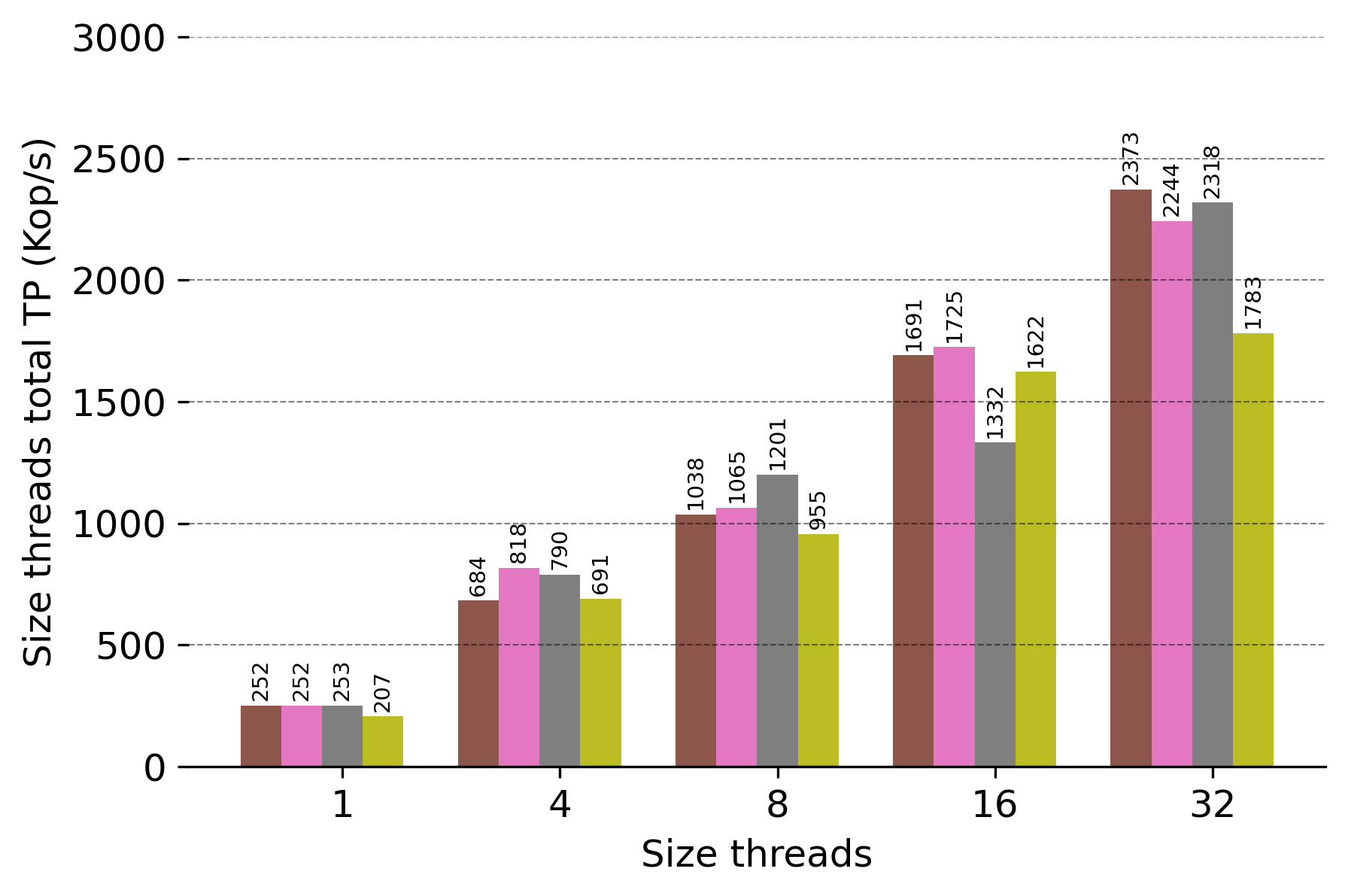}\hspace{2.5em}
	\includegraphics[width=.45\textwidth,trim={0 0 0 .1cm}]{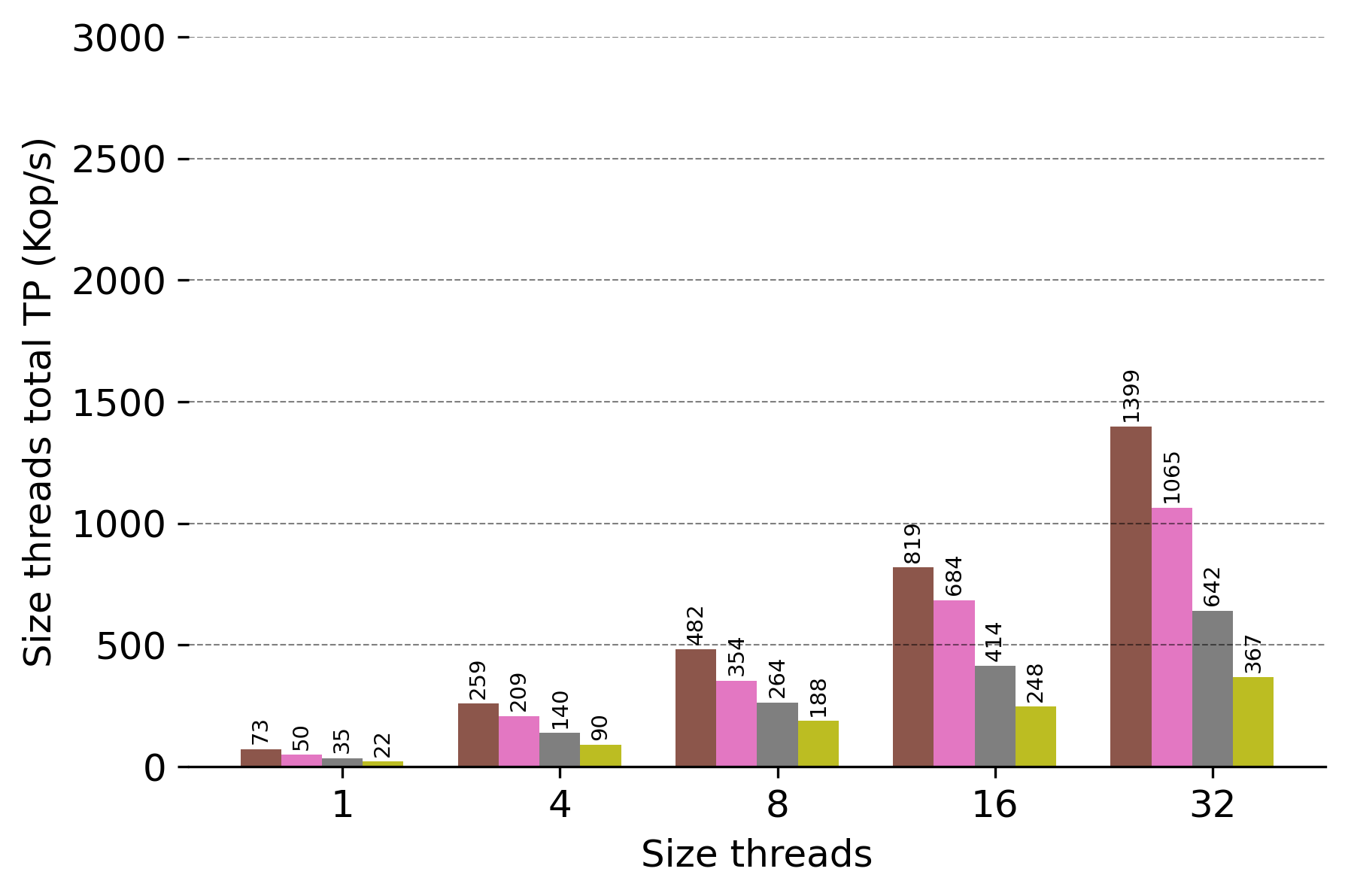}\vspace{-1.2em}
	\caption{MAX\_TRIES measurements in BST}
	\label{fig:MAX_TRIES BST}
\end{figure*}

\begin{figure*}[htbp]
	\centering
	\medskip
	\textit{Read heavy}\hspace{10cm}
	\textit{Update heavy}\par
	\medskip
	\includegraphics[height=.025\textwidth]{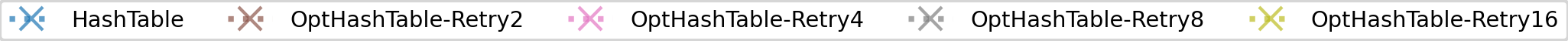}\vspace{0.5em}
    \text{MAX\_TRIES effect on non-\size{} operations' performance in hash table}\par
        \medskip
	\includegraphics[width=.45\textwidth,trim={0 0 0 .1cm}]{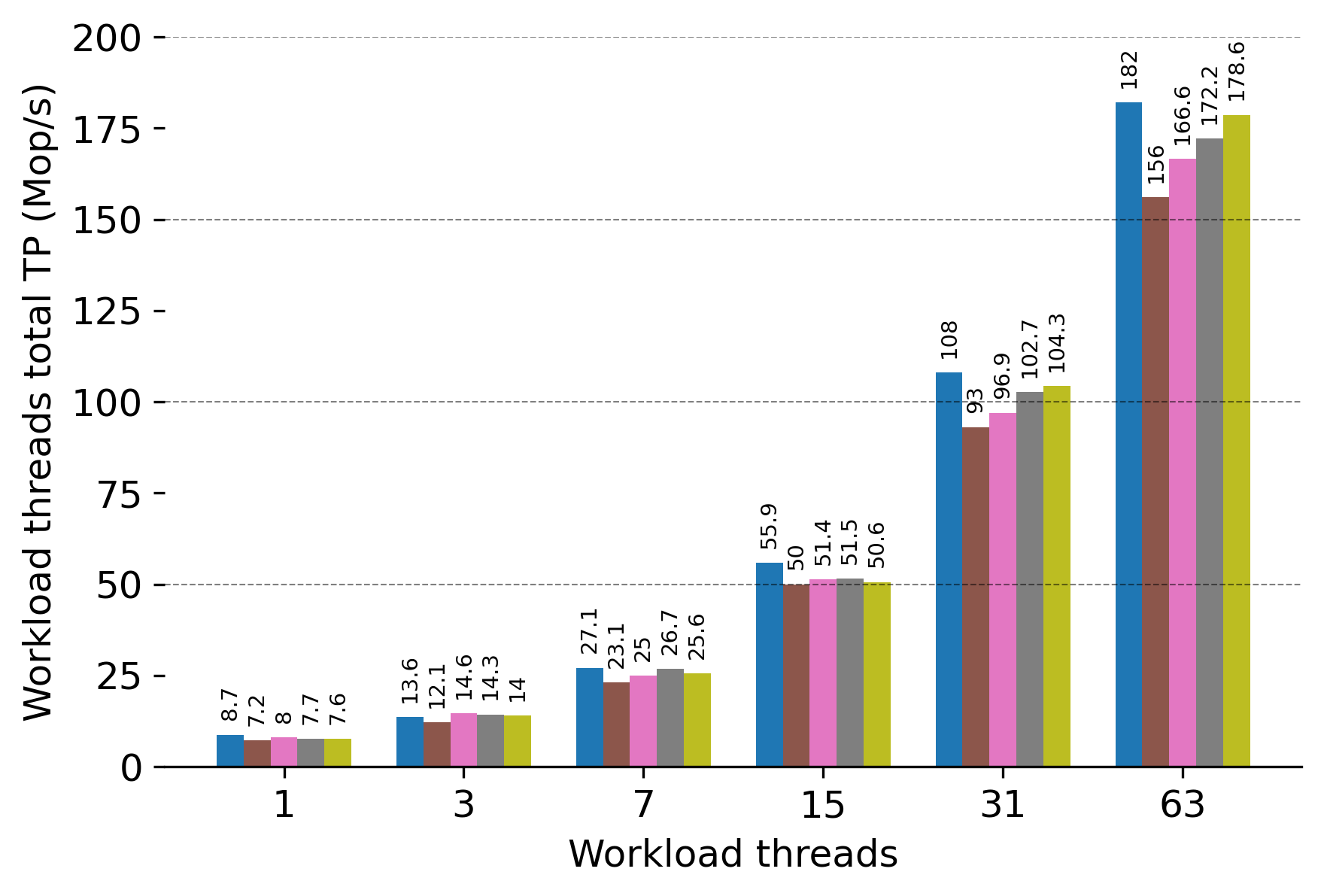}\hspace{2.5em}
	\includegraphics[width=.45\textwidth,trim={0 0 0 .1cm}]{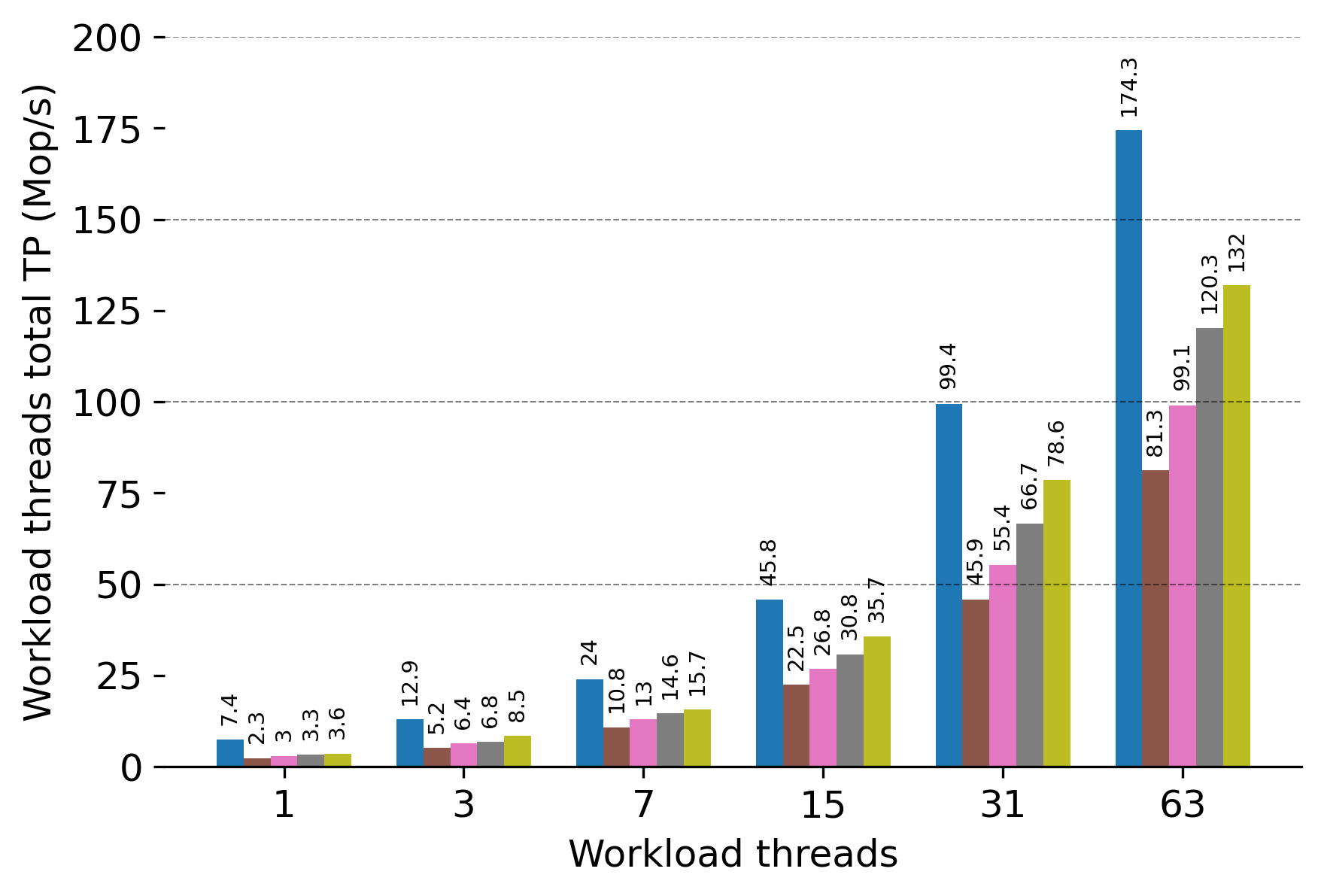}\vspace{-1.2em}

    	        \vspace{1.2em}
                \medskip
	\text{MAX\_TRIES effect on \size{} operations' performance in hash table}\par
        \medskip
	\includegraphics[width=.45\textwidth,trim={0 0 0 .1cm}]{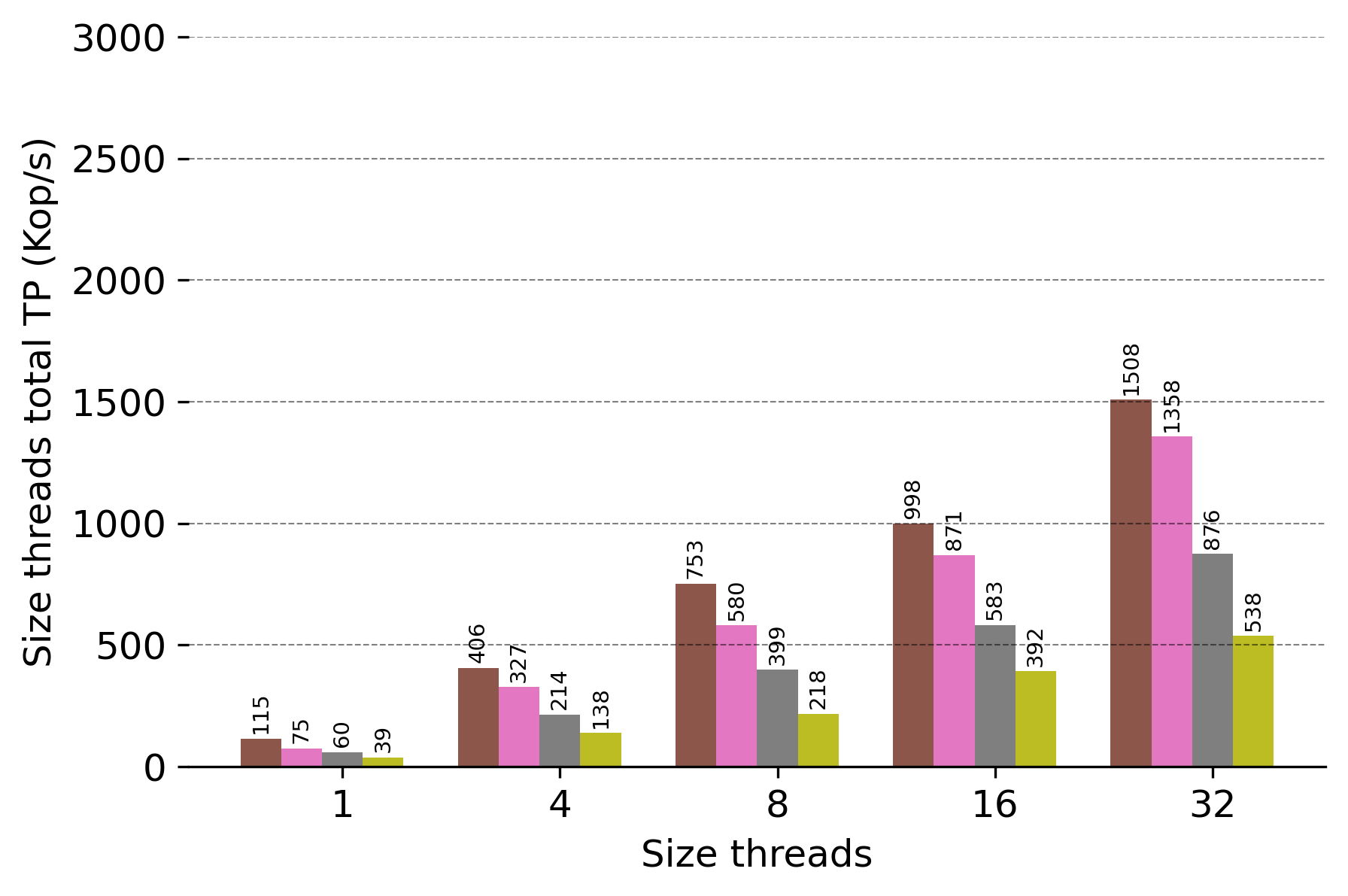}\hspace{2.5em}
	\includegraphics[width=.45\textwidth,trim={0 0 0 .1cm}]{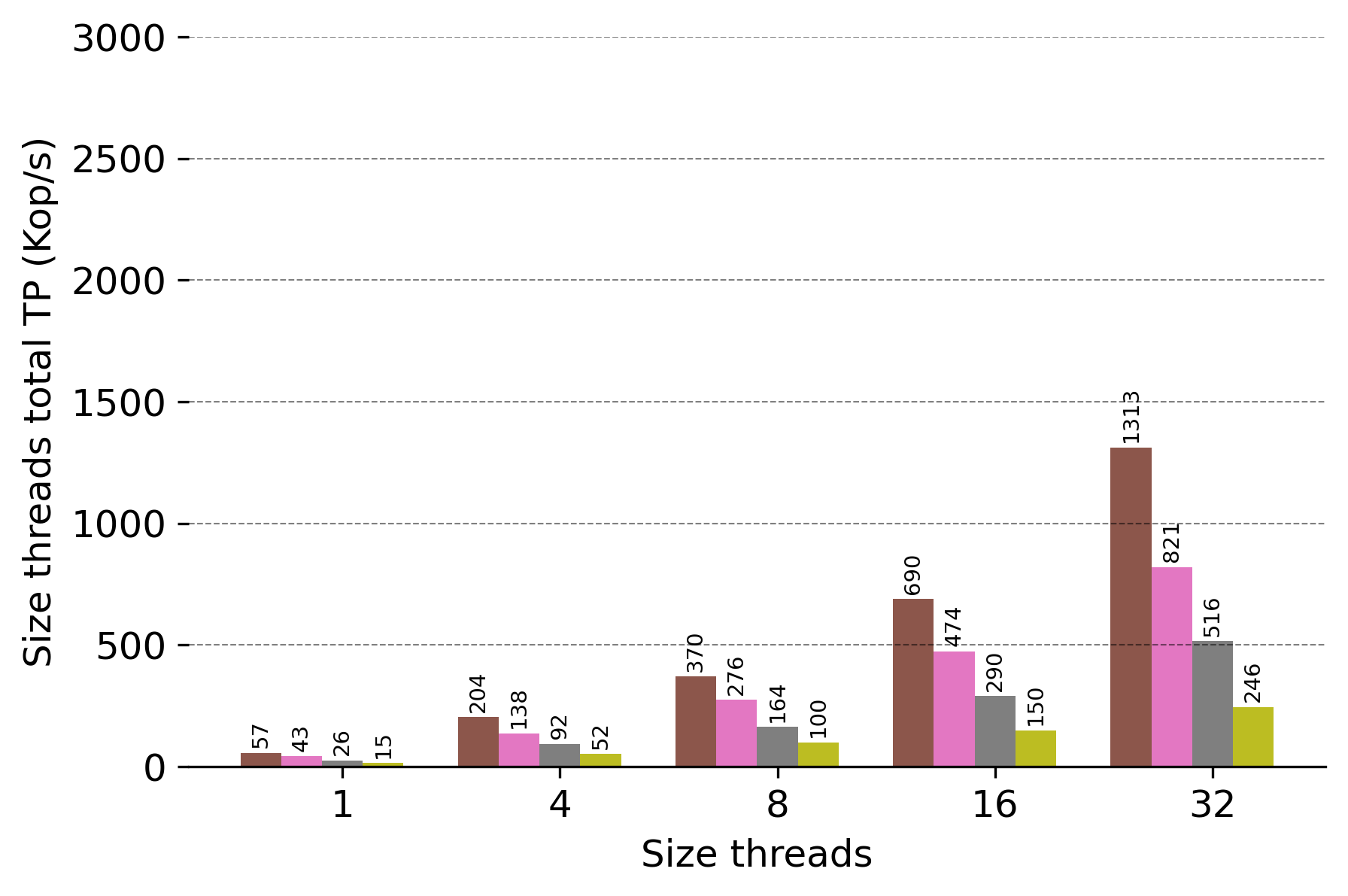}\vspace{-1.2em}
	\caption{MAX\_TRIES measurements in hash table}
	\label{fig:MAX_TRIES HT}
\end{figure*}

To assess the impact on the original data structure’s operations, we measured the total throughput of non-\size{} threads with a varying number of workload threads (non-\size{} threads) running alongside a single \size{} thread with a different \codestyle{MAX\_TRIES} value for each measurement. The \size{} thread continuously invokes the \size{} operation without any delay. 

The second measurement we conducted focused on the effect on performance of \size{} operations with respect to \codestyle{MAX\_TRIES}. In this experiment, the number of workload threads was held constant at 32, while the number of \size{} threads and \codestyle{MAX\_TRIES} were varied.

These graphs demonstrate a clear trend regarding the impact of the \codestyle{MAX\_TRIES} parameter on system performance. In general, increasing the \codestyle{MAX\_TRIES} value leads to improved performance for non-\size{} threads. This is because a larger \codestyle{MAX\_TRIES} value causes the \size{} threads to attempt the optimistic \size{} computation for a longer period before interrupting the non-\size{} threads to request assistance with calculating the data structure’s size. As a result, the non-\size{} threads experience fewer interruptions, allowing them to execute their operations with less interference.

Conversely, the performance of \size{} threads tends to degrade as \codestyle{MAX\_TRIES} increases. This degradation occurs because, under high contention — when many non-\size{} threads are actively operating on the data structure — optimistic size computations are more likely to fail. When the \codestyle{MAX\_TRIES} value is large, \size{} threads spend more time unsuccessfully retrying the optimistic approach before ultimately requesting assistance from non-\size{} threads. This extended period of unsuccessful retries increases the overall latency of the \size{} operation itself.

Although this general trend holds in most cases, the data also reveals certain exceptions, particularly in some read-heavy workloads where larger \codestyle{MAX\_TRIES} values lead to better performance for the \size{} operation as well. This happens because in read-heavy conditions, the optimistic size computation has a higher probability of succeeding without assistance. Consequently, in some cases, the overhead introduced by interrupting non-\size{} threads to help with size calculation outweighs the cost of simply retrying the optimistic approach several times. This trade-off highlights the importance of workload characteristics when tuning the \codestyle{MAX\_TRIES} parameter.

\newpage
\section{Zipfian measurements}\label{sec:zipfian graphs}
This section presents an additional evaluation using a Zipfian-distributed workload to examine the robustness of our findings under non-uniform access patterns. The results are provided in \Cref{fig:SkipList zipfian overhead,fig:BST zipfian overhead,fig:HashTable zipfian overhead}. As described in \Cref{sec-zipfian}, only \contains{} operations were modified to select keys according to a Zipfian distribution, while \ins{} and \del{} operations continued to use keys drawn uniformly at random. We employed the \texttt{ScrambledZipfianGenerator} from the YCSB suite~\cite{cooper2010benchmarking}, with a skew parameter $\theta = 0.99$. This generator avoids biasing access patterns toward low-value keys, which is particularly important for ordered data structures such as skip lists and binary search trees. All other experimental parameters remained consistent with those used in the uniform workload experiments.

As \Cref{fig:SkipList zipfian overhead,fig:BST zipfian overhead,fig:HashTable zipfian overhead} show, the performance trends and relative rankings of the synchronization methods under Zipfian access are largely consistent with the results obtained under uniform access. This suggests that the evaluated methodologies exhibit stable behavior under skewed key access patterns, as commonly encountered in real-world systems.

\clearpage
\vspace*{\fill}
\begin{figure*}[h]
	\centering
	\medskip
	\textit{Read heavy}\hspace{0.3em}
	\includegraphics[height=.02\textwidth]{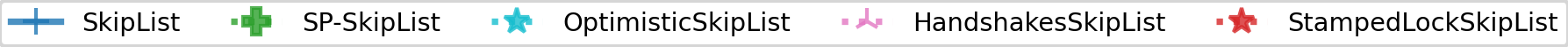}\hspace{0.3em}
	\textit{Update heavy}\par
	\medskip
	\text{Without a concurrent \size{} thread}\par
        \smallskip
	\includegraphics[width=.45\textwidth,trim={0 0 0 .1cm}]{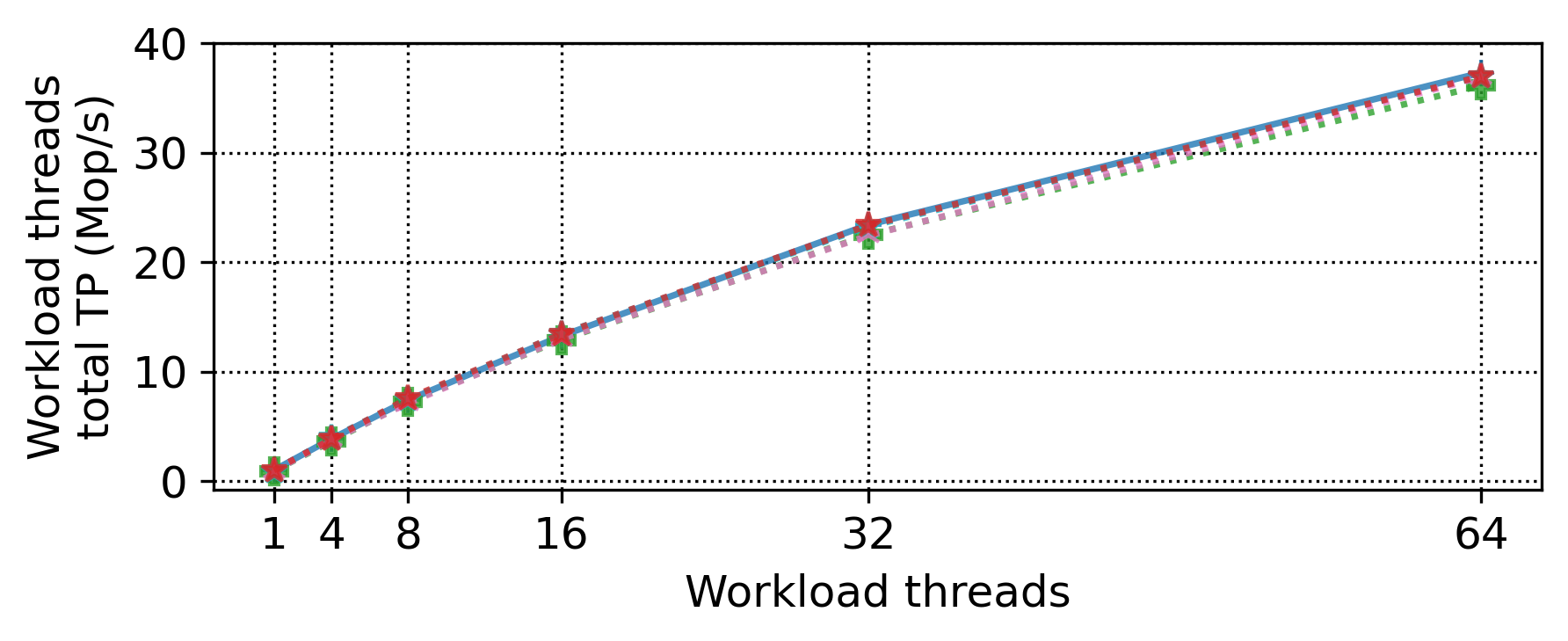}\hspace{2.5em}
	\includegraphics[width=.45\textwidth,trim={0 0 0 .1cm}]{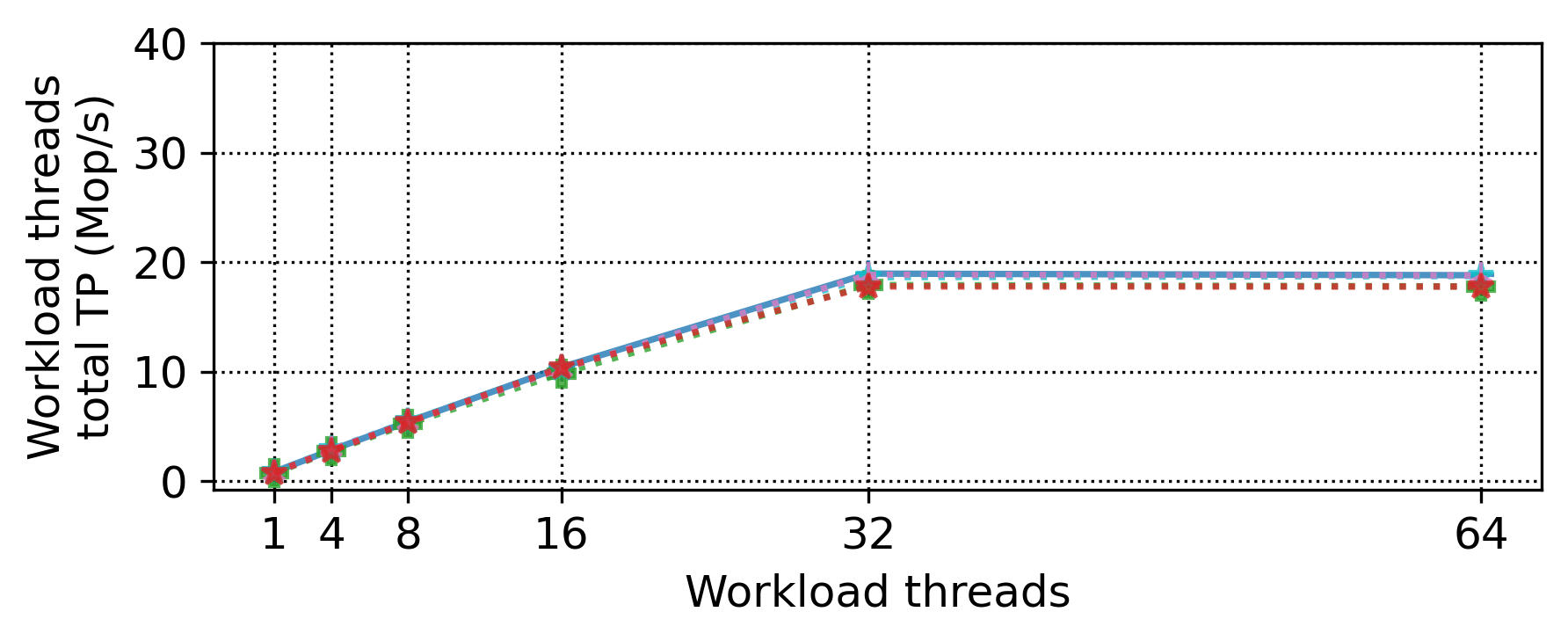}\par
	\includegraphics[width=.45\textwidth,trim={0 0 0 .2cm}]{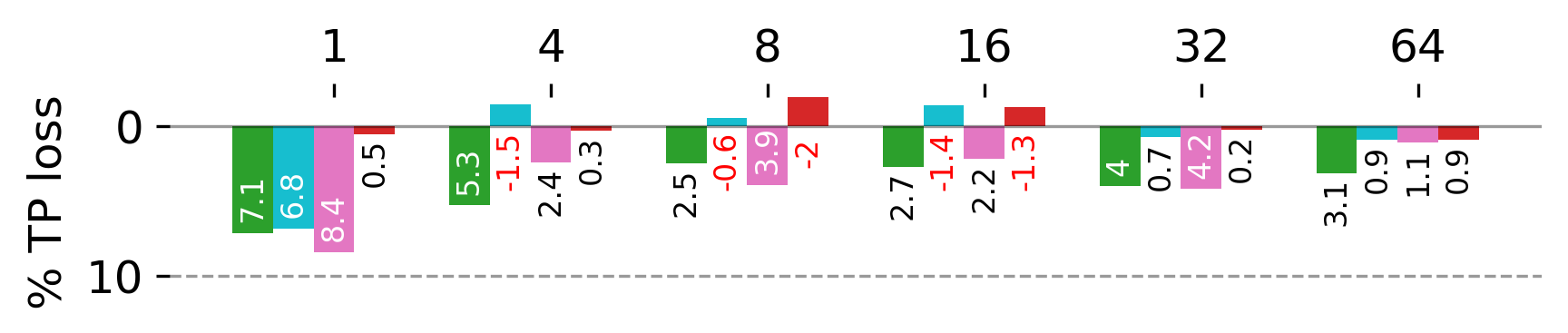}\hspace{2.5em}
	\includegraphics[width=.45\textwidth,trim={0 0 0 .2cm}]{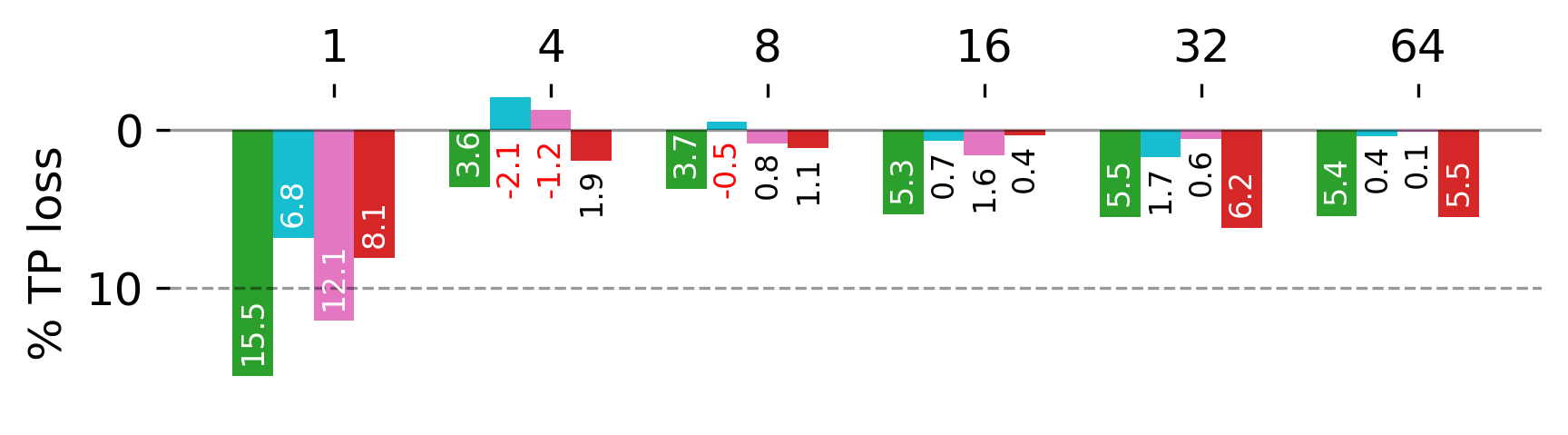}\par
	\medskip
	\text{With a concurrent \size{} thread and no delay}\par
	\includegraphics[width=.45\textwidth,trim={0 0 0 .1cm}]{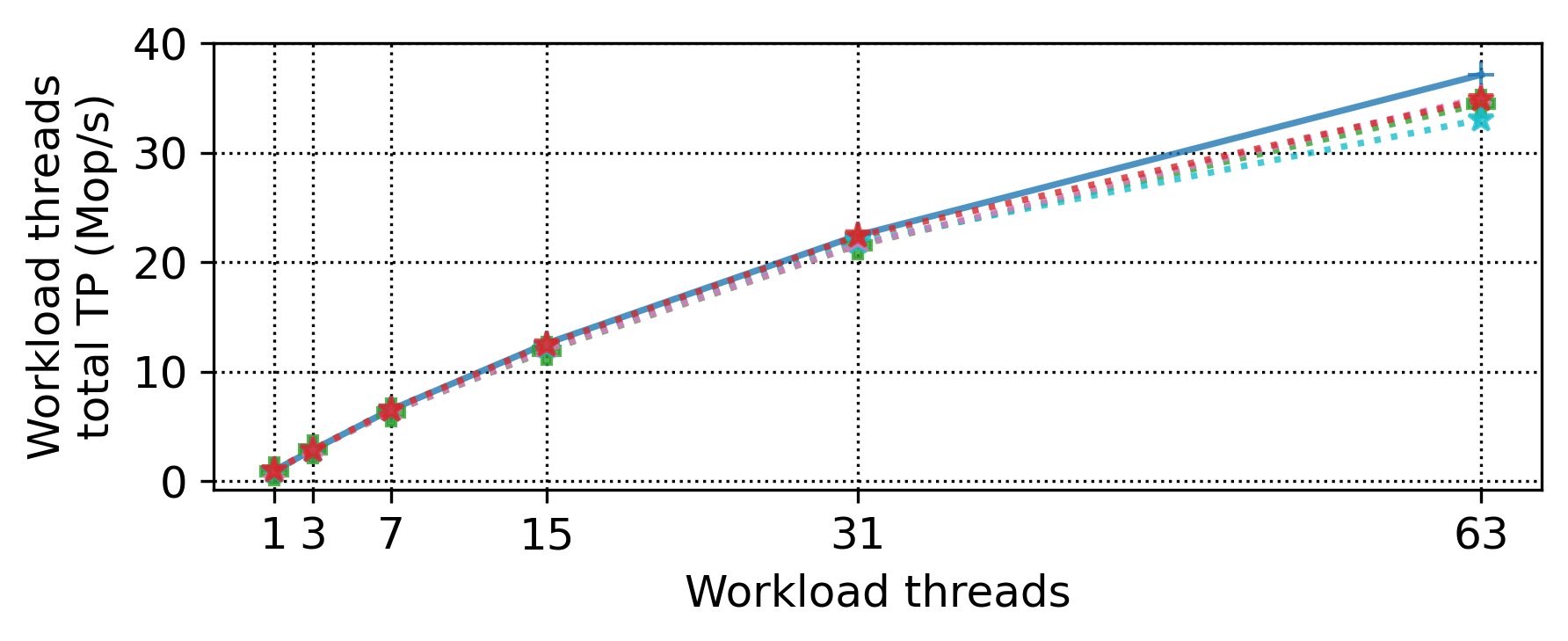}\hspace{2.5em}
	\includegraphics[width=.45\textwidth,trim={0 0 0 .1cm}]{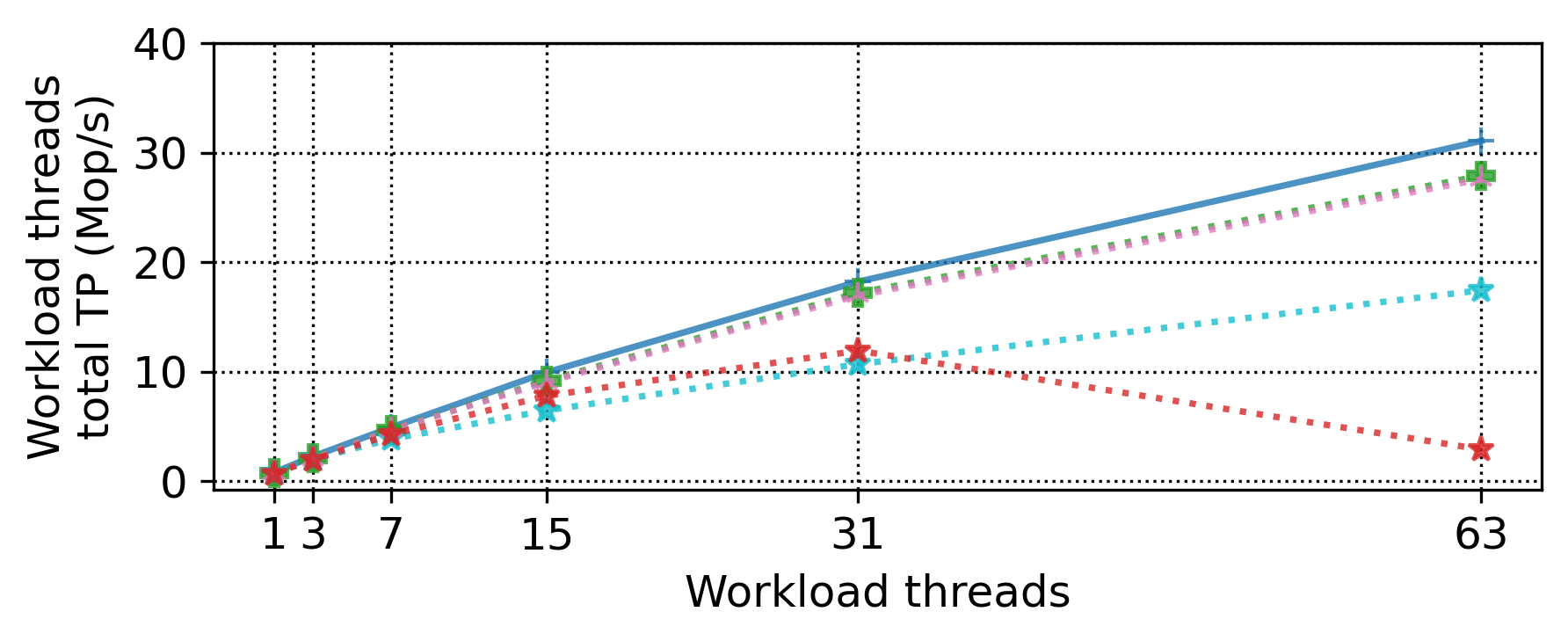}\par
	\includegraphics[width=.45\textwidth,trim={0 0 0 .2cm}]{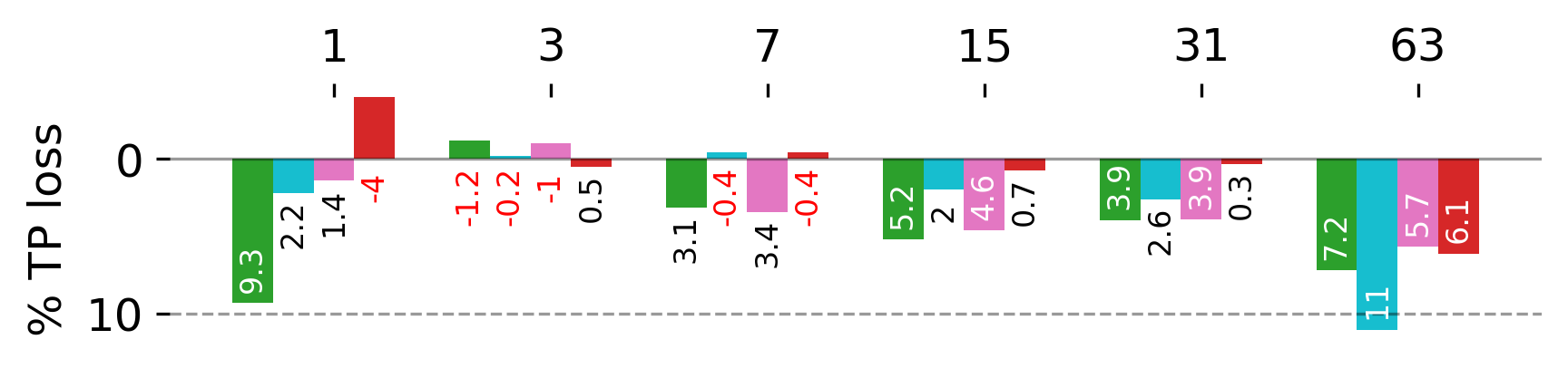}\hspace{2.5em}
	\includegraphics[width=.45\textwidth,trim={0 0 0 .2cm}]{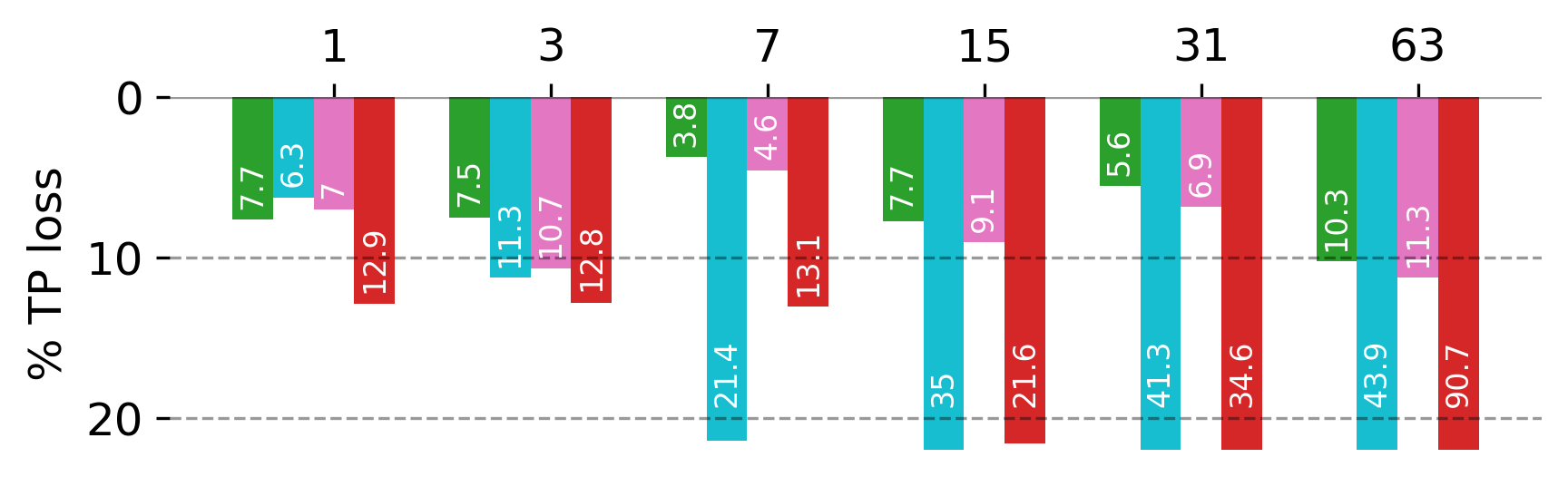}\par
	\medskip
	\text{With a concurrent \size{} thread and 700 \si{\micro\second} delay}\par
	\includegraphics[width=.45\textwidth,trim={0 0 0 .1cm}]{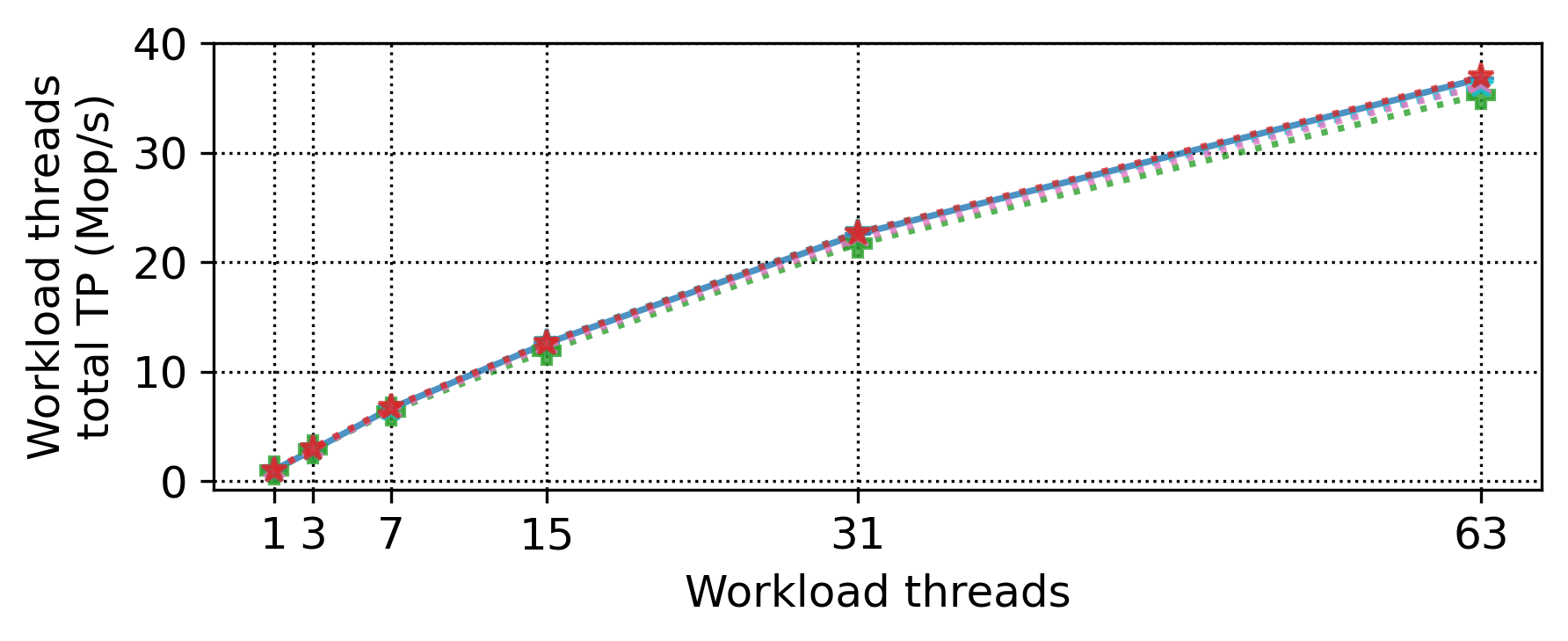}\hspace{2.5em}
	\includegraphics[width=.45\textwidth,trim={0 0 0 .1cm}]{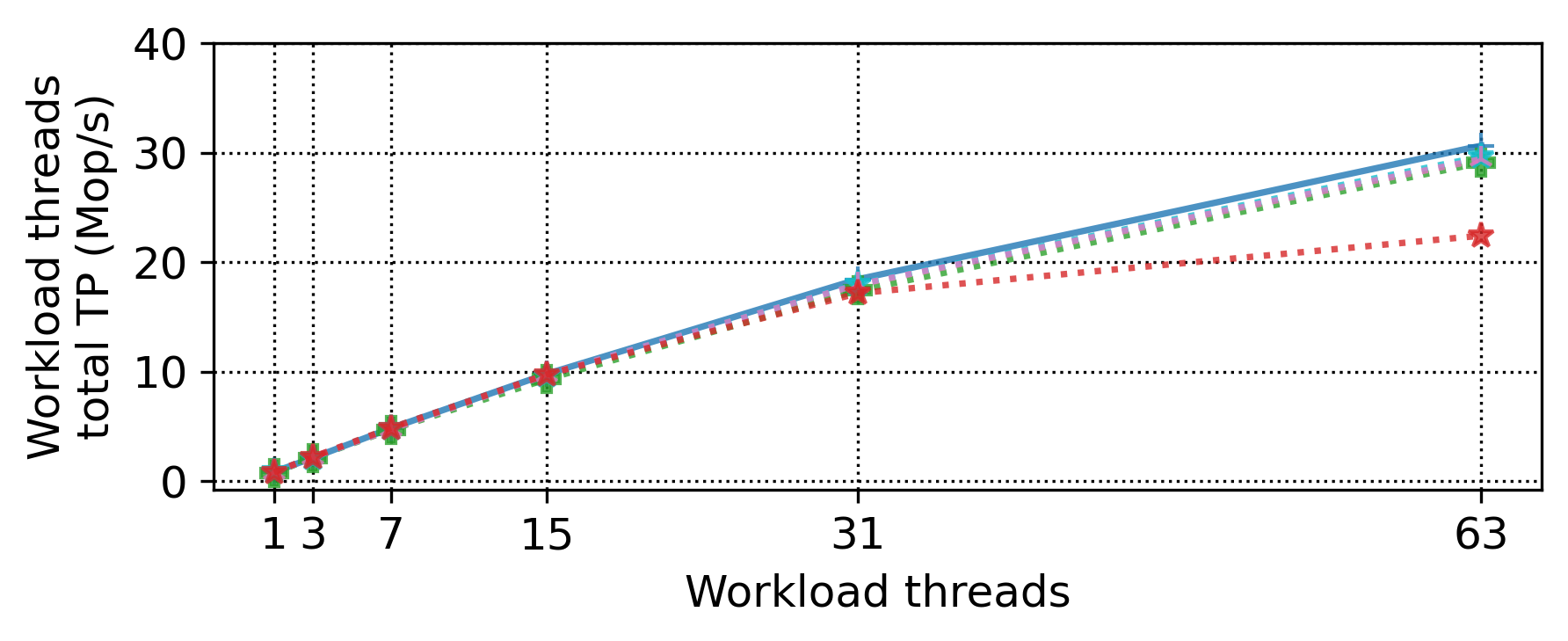}\par
	\includegraphics[width=.45\textwidth,trim={0 0 0 .2cm}]{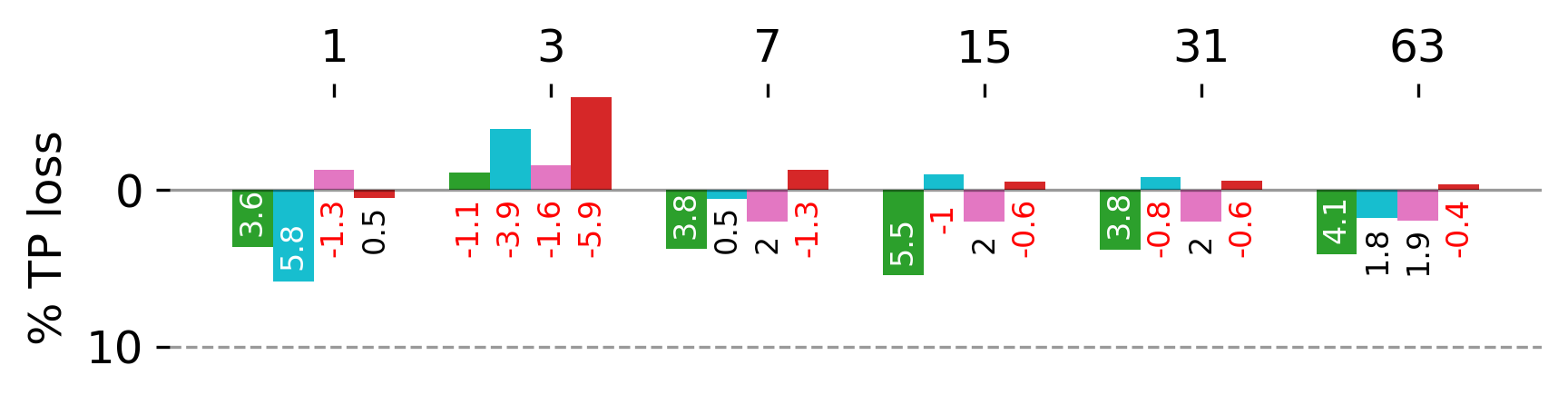}\hspace{2.5em}
	\includegraphics[width=.45\textwidth,trim={0 0 0 .2cm}]{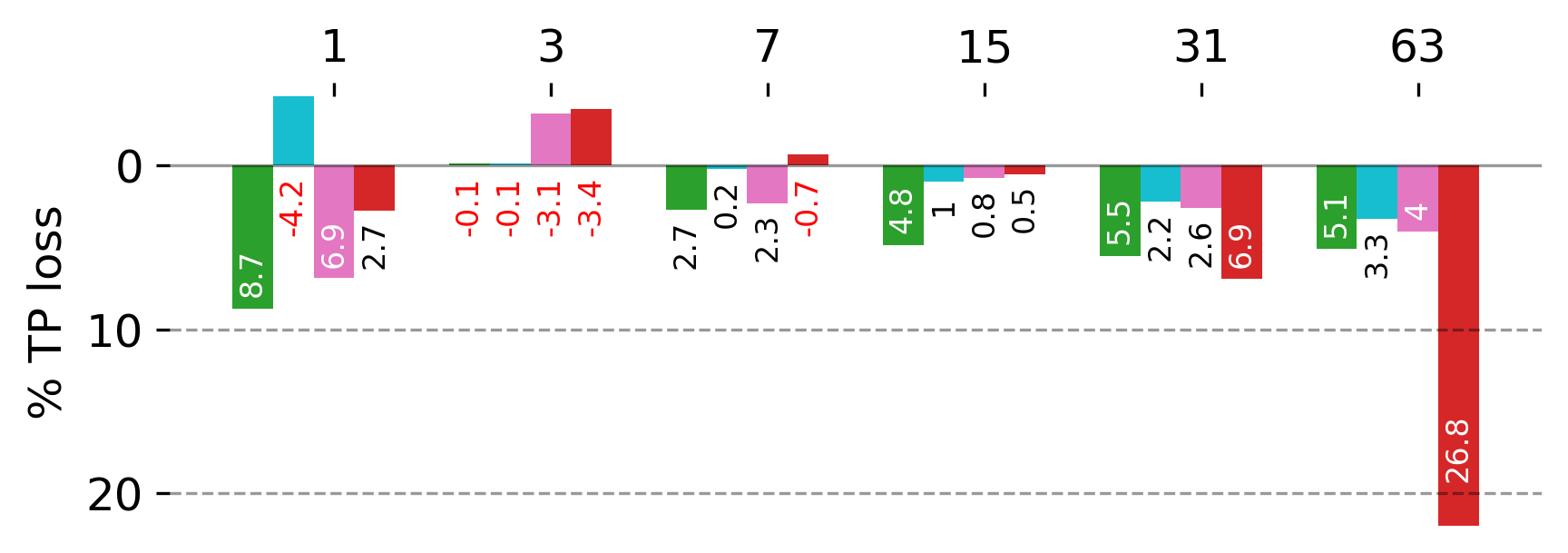}\par
	\caption{Overhead on skip list operations with Zipfian-distributed \contains{}}
	\label{fig:SkipList zipfian overhead}
\end{figure*}
\vspace*{\fill}
\clearpage

\clearpage
\vspace*{\fill}
\begin{figure*}[htbp]
	\centering
	\medskip
	\textit{Read heavy}\hspace{2.3em}
	\includegraphics[height=.02\textwidth]{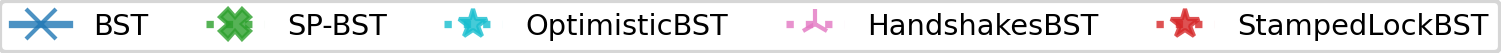}\hspace{2.3em}
	\textit{Update heavy}\par
	\medskip
	\text{Without a concurrent \size{} thread}\par
        \smallskip
	\includegraphics[width=.45\textwidth,trim={0 0 0 .1cm}]{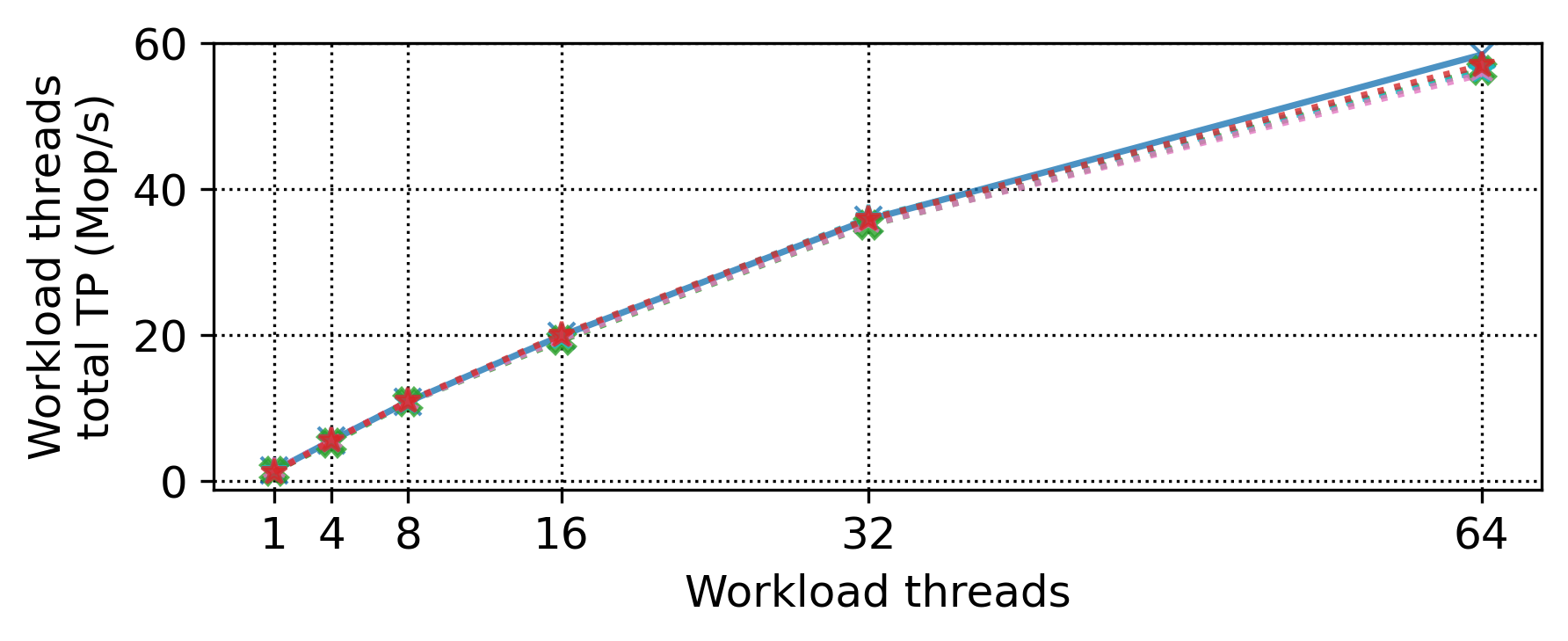}\hspace{2.5em}
	\includegraphics[width=.45\textwidth,trim={0 0 0 .1cm}]{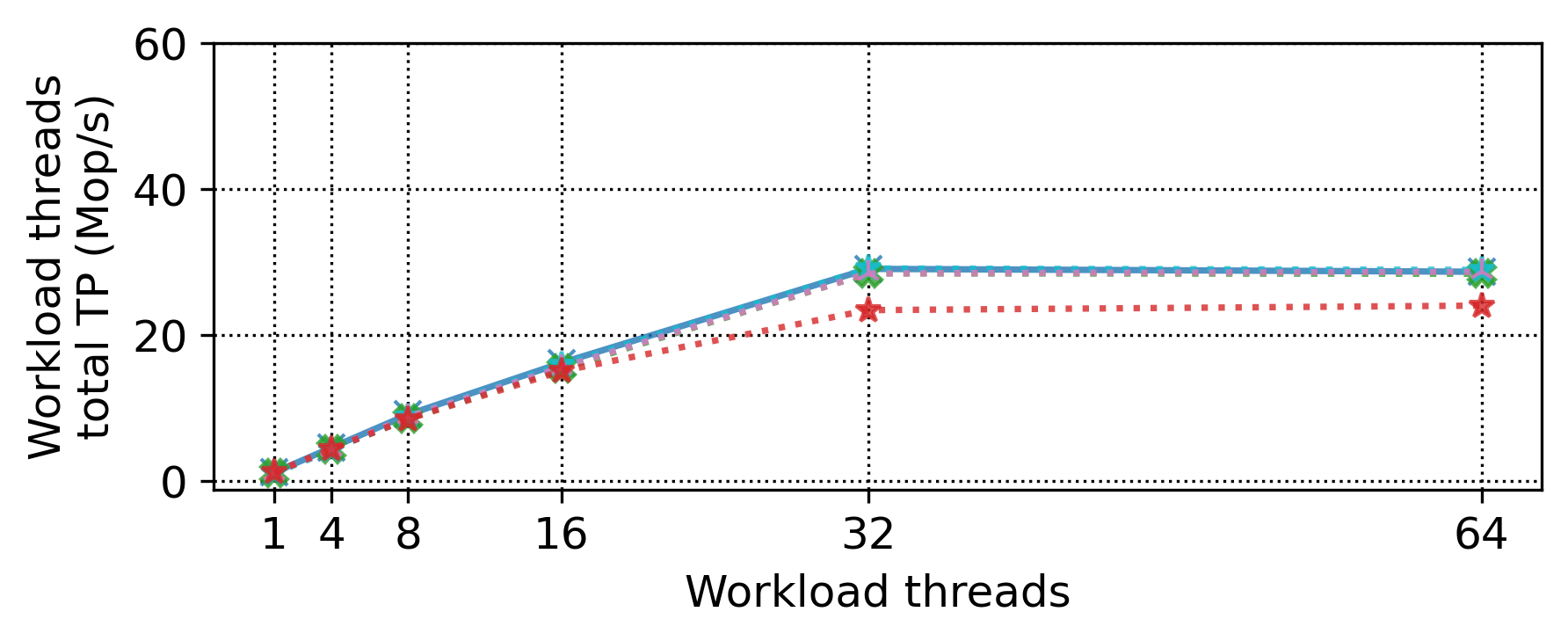}\par
	\includegraphics[width=.45\textwidth,trim={0 0 0 .2cm}]{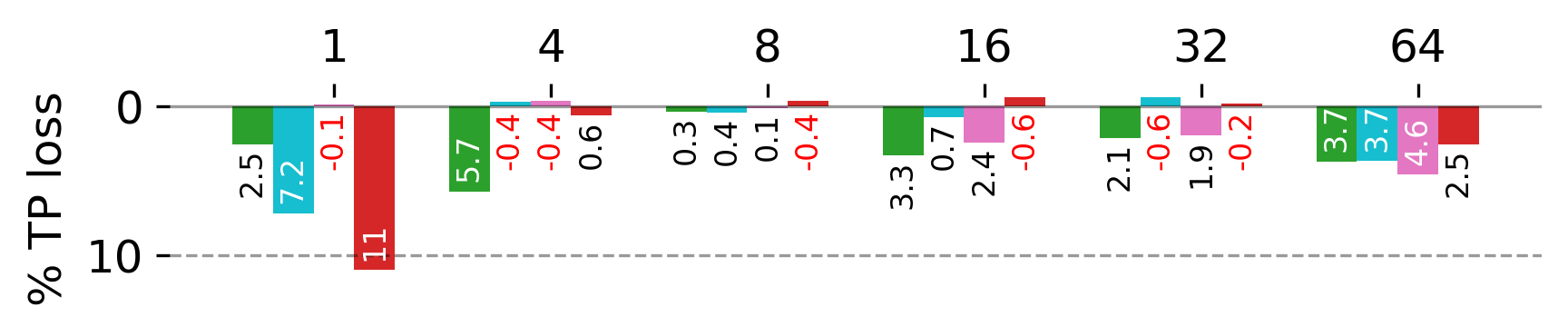}\hspace{2.5em}
	\includegraphics[width=.45\textwidth,trim={0 0 0 .2cm}]{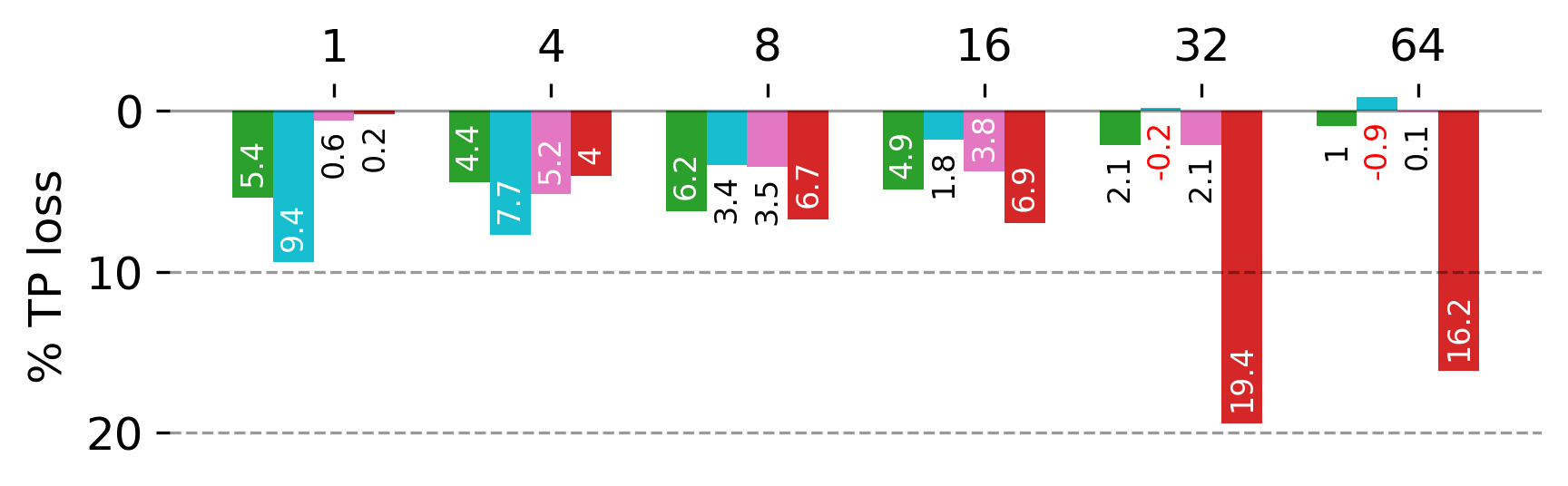}\par
	\medskip
	\text{With a concurrent \size{} thread and no delay}\par
	\includegraphics[width=.45\textwidth,trim={0 0 0 .1cm}]{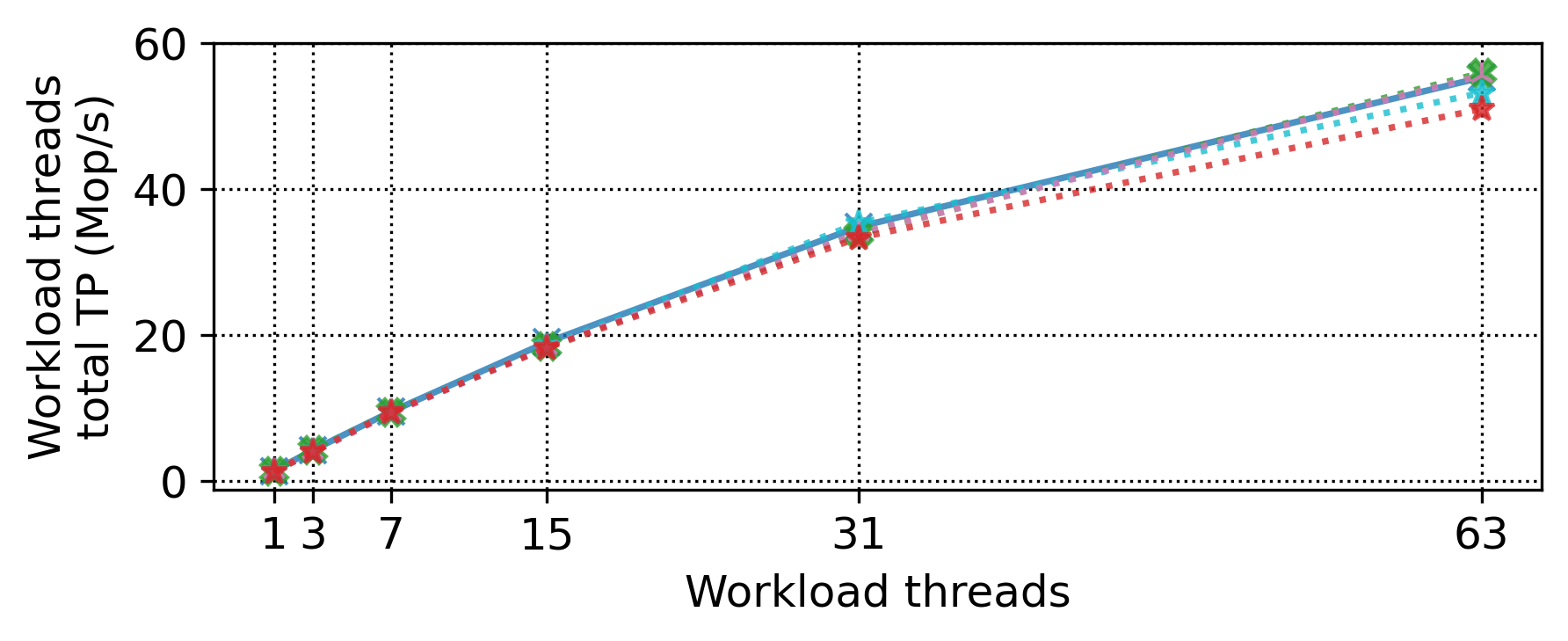}\hspace{2.5em}
	\includegraphics[width=.45\textwidth,trim={0 0 0 .1cm}]{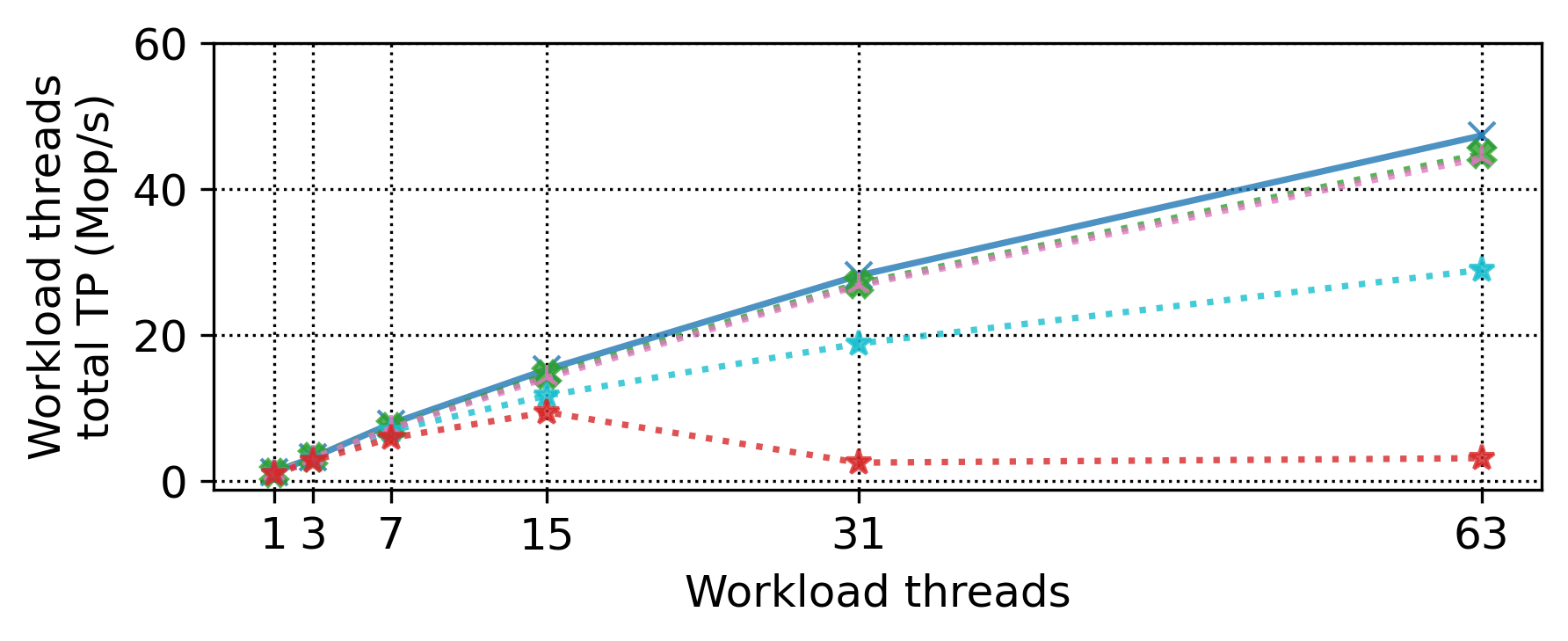}\par
	\includegraphics[width=.45\textwidth,trim={0 0 0 .2cm}]{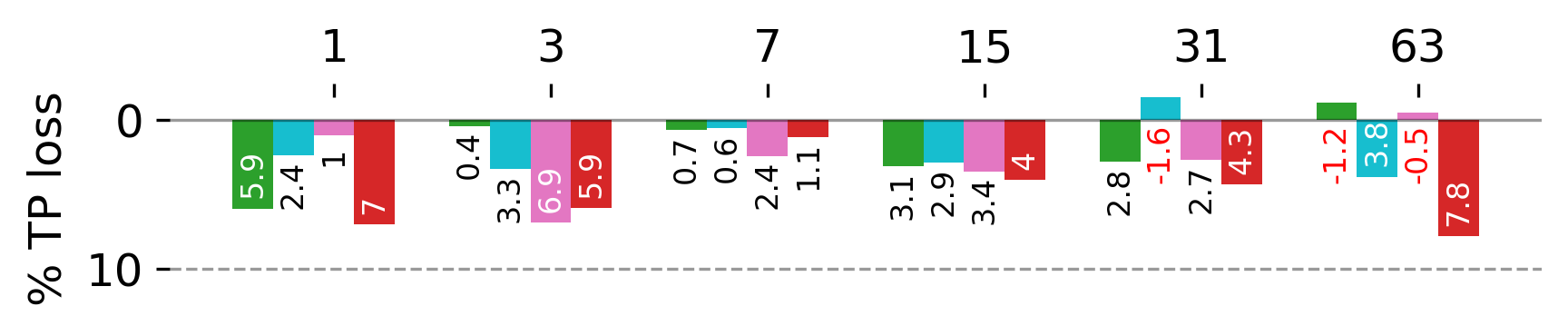}\hspace{2.5em}
	\includegraphics[width=.45\textwidth,trim={0 0 0 .2cm}]{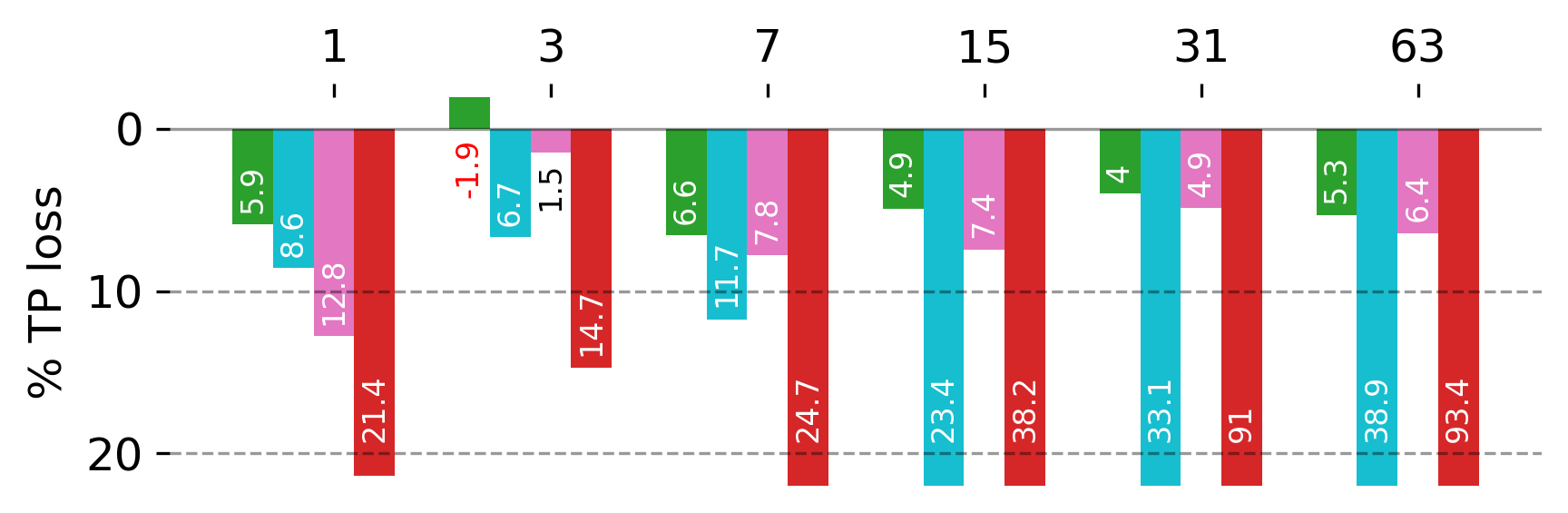}\par
	\medskip
	\text{With a concurrent \size{} thread and 700 \si{\micro\second} delay}\par
	\includegraphics[width=.45\textwidth,trim={0 0 0 .1cm}]{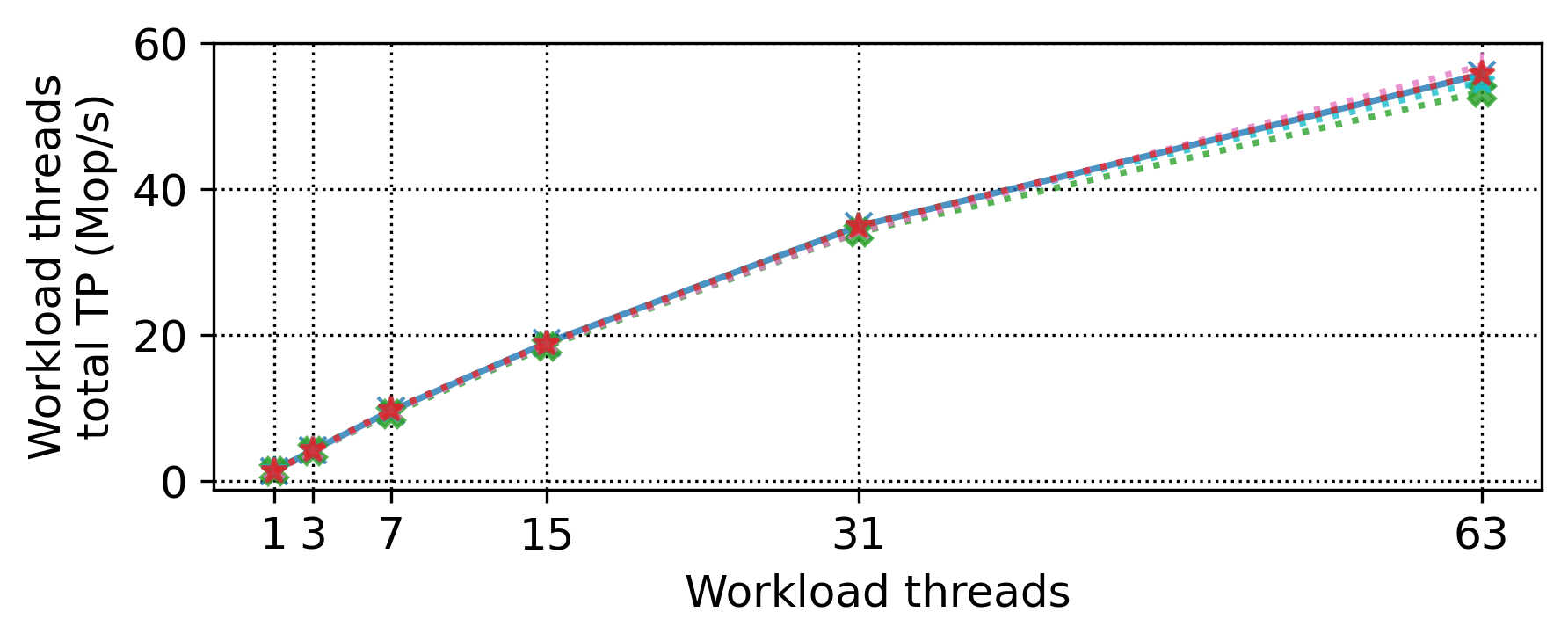}\hspace{2.5em}
	\includegraphics[width=.45\textwidth,trim={0 0 0 .1cm}]{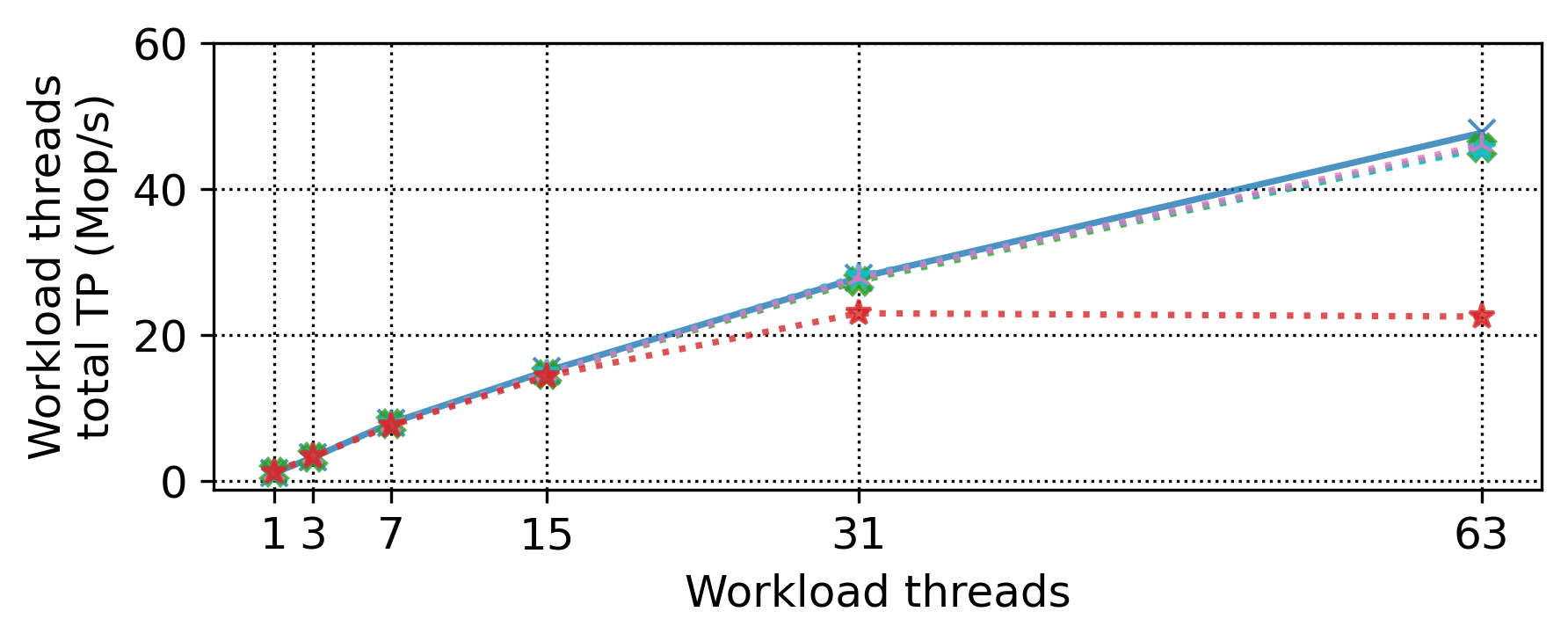}\par
	\includegraphics[width=.45\textwidth,trim={0 0 0 .2cm}]{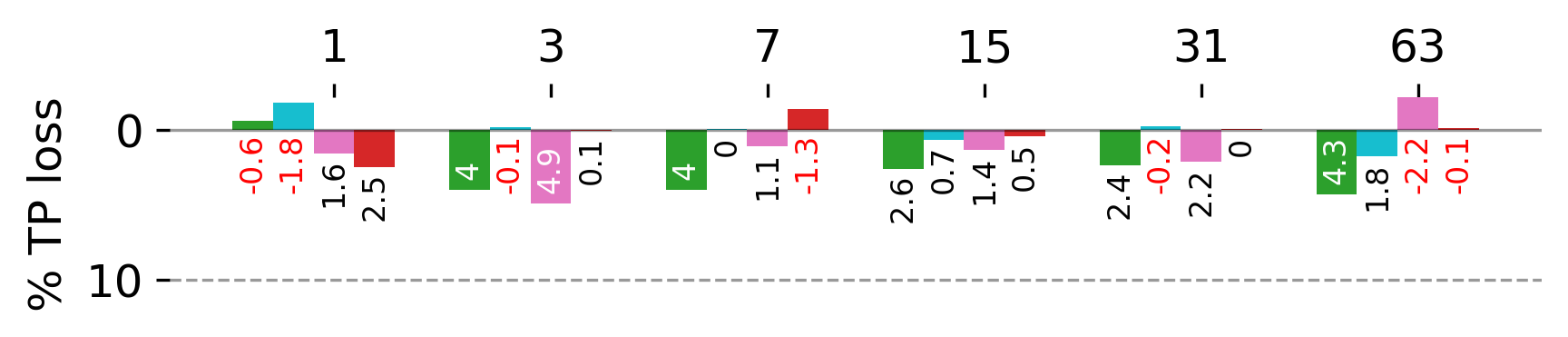}\hspace{2.5em}
	\includegraphics[width=.45\textwidth,trim={0 0 0 .2cm}]{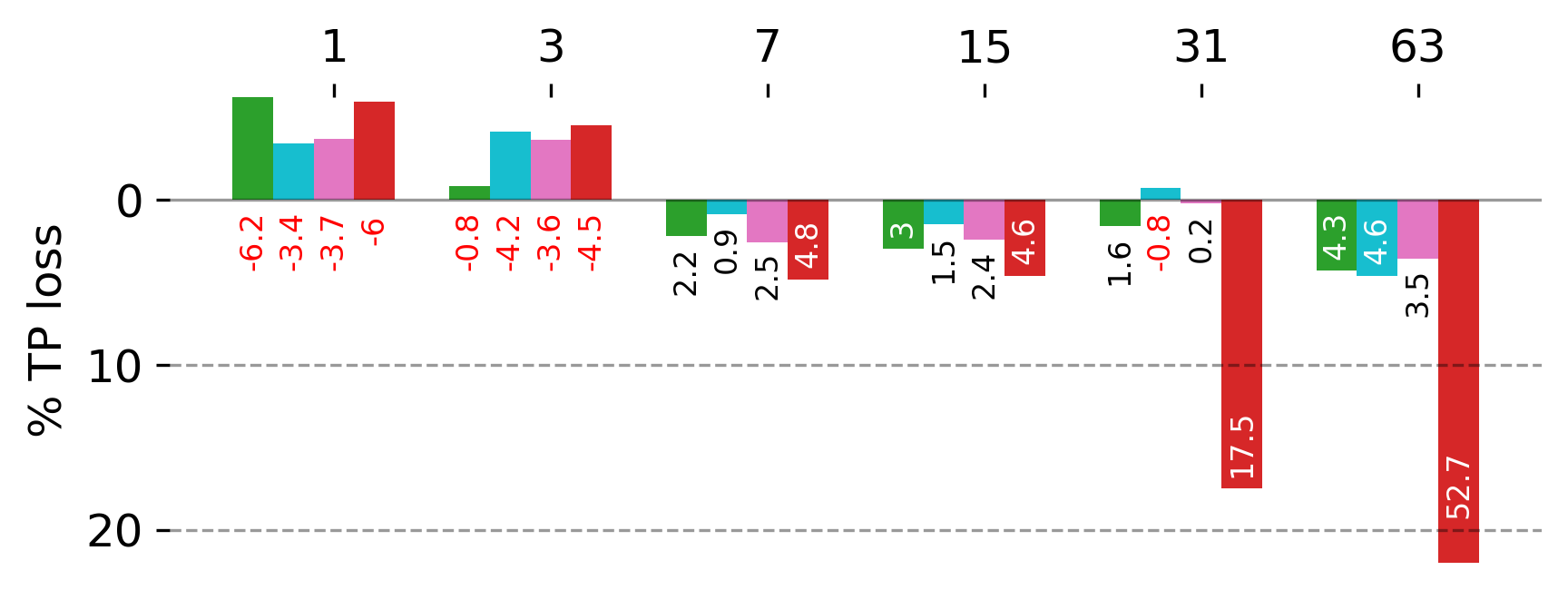}\par
	\caption{Overhead on BST operations with Zipfian-distributed \contains{}}
	\label{fig:BST zipfian overhead}
\end{figure*}
\vspace*{\fill}
\clearpage

\clearpage
\vspace*{\fill}
\begin{figure*}[htbp]
	\centering
	\medskip
	\textit{Read heavy}\hspace{0.3em}
	\includegraphics[height=.018\textwidth]{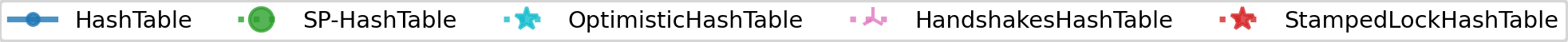}\hspace{0.3em}
	\textit{Update heavy}\par
	\medskip
	\text{Without a concurrent \size{} thread}\par
        \smallskip
	\includegraphics[width=.45\textwidth,trim={0 0 0 .1cm}]{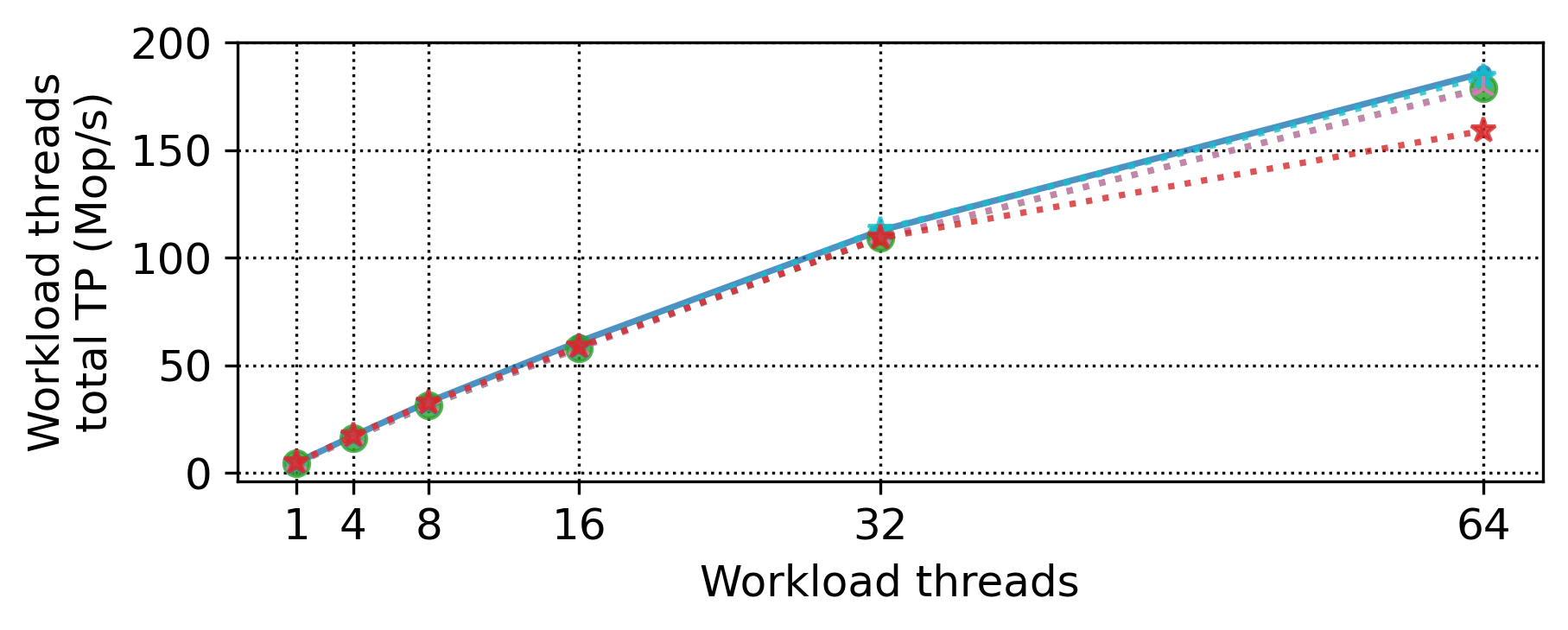}\hspace{2.5em}
	\includegraphics[width=.45\textwidth,trim={0 0 0 .1cm}]{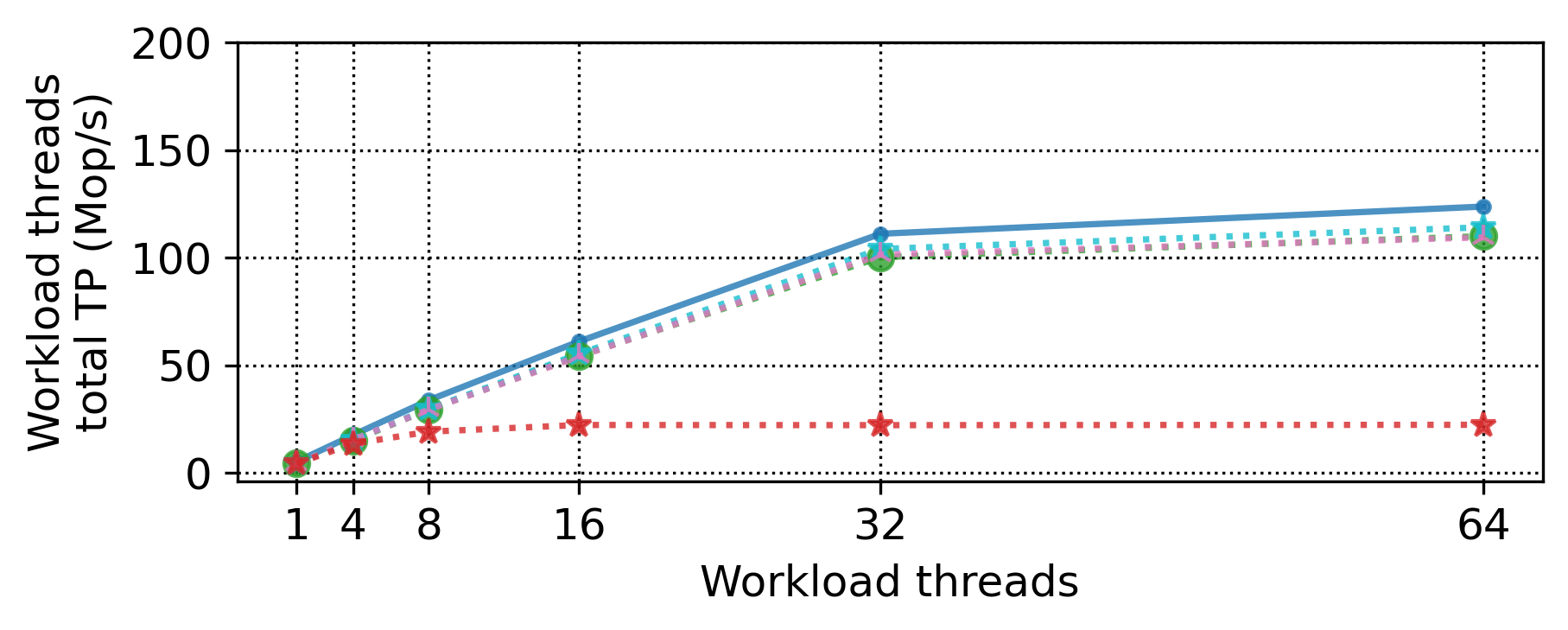}\par
	\includegraphics[width=.45\textwidth,trim={0 0 0 .2cm}]{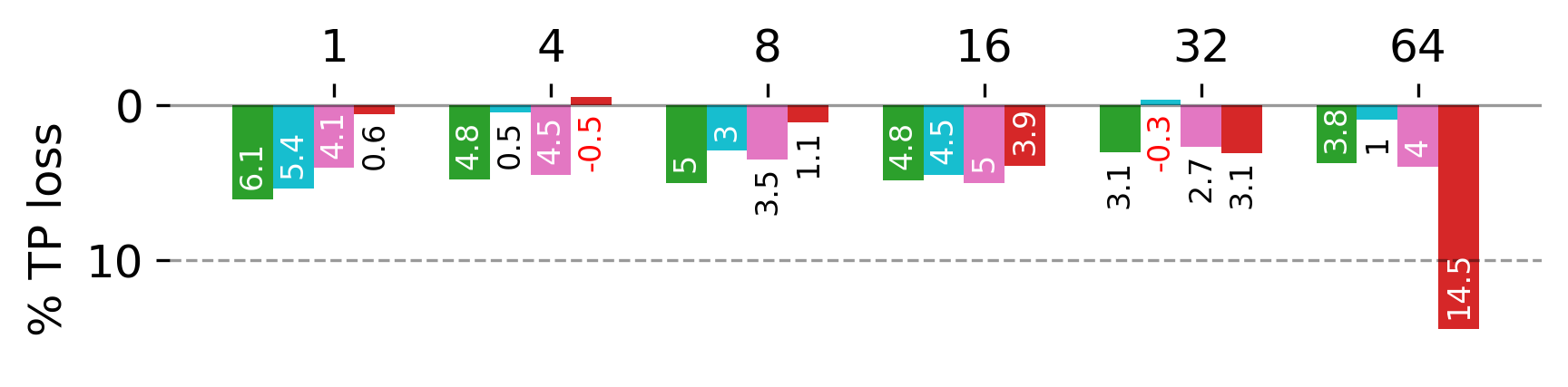}\hspace{2.5em}
	\includegraphics[width=.45\textwidth,trim={0 0 0 .2cm}]{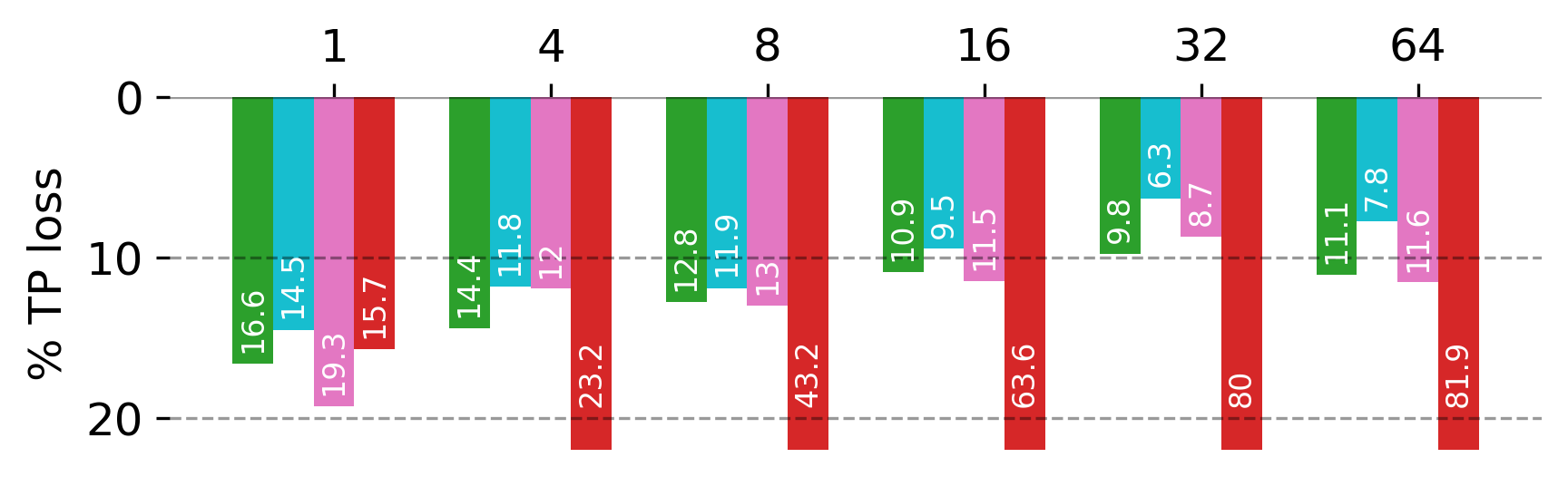}\par
	\medskip
	\text{With a concurrent \size{} thread and no delay}\par
	\includegraphics[width=.45\textwidth,trim={0 0 0 .1cm}]{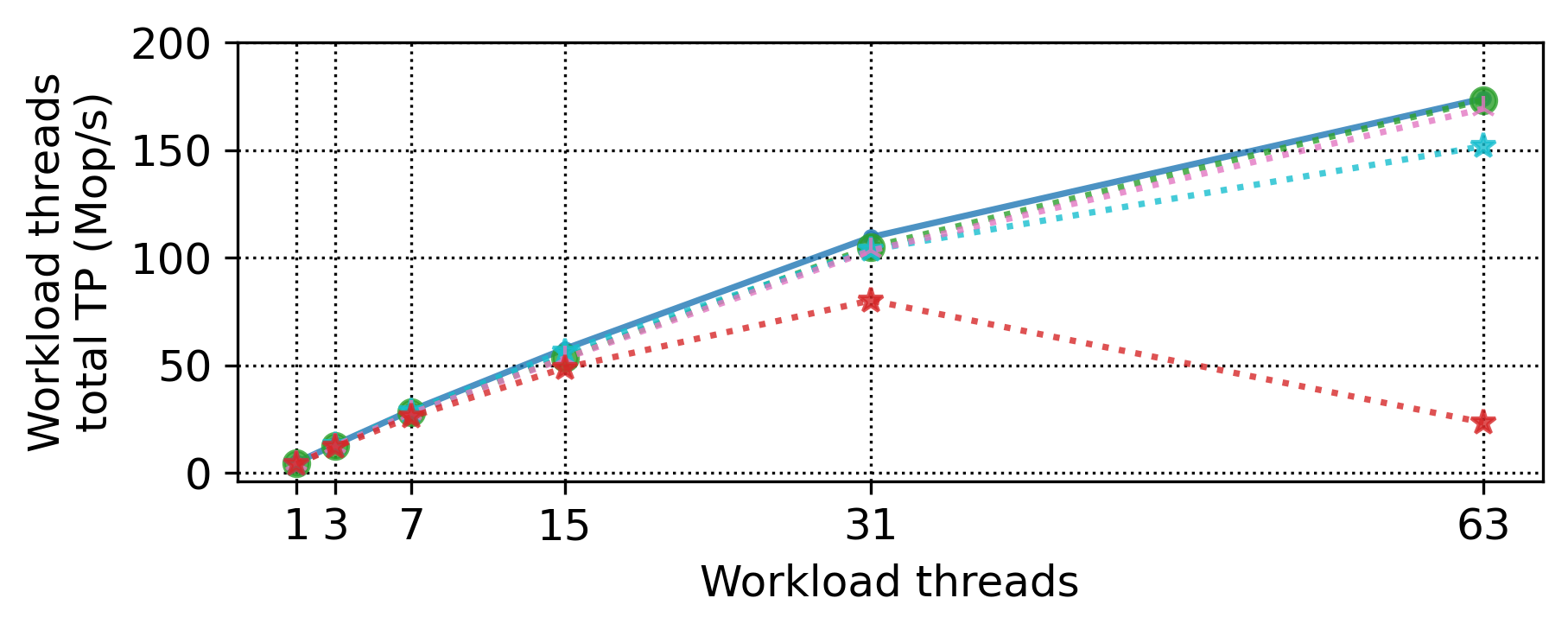}\hspace{2.5em}
	\includegraphics[width=.45\textwidth,trim={0 0 0 .1cm}]{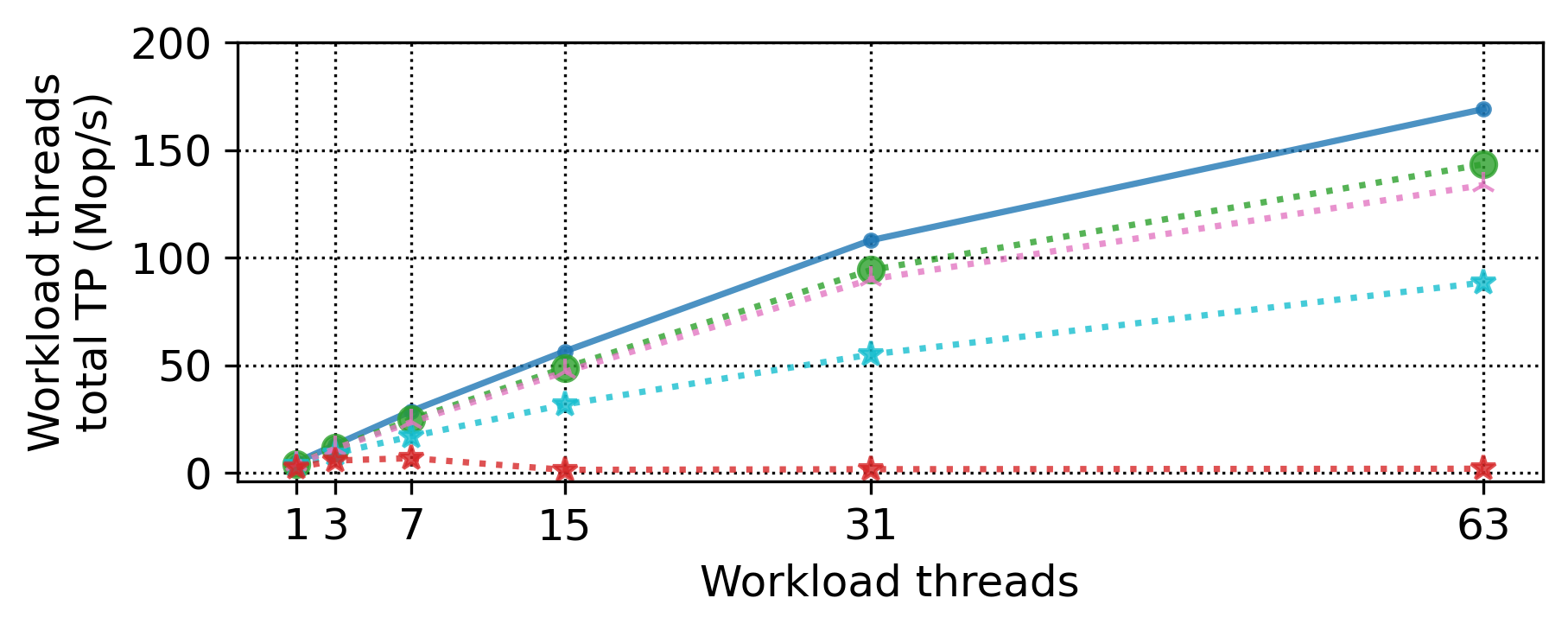}\par
	\includegraphics[width=.45\textwidth,trim={0 0 0 .2cm}]{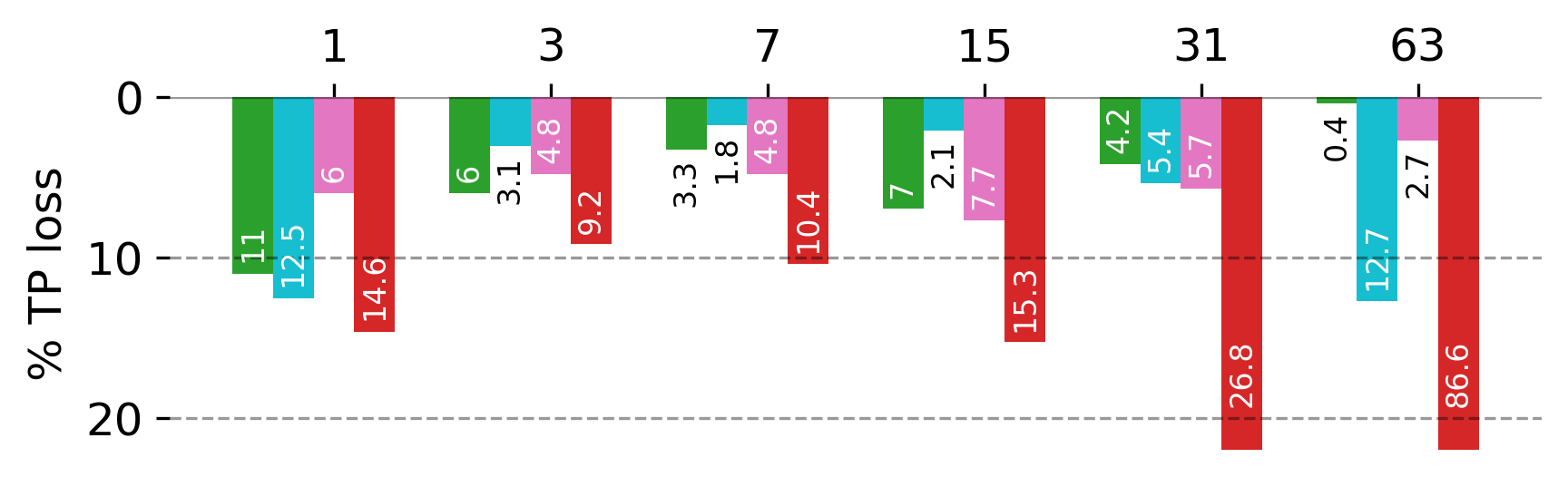}\hspace{2.5em}
	\includegraphics[width=.45\textwidth,trim={0 0 0 .2cm}]{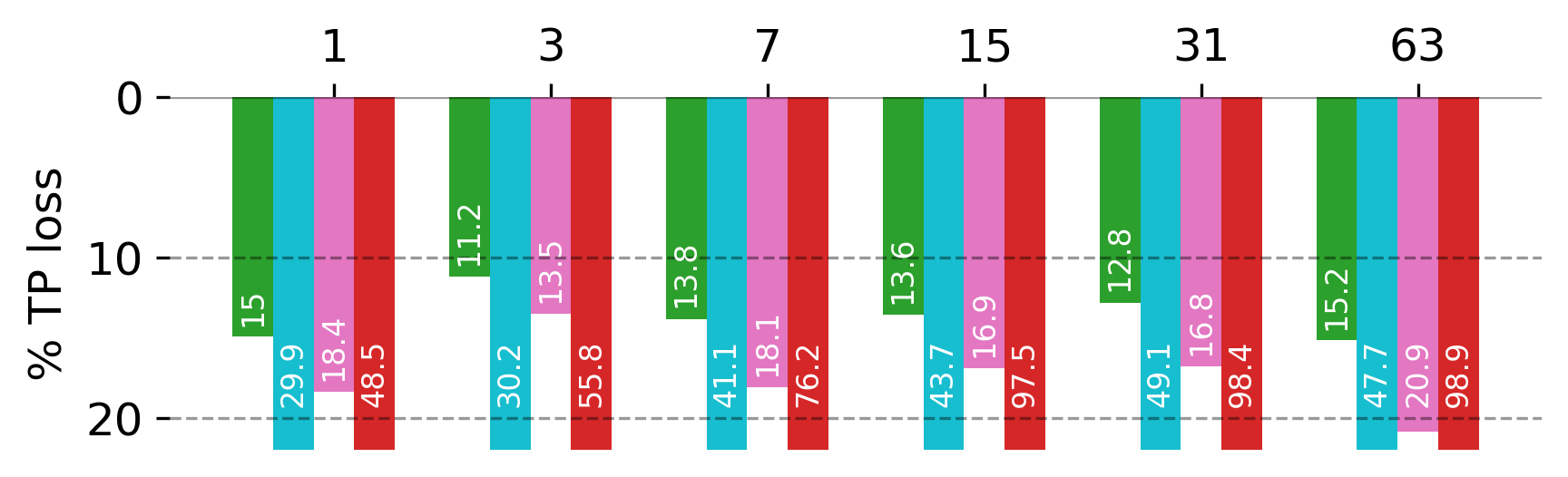}\par
	\medskip
	\text{With a concurrent \size{} thread and 700 \si{\micro\second} delay}\par
	\includegraphics[width=.45\textwidth,trim={0 0 0 .1cm}]{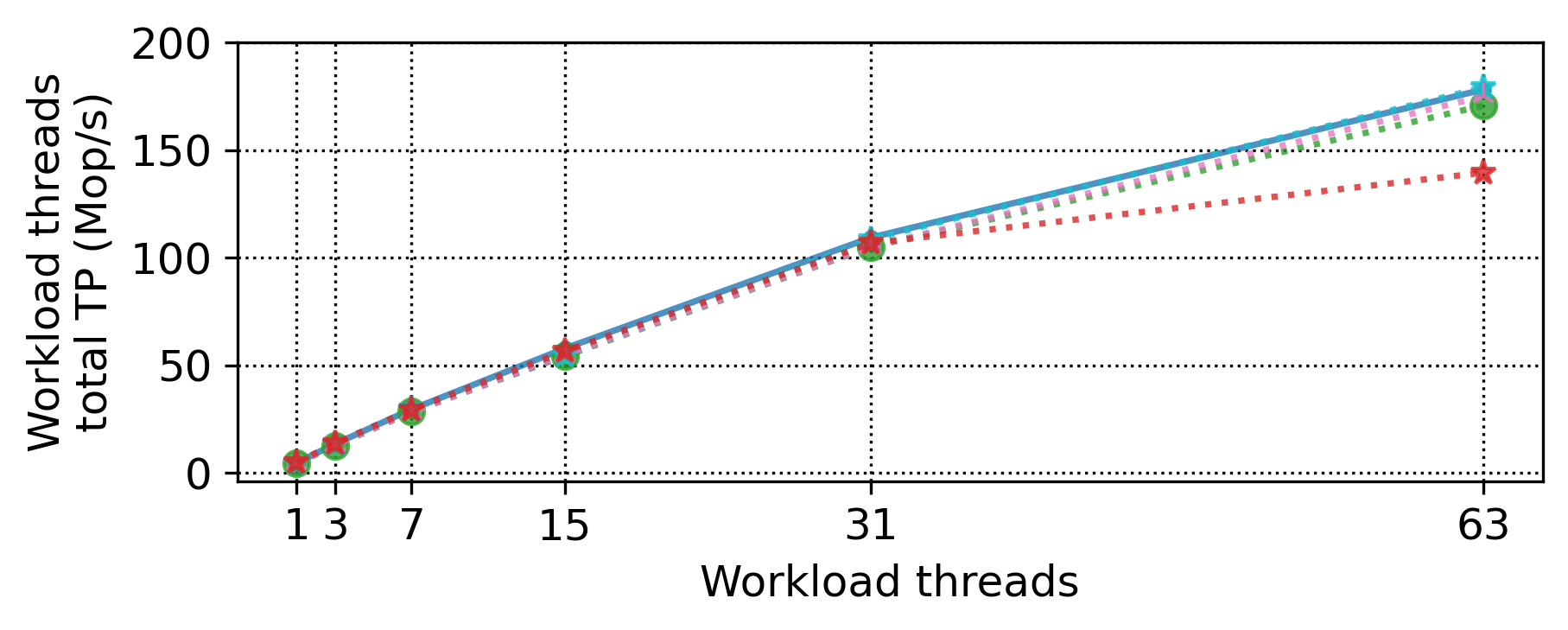}\hspace{2.5em}
	\includegraphics[width=.45\textwidth,trim={0 0 0 .1cm}]{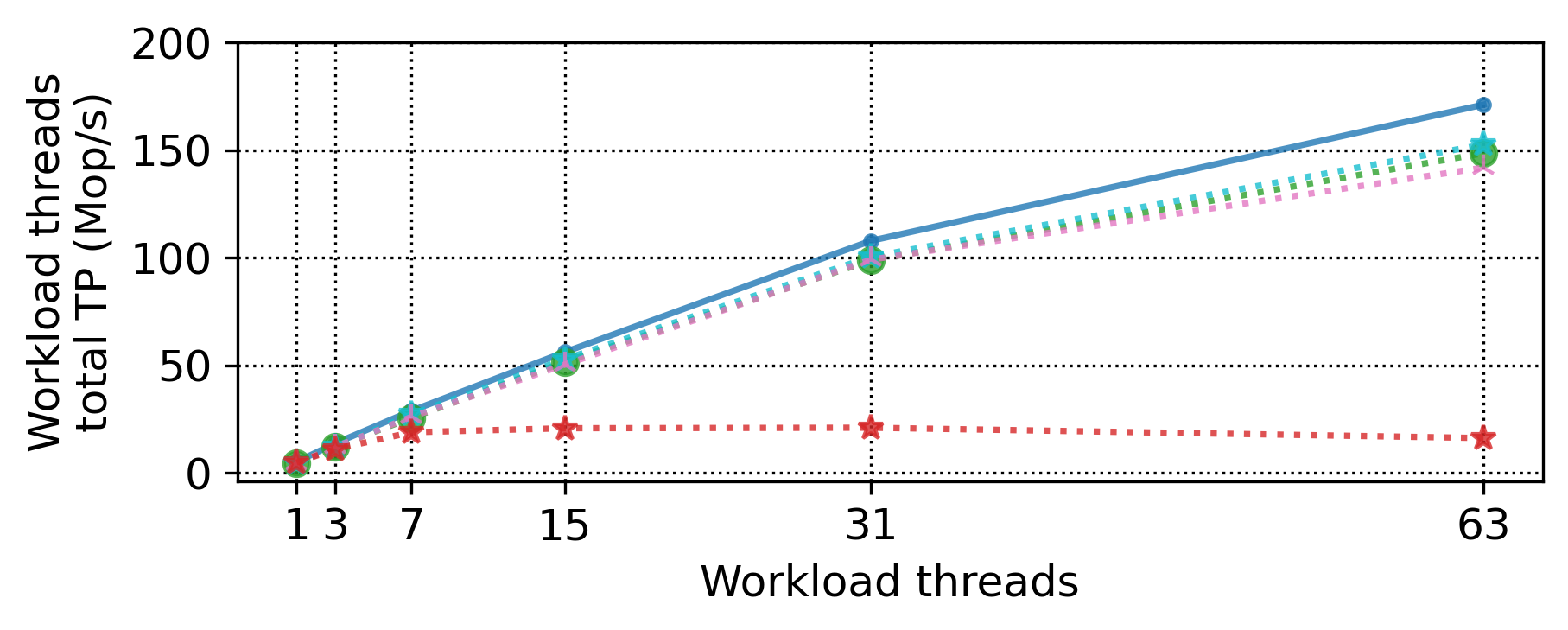}\par
	\includegraphics[width=.45\textwidth,trim={0 0 0 .2cm}]{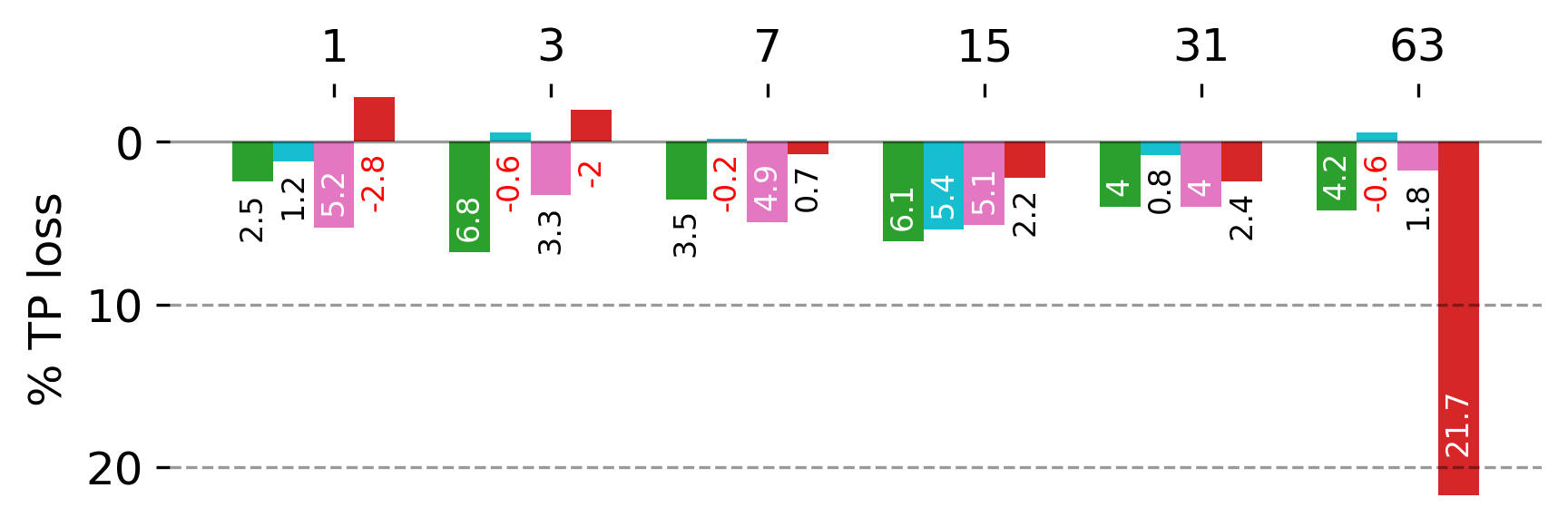}\hspace{2.5em}
	\includegraphics[width=.45\textwidth,trim={0 0 0 .2cm}]{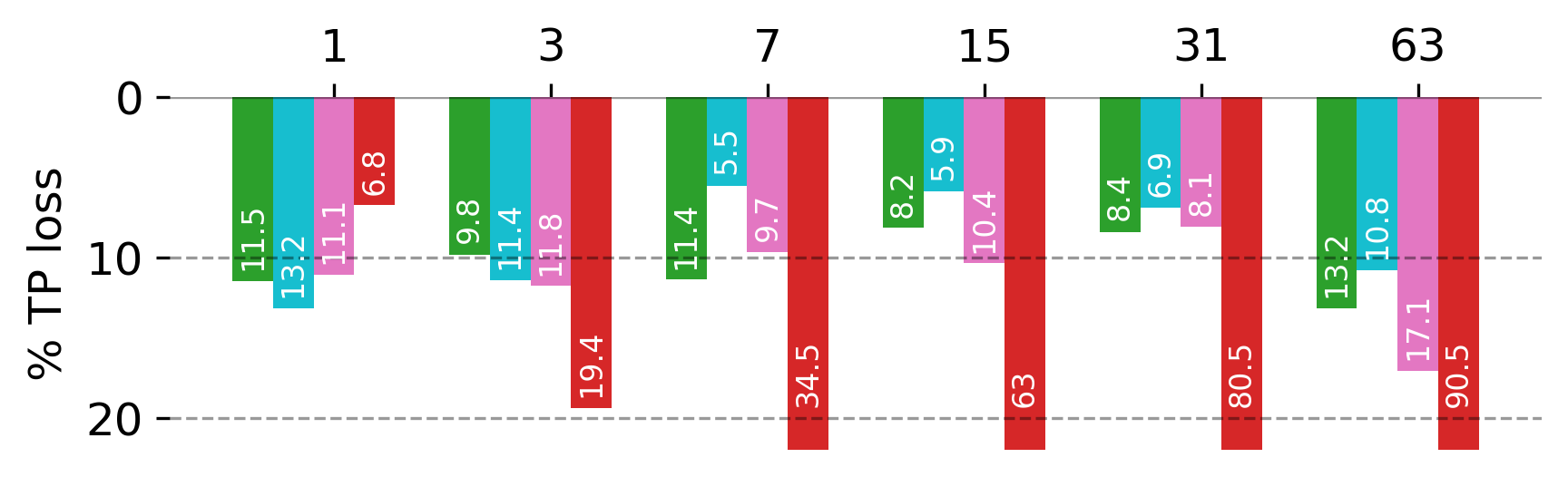}\par
	\caption{Overhead on hash table operations with Zipfian-distributed \contains{}}
	\label{fig:HashTable zipfian overhead}
\end{figure*}
\vspace*{\fill}
\clearpage

\section{Additional implementation details}\label{section:implementation-details}
\subsection{Complementary pseudocode}
We bring the pseudocode of the data structure transformation for our different methodologies in \Cref{fig: transformed data structure with handshakes,fig: transformed data structure with an optimistic scheme,fig: transformed data structure with a readers-writer lock}. In addition, the methods of the \codestyle{HandshakeCountersSnapshot} object appear in \Cref{fig:HandshakeCountersSnapshot}, and the methods of the \codestyle{LocksSizeCalculator} object appear in \Cref{fig:LocksSizeCalculator}.

\begin{figure}[h]
\begin{lstlisting}
IDLE_PHASE = 0, FAST_PHASE = 1
@\underline{class TransformedDataStructureWithHandshakes}@:
    @\underline{TransformedDataStructureWithHandshakes()}@:
        Initialize as originally.
        this.sizeCalculator = new HandshakeSizeCalculator()
    @\underline{$fast\_op$(k)}@:
        @Perform the original operation$^{*}$. For an insert or a delete operation that succeeded call @this@.sizeCalculator.fastUpdateMetadata with the relevant opKind.@
        Return the result of the original operation.
    @\underline{$slow\_op$(k)}@:
        @Perform the transformed operation defined in \cite{sela2021concurrentSize}$^{**}$ and return its result.@
    @\underline{$op$(k)}@: // this transformation is for insert/delete operations@\label{trans_op}@
        this.sizeCalculator.setOpPhaseVolatile(FAST_PHASE)@\label{trans_op: set fast phase}@
        currSizePhase = this.sizeCalculator.getSizePhase() @\label{trans_op: read size phase}@
        if currSizePhase%4 == 0: @\label{trans_op: start perform operation}@
            ret = @$fast\_op$(k)@ @\label{trans_op: execute fastop}@
        else: // Some thread runs size()
            this.sizeCalculator.setOpPhase(currSizePhase)@\label{trans_op: set slow phase}@
            ret = $slow\_op$(k) @\label{trans_op: execute slowop}@
        this.sizeCalculator.setOpPhase(IDLE_PHASE) @\label{trans_op: set idle}@
        return ret @\label{trans_op: return}@
    @\underline{$contains$(k)}@:
        return $slow\_op$(k)
    @\underline{size()}@:
        return this.sizeCalculator.compute()
@$^*$Consider also nodes with a non-$\nul{}$ deleteInfo field as marked for deletion.@
@$^{**}$Do not call @this@.sizeCalculator.updateMetadata before unlinking nodes marked using the original data structure's marking scheme (call it only for nodes marked using a non-$\nul{}$ deleteInfo field).@

\end{lstlisting}
\caption{A transformed data structure with handshakes}\label{fig: transformed data structure with handshakes}
\end{figure}

\clearpage
\vspace*{\fill}
\begin{figure}
    \begin{lstlisting}
@\underline{class TransformedDataStructureOptimistic}@:
    @\underline{TransformedDataStructureOptimistic()}@:
        Initialize as originally
        this.sizeCalculator = new OptimisticSizeCalculator()
    @\underline{contains(k)}@:
        @Perform the original contains operation@
    @\underline{insert / delete(k)}@:
        this.sizeCalculator.helpSize()
        @Search as originally for the place to insert k in case of insert / for a node with key k in case of delete@
        @Return on failure (if k is present in an unmarked node in case of insert / not present in case of delete)@
        this.sizeCalculator.incrementActivityCounter()@\label{code:incrementActivityCounter before modification}@
        @Perform the original modification attempt, if successful perform @this.sizeCalculator.updateMetadata(INSERT / DELETE)
        this.sizeCalculator.incrementActivityCounter()@\label{code:incrementActivityCounter after modification}@
        Return the result of the original modification attempt
    @\underline{size()}@:
        return this.sizeCalculator.computeSize()
\end{lstlisting}
\caption{A transformed data structure with an optimistic scheme}\label{fig: transformed data structure with an optimistic scheme}
\end{figure}

\begin{figure}
\begin{lstlisting}
@\underline{class TransformedDataStructureWithRWLock}@:
    @\underline{TransformedDataStructureWithRWLock()}@:
        Initialize as originally
        this.sizeCalculator = new LocksSizeCalculator()
    @\underline{contains(k)}@:
        @Perform the original contains operation@
    @\underline{insert / delete(k)}@:
        @Search as originally for the place to insert k in case of insert / for a node with key k in case of delete@
        @Return on failure (if k is present in an unmarked node in case of insert / not present in case of delete)@
        this.sizeCalculator.readWriteLock.readLock()$^*$
        @Perform the original modification attempt, if successful call this.sizeCalculator.updateMetadata(INSERT / DELETE)@
        this.sizeCalculator.readWriteLock.readUnlock()$^*$
        Return the result of the original modification attempt
    @\underline{size()}@:
        return this.sizeCalculator.computeSize()
@$^*$The locking and unlocking scheme depends on the implementation of the lock used and may look different (e.g. when using a StampedLock).
    
\end{lstlisting}
\caption{A transformed data structure with a readers-writer lock}\label{fig: transformed data structure with a readers-writer lock}
\end{figure}
\vspace*{\fill}
\clearpage

\begin{figure}
\begin{lstlisting}
@\underline{\textbf{class} HandshakeCountersSnapshot}@:
    @\underline{HandshakeCountersSnapshot()}@:@\label{code:HandshakeCountersSnapshot ctor}@
        this.snapshot = new long[numThreads][2]
        setVolatile all cells of this.snapshot to INVALID
        this.collecting.setVolatile(true)
        this.size.setVolatile(INVALID)
    @\underline{computeSize(fastSize)}@:
        computedSize = fastSize@\label{code:HandshakeCountersSnapshot computeSize compute start}@
        for each tid:
            computedSize += this.snapshot[tid][INSERT].getVolatile() - @ @ this.snapshot[tid][DELETE].getVolatile() // INSERT=0, DELETE=1@\label{code:HandshakeCountersSnapshot computeSize compute end}@
        this.size.setVolatile(computedSize)@\label{code:HandshakeCountersSnapshot computeSize CAS}@
        return computedSize
\end{lstlisting}
\caption{\codestyle{HandshakeCountersSnapshot} methods}\label{fig:HandshakeCountersSnapshot}
\end{figure}

\begin{figure}
\begin{lstlisting}
@\underline{\textbf{class} LocksSizeCalculator}@:
    @\underline{LocksSizeCalculator()}@:
        this.metadataCounters = new long[numThreads]
        this.readWriteLock = new ReadWriteLock()
        this.sizeInfo = new SizeInfo()
    @\underline{updateMetadata(opKind)}@: @\label{code: locks updatemetadata}@
        tid = ThreadID.threadID.get()  
        if opKind == INSERT: @\label{code:locks opkind check if}@
            this.metadataCounters[tid].setVolatile(1+this.metadataCounters[tid].getVolatile())
        else:
            this.metadataCounters[tid].setVolatile(-1+this.metadataCounters[tid].getVolatile())
    @\underline{computeSize()}@:
        currentSizeInfo = this.sizeInfo.getVolatile() 
        if currentSizeInfo.size.getVolatile() != INVALID_SIZE: 
            newSizeInfo = new SizeInfo()
            witnessedSizeInfo = this.sizeInfo.compareAndExchange(currentSizeInfo, newSizeInfo)
            if witnessedSizeInfo == currentSizeInfo:
                size = _computeSize()
                newSizeInfo.size.setVolatile(size)
                return size
            currentSizeInfo = witnessedSizeInfo                
        return _waitForComputing(currentSizeInfo)
    @\underline{\_computeSize()}@:
        sum = 0
        this.readWriteLock.writeLock()*
        for each tid:
            sum += this.metadataCounters[tid].getVolatile()
        this.readWriteLock.writeUnlock()*
        return sum
    @\underline{\_waitForComputing(currentSizeInfo)}@:
        while true:
            currentSize = currentSizeInfo.size
            if currentSize != INVALID_SIZE:
                return currentSize
    @$^*$The locking and unlocking scheme depends on the implementation of the lock used and may look different (e.g. when using a StampedLock).

\end{lstlisting}
\caption{\codestyle{LocksSizeCalculator} methods}\label{fig:LocksSizeCalculator}
\end{figure}


\subsection{Details of the Locks approach}\label{app-locks-details}
In this section we detail the data structure transformation to make it support our lock-based size mechanism. The  transformation pseudo-code appears in \Cref{fig: transformed data structure with a readers-writer lock}.
We add a readers-writer lock to the data structure in the form of a field named \codestyle{readWriteLock} placed in a \codestyle{LocksSizeCalculator} object (see \Cref{fig:LocksSizeCalculator} for its full method pseudocode). Different implementations of such a lock can be used; we used Java's \texttt {StampedLock} class from the \texttt{java.util.concurren\-t.locks} package in our evaluation as it provided the best results out of the tested lock implementations.
Additionally, we add an array named \codestyle{metadataCount\-ers} to the  \codestyle{LocksSizeCalculator} object, with a cell per thread to keep track of the size metadata for each thread.

An \ins{} operation starts with a search to find the insertion point. If an unmarked node with the required key is already found, it returns a failure. Otherwise, the read lock is acquired by invoking the \codestyle{readLock()} method on the \codestyle{readWriteLock} object. Following this, an insertion attempt is executed as in the original data structure. If it concludes successfully, the current thread's cell in the \codestyle{metadataCounters} array is incremented by 1. To wrap up the process, the read lock is released by calling the \codestyle{readUnlock()} method on the \codestyle{readWriteLock} object, and the result of the insertion attempt is returned.

Similarly, a \del{} operation begins by searching for a node with the key it wishes to delete. If such a node is not located, the operation promptly returns a failure. However, if found, the read lock is acquired by invoking the \codestyle{readLock()} method on the \codestyle{readWrite\-Lock} object. The operation then advances to execute a deletion attempt like in the original data structure. If successfully completed, the current thread's cell in the \codestyle{metadataCounters} array is decremented by 1. Finally, the read lock is unlocked by calling the \codestyle{read\-Unlock()} method on the \codestyle{readWriteLock} object, and the outcome of the deletion attempt is returned.

The \contains{} operation remains as in the original data structure. 
Lastly, the \size{} operation sums the values of all the cells in the \codestyle{metadataCounters} array. To do so, the write lock is acquired by invoking the \codestyle{writeLock()} method on the \codestyle{readWriteLock} object. The operation then iterates over the array and sums the values of all the cells. Once the summation is complete, the write lock is released using the \codestyle{writeUnlock()} method on the \codestyle{readWriteLock} object and the result of the summation is returned. The acquisition of the write lock ensures that no \ins{} or \del{} operation is in an inconsistent state while the summation is executed. 

The updates and summation of the \codestyle{metadataCounters} array are not executed using a special snapshot algorithm. This is because due to the mutual exclusion guaranteed by the acquisition of the write lock, when the \size{} operation is accessing the array, no other thread can update it. This makes a simple pass over the array sufficient to correctly compute the size.

To allow multiple \size{} operations to be performed concurrently in an efficient manner, we place a field holding a shared object of type \codestyle{SizeInfo} in the \codestyle{LocksSizeCalculator} object, which has a single field for holding the computed size. At the start of a \size{} operation, it checks if the \codestyle{SizeInfo} instance currently installed in that field has a valid size value written to it. If not, the operation waits until a valid size value is written to the \codestyle{SizeInfo} instance and then returns that value. Otherwise, the operation attempts to replace the existing \codestyle{SizeInfo} instance with a new one with a size field initialized to \codestyle{INVALID\_SIZE} using \codestyle{compareAndExchange}. If the \codestyle{compareAndExchange} fails, the operation waits until a valid size value is written to the \codestyle{SizeInfo} instance and then returns it. If the \codestyle{compareAndExchange} succeeds, the \size{} operation is responsible for computing the size by acquiring the write lock, summing the metadata array, releasing the write lock and writing the computed size value into the \codestyle{SizeInfo} instance; it then returns the computed size.

\subsection{Thread registration} \label{subsection: Thread registration}

\begin{figure}[b]
\begin{lstlisting}
@\underline{\textbf{class} ThreadID}@:
    MAX_THREADS = 128
    this.threadID = ThreadLocal<Integer>()
    this.pool = PriorityBlockingQueue<Integer>(MAX_THREADS)
    this.nextId = AtomicInteger(0)
    @\underline{register()}@:
        if this.threadID.get() is not null:
            throw new RuntimeException("Thread already registered")
        tid = this.pool.poll()
        if tid is null:
            tid = this.nextId.getAndIncrement()
            if tid >= MAX_THREADS:
                throw new RuntimeException("Too many threads")
        this.threadID.set(tid)
    @\underline{deregister()}@:
        tid = this.threadID.get()
        if tid is null:
            throw new RuntimeException("Thread not registered")
        this.pool.add(tid)
        this.threadID.set(null)

\end{lstlisting}
\caption{\codestyle{ThreadID} class methods}\label{fig:ThreadID}
\end{figure}

Each methodology we study utilizes a metadata array to effectively track the count of insertions and deletions on a per-thread basis.
Within this metadata array, every thread is allocated a distinct cell.
To allocate a cell to each thread, we incorporate a registration mechanism, assigning a unique ID to each thread that aligns with a cell in the array.
To facilitate this thread identification and management in a concurrent environment, we introduce the \codestyle{ThreadID} class presented in \Cref{fig:ThreadID}, in which we have implemented a mechanism to manage thread registration.
Within this class, an \texttt{AtomicInteger} variable is utilized to keep track of the next thread ID that has not yet been used.
A pool is maintained to store the thread IDs that have been released by other threads, ensuring reuse of these IDs for subsequent thread registrations.
Before performing any operation on the data structure a thread must call \codestyle{ThreadID.r\-egister()}.
When it is done using the data structure it should call \codestyle{ThreadID.deregister()} to release its thread ID allowing other threads to use it.

It is important to note that the reassignment of a thread ID from one thread to another does not compromise the correctness of the \size{} operation.
A thread should perform deregistration only after it is done operating on the data structure, and each operation on the data structure returns only after it has been finalized.
Consequently, a specific thread ID is allocated to one operating thread only at any given time, and if a new thread is allocated a previously used ID, the data structure continues to reflect the cumulative effects of all operations conducted under that ID. 
Therefore, the data structure maintains its integrity and correctness even as thread IDs are dynamically allocated and deallocated among different threads. 

We used Java's \codestyle{PriorityBlockingQueue} class from the \texttt{java.ut\-il\-.concurrent} package to serve as our concurrent pool, ensuring the management of concurrent accesses.
This class implements the \codestyle{poll()} method to allow extraction of an element from the pool and the \codestyle{add()} method to allow insertion of a new element into the pool.
In addition, we utilized Java's \texttt{AtomicInteger} from the \texttt{java.util.concurrent} package to keep track of the next available thread ID in an atomic manner.

\subsection{General Optimizations} \label{subsection: general optimizations}
\subsubsection{Avoid false sharing}
To prevent false sharing among threads while accessing arrays that hold per-thread data, for each array of this kind utilized in the different methodologies (\codestyle{meatadataCounters}, \codestyle{fastMeatadataCount\-ers}, \codestyle{opPhase} and \codestyle{activityCounters}), we pad its cells so that the data of each thread occupies a full cacheline.

\subsubsection{Partial array iteration} \label{subsubsection: partial array iteration opt}
The usage of the registration scheme described in \Cref{subsection: Thread registration} enables the determination of the maximum number of threads that have been operating concurrently on the data structure. This number is represented as the value of the \codestyle{nextId} variable. Consequently, when going over any metadata array sized to the number of threads in each methodology (for example, the \codestyle{activityCounters} array in the optimistic methodology), it is possible to go over only the first \codestyle{nextId} cells and not iterate over the entire array. This can improve performance in cases where the maximal number of concurrent active threads is fewer than the maximal number of threads (which determines the size of the \size{} metadata array).
To implement this, in places where we iterate over the \size{} metadata array (not including the initialization of this array), we first read the value of the \codestyle{nextId} variable. Then, we iterate over only the first few cells of the array based on the value we read from \codestyle{nextId}.
In order to address the race condition in which a new thread has registered causing \codestyle{nextId} to increment after we have read the value of the \codestyle{nextId} variable allowing that new registered thread to modify a cell which we are not aware exists, we read this value again when we have finished iterating. If the value has changed (which would only be an increase), we repeat the process by reading the \codestyle{nextId} variable again. We will only finish once we reach an iteration where this value has not changed.
This verification loop is not necessary in all cases, for example, in the optimistic methodology when executing the \codestyle{\_readActivityCounters} function there is no need to verify the value of \codestyle{nextId} as if we missed a thread's registration - in the case that thread has modified the \codestyle{activityCounters} array we will find out about it in the \codestyle{\_retryActivityCounters} function, causing the \size{} attempt to retry.
To ensure a fair comparison, we also incorporated this optimization into the \spsize{} implementation in our measurements.

\subsubsection{Usage of tailored \codestyle{opKinds}} \label{subsubsection: tailored opkinds opt}

The \size{} metadata in \cite{sela2021concurrentSize} includes 2-cell per-thread arrays containing separate insertion counter and deletion counter.
The \codestyle{opKind} values in \cite{sela2021concurrentSize} are 0 (for \codestyle{INSERT}) and 1 (for \codestyle{DELETE}), since they are used as an index when accessing those arrays. This is also the way \codestyle{opKind} values are used in slow operations in the handshakes methodology in this paper.

However, the rest of our methodologies (including fast operations in the handshakes methodology) include a single counter per thread, and \codestyle{opKind} is used to distinguish whether to add 1 to a counter (in case of insertion) or to subtract 1 from the counter (in case of deletion).
Therefore, in these cases, we can make usage of \codestyle{opKind} values that will allow us to eliminate the conditional statements when updating the \size{} metadata, appearing in \Cref{code:HandshakeSizeCalculator opkind check if,code:OptimisticSizeCalculator opkind check if,code:locks opkind check if}. To achieve this, we define new, different \codestyle{opKind} values to be \codestyle{INSERT=1} and \codestyle{DELETE=-1}. When updating the metadata, rather than checking whether the \codestyle{opKind} indicates an \ins{} or a \del{} operation to decide whether to increase or decrease the counter, we simply add the \codestyle{opKind} value directly to the corresponding cell in the \size{} metadata array.

\subsection{Memory Model}\label{subsection: Memory Model}
In our Java implementation, volatile semantics are consistently used in read, write, and \cas{} operations on non-final fields of shared objects. This approach is carried out through the usage of volatile variables, \codestyle{VarHandle}s and \codestyle{AtomicReferenceFieldUpdater}s.

Under the Java memory model, any access that utilizes volatile semantics is treated as a synchronization action. The model, in turn, commits to a synchronization order for these actions. This allows for a total order that seamlessly aligns with each thread's program order. Moreover, any read operation executed on a volatile variable is promised to fetch its most recently written value, in accordance with the synchronization order.

\section{Correctness of the handshake-based methodology}\label{handshakes:linearizability}
Let us prove that the \size{} method that uses the handshake-based synchronization is correct.


\subsection{Linearization points} 
Our handshake-based methodology is linearizable. We detail the methods' linearization points next, and bring the linearizability proof in \Cref{section:handshake linearizability proof}.

A \size{} operation that has managed to successfully install its \codestyle{HandshakeCoun\-tersSnapshot} object in \Cref{code:HandshakeSizeCalculator compareAndExchange} is linearized as in \cite{sela2021concurrentSize}, namely, when the collecting field of the \codestyle{CountersSnapshot} instance it operates on is set to \textit{false} for the first time. Otherwise, a \size{} operation that had to wait on another \size{}'s \codestyle{HandshakeCoun\-tersSnapshot} object (i.e., called \codestyle{\_waitForComputing}) is linearized at the linearization point of the \size{} which installed the \codestyle{Handshake\-CountersSnapshot} object it obtained in \Cref{code: HandshakeSizeCalculator read counterSnapshot}.

Fast operations are always linearized according to the original linearization point. Any successful slow \ins{} or \del{} operation that has seen the phase number of the first handshake (i.e., has read \codestyle{size\_phase} $\equiv_{4}{1}$ in \Cref{trans_op: read size phase}) is linearized according to its original linearization point. Otherwise, a successful slow \ins{} or \del{} operation that has seen the number of the second handshake phase (\codestyle{size\_phase} $\equiv_{4}{2}$ in \Cref{trans_op: read size phase}) is linearized according to the linearization point defined in \cite{sela2021concurrentSize}. Specifically, if no \size{} operation is collecting at the time of the metadata update, the operation is linearized at the metadata update itself. Otherwise, if a \size{} operation is collecting, the operation is linearized according to that \size{} operation: if the \size{} operation takes the update into account, then the operation is linearized at the metadata update; otherwise, it is retrospectively linearized immediately after the linearization point of that \size{} operation. However, if a dependent concurrent fast operation is linearized between the original linearization point and the linearization point in \cite{sela2021concurrentSize}, the slow operation is linearized at the latter of its original linearization point and right after the linearization point of the last concurrent \size{} operation that does not take it into account (due to a scenario we elaborate on in the next paragraph).
Note that two slow operations can be linearized at the same time (immediately after a \size{} operation), in which case we order them according to the order of the linearization points defined in \cite{sela2021concurrentSize}. 
The linearization points of \contains{} and failing slow \ins{} or \del{} operations follow a methodology similar to that presented in \cite{sela2021concurrentSize} with the slight modification that we now linearize them based on the new linearization points of \ins{} or \del{} operations.
In detail, an operation $op$ which is a \contains{} or a failing slow \ins{} or \del{} is linearized at the original linearization point unless the operation it depends on (namely, the last successful update operation on $k$ whose original linearization point precedes $op$'s original linearization point) is a slow operation that has yet to be linearized (according to the new linearization point defined above) at $op$'s original linearization point, in which case we linearize $op$ immediately after that operation is linearized.

The scenario in which the metadata related to a successful slow \ins{} or \del{} operation, is updated later than the linearization point of a dependent fast operation, requires special handling of the slow operation's linearization point. It cannot be linearized like in \cite{sela2021concurrentSize} at its metadata update as it occurs after the linearization point of the dependent fast operation. Instead, we linearize it beforehand as we demonstrate on the scenario illustrated in \Cref{fig:execution-with-fast-concurrent-with-old-slow-example}, where a slow \ins{}(1) updates the metadata only after a dependent fast \del(1) is linearized. The \size{} operation in the figure is the one that set \codestyle{sizePhase} to the value obtained by the \ins{}. This is the last \size{} operation that does not take the \ins{} into account (this \size{} operation does not take the \ins{} into account since it must have completed before the fast \del{} ran which happened in turn before the \ins{} updated the size metadata; following \size{} operations perform handshakes before computing the size and may hence compute the size from the metadata only after the \ins{} completes including its metadata update). 
This \size{} operation may be linearized either before the \ins{} inserts 1 to the data structure or afterwards. If it is linearized first, then we linearize the \ins{} when it inserts 1 to the data structure. Else, we linearize the \ins{} right after the linearization point of the \size{}. In both cases, the chosen linearization guarantees that the \ins{} is linearized after the \size{} which did not take it into account and before any dependent operation.

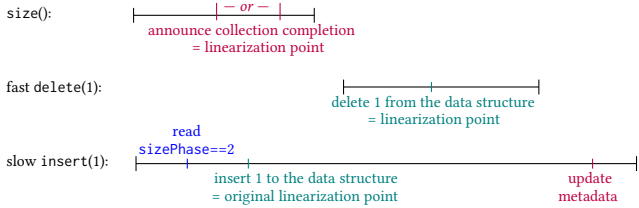
\begin{figure}
    \begingroup
    \renewcommand{\codestyle}[1]{{\fontsize{6}{8.5}\selectfont\texttt{#1}}} 
    \fontsize{6}{9.5}\selectfont
    \setlength{\unitlength}{0.065mm}
    \begin{picture}(700,230)
        \put(-300,150){\size{}():}
        \put(-300,0){fast \del(1):}
        \put(-300,-150){slow \ins(1):}

        \put(-40,170){\line(0,-1){30}} 
        \put(-40,155){\line(1,0){370}} 
        \put(330,170){\line(0,-1){30}} 
        \put(130,180){\color{purple} \line(0,-1){35}}
        \put(145,163){\color{purple} --- \emph{or} ---}
        \put(260,180){\color{purple} \line(0,-1){35}}
        \put(-10,115){\color{purple} announce collection completion}
        \put(80,80){\color{purple} = linearization point}

        \put(390,25){\line(0,-1){30}} 
        \put(390,10){\line(1,0){400}} 
        \put(790,25){\line(0,-1){30}} 
        \put(570,20){\color{teal} \line(0,-1){20}} 
        \put(365,-30){\color{teal} delete 1 from the data structure} 
        \put(440,-65){\color{teal} = linearization point} 

        \put(-35,-130){\line(0,-1){30}} 
        \put(-35,-145){\line(1,0){1025}} 
        \put(990,-130){\line(0,-1){30}} 
        \put(70,-135){\color{blue} \line(0,-1){20}}
        \put(40,-90){\color{blue} read}
        \put(-30,-125){\color{blue} \codestyle{sizePhase}==2} 
        \put(195,-135){\color{teal} \line(0,-1){20}} 
        \put(125,-185){\color{teal} insert 1 to the data structure}
        \put(120,-220){\color{teal} = original linearization point} 
        \put(900,-135){\color{purple} \line(0,-1){20}}
        \put(850,-185){\color{purple} update}
        \put(830,-220){\color{purple} metadata}
    \end{picture}
    \vspace{50pt}
    \caption{An execution with a slow \texttt{insert} that takes part in a second handshake with \texttt{size} and updates the metadata after a dependent fast \texttt{delete} executes}
    \label{fig:execution-with-fast-concurrent-with-old-slow-example}
    \endgroup
\end{figure}

\subsection{Linearizability proof}\label{section:handshake linearizability proof}
In this section we prove the linearizability of the handshake-based methodology presented in \Cref{section:handshakes}, using the linearization points stated in \Cref{handshakes:linearizability}.
We will prove linearizability in a manner similar to \cite{sela2021concurrentSize}.
This requires us to show (1) each linearization point occurs within the operation's execution time, and (2) ordering an execution's operations (with their results) according to their linearization points forms a legal sequential history.
We prove Property (1) in \Cref{claim:lin within op} in \Cref{linearization timing} and Property (2) in \Cref{claim:results comply} in \Cref{linearization is legal}.

Compared to the proof presented in \cite{sela2021concurrentSize}, the main distinction in this context lies in having to consider new linearization points as well as to closely examine the handshake mechanism. It is essential to establish and prove significant observations related to this mechanism. These observations are important to determining which types of operations (fast path operations and slow path operations) can be executed concurrently to a \size{} operation and which can not. Our linearizability proof will rely closely on these observations.

\subsubsection{Linearization points occur within operations' intervals}\label{linearization timing}

\begin{claim}\label{claim:lin within op}
The linearization point of each operation occurs within its execution time.
\end{claim}

\begin{proof}
    For fast \ins{}, fast \del{}, slow \ins{}, slow \del{} and \contains{} operations that are linearized according to the original linearization point, the claim follows from the linearizability of the original data structure.\footnote{\label{fn: original lin points are within execution}There is a selection of linearization points for every linearizable data structure such that each of them is placed within the execution period of the relevant operation. We only use linearization points that meet this criterion when we discuss linearization points of the original data structure.}
    For \size{} operations, they are linearized in the same manner as in \cite{sela2021concurrentSize}. As the \codestyle{CountersSnapshot} object each \size{} operation holds is obtained in the same way as in \cite{sela2021concurrentSize} and the operation is linearized according to \cite{sela2021concurrentSize}, the same arguments in \emph{Claim 8.1} in \cite{sela2021concurrentSize} can be applied and the claim is valid.
    For successful slow \ins{} or slow \del{} that are not linearized at the original linearization point, from \Cref{lemma:lin point is before lin point from paper} and \Cref{lemma: lin point of succ ins/del is after orig lin} we conclude that the linearization point of such an operation occurs within its execution time.
    It is left to show the claim that for \contains{} and failing slow \ins{} or slow \del{} operations that are not linearized according to the original linearization point.
    Let $op$ be such an operation.
    Since $op$ was not linearized at its original linearization point and from the way we defined our linearization points we know that there must exist some slow operation $op_2$ that $op$ depends on and has yet to be linearized at $op$'s original linearization point.
    In this case, $op$ is linearized immediately after the linearization point $op_2$. 
    Since $op_2$ was yet to be linearized at the time of $op$'s original linearization, we know that $op_2$'s linearization must occur after $op$'s original linearization point.
    Moreover, because $op_2$ is a successful slow operation by \Cref{lemma:lin point is before lin point from paper} we conclude that $op$ observes $op_2$'s metadata and calls \codestyle{updateMetadata} on behalf of $op_2$.
    By \Cref{corolarry:lin when updateMetadata returns}, $op_2$ must be linearized by the time \codestyle{updateMetadata} returns and thus $op$ is linearized by that time as well. 
    Therefore, $op$ is linearized within its execution time.
\end{proof}
When proving \Cref{claim:lin within op} we rely on the validity of \emph{Claim 8.1} from \cite{sela2021concurrentSize}, this validity holds only if \emph{Lemma 8.2} in \cite{sela2021concurrentSize} still holds. We next show in \Cref{corolarry:lin when updateMetadata returns} that this is indeed the case.
To show this we recall the following lemma from \cite{sela2021concurrentSize}, the correctness of these lemma follows from the proof of \emph{Lemma 8.2} in \cite{sela2021concurrentSize}:

\begin{lemma}\label{lemma:lin of paper when updateMetadata returns} 
When a call to \codestyle{updateMetadata} returns, the operation whose \codestyle{updateInfo} was passed to the call is guaranteed to have reached the linearization point defined in \cite{sela2021concurrentSize}.
\end{lemma}

Now, as a direct conclusion from \Cref{lemma:lin of paper when updateMetadata returns,lemma:lin point is before lin point from paper}:
\begin{corollary}
    \label{corolarry:lin when updateMetadata returns}
    When a call to \codestyle{updateMetadata} returns, the operation whose \codestyle{updateInfo} was passed to the call is guaranteed to be linearized.
\end{corollary}

The following lemmas are also used as part of \Cref{claim:lin within op}'s proof. 

\begin{lemma}\label{lemma:lin point is before lin point from paper}
    The linearization point of any successful slow \ins{} or \del{} operation occurs no later than the linearization point of that operation as defined in \cite{sela2021concurrentSize}.
\end{lemma}
\begin{proof}
    Denote by $op$ some successful slow \ins{} or \del{} operation.
    If $op$ is linearized at its original linearization point or at the linearization point as defined in \cite{sela2021concurrentSize} then by \Cref{lemma:lin points from paper are after original points} the lemma holds. 
    
    Otherwise, $op$ is linearized immediately after a concurrent \size{} operation that does not take it into account and that is linearized after $op$'s original linearization point: If the \size{} operation is collecting when $op$ is performing its metadata update, then by definition $op$'s linearization point as defined in \cite{sela2021concurrentSize} is immediately after that \size{} and we are done \footnote{Note that because we order successful update operations that are linearized at the same time in the same manner as in \cite{sela2021concurrentSize} then our linearization point will be the same as the one in \cite{sela2021concurrentSize}.}.
    Otherwise, if the \size{} operation is not collecting when $op$ is performing its metadata update, then the $op$' linearization point as in \cite{sela2021concurrentSize} will be at the metadata update.
    This metadata update must occur after \size{} finished collecting; otherwise, by \Cref{lemma:counter inc after first 2 lines of updateMetadata} \size{} should have seen the update.
    Thus, because \size{} is linearized when the collecting field is set to false, then the \size{} operation must be linearized before $op$'s linearization point as in \cite{sela2021concurrentSize}, and consequently $op$ will be linearized before its linearization point as in \cite{sela2021concurrentSize}.
\end{proof}

\begin{lemma}\label{lemma:counter inc after first 2 lines of updateMetadata}
    Consider a call to \codestyle{updateMetadata} on behalf of $op$, which is the $c$-th successful slow \ins{} or \del{} operation by a thread $T$. 
    After this call executes Lines 78-79 in \cite{sela2021concurrentSize}, the relevant metadata counter's value is $\geq c$.
    \end{lemma}
    
\begin{lemma}\label{lemma:lin points from paper are after original points}
    The linearization point of each successful insert or delete operation as defined in \cite{sela2021concurrentSize} occurs after its original linearization point.
\end{lemma}

\Cref{lemma:counter inc after first 2 lines of updateMetadata} correlates to \emph{Lemma 8.3} in \cite{sela2021concurrentSize} and \Cref{lemma:lin points from paper are after original points} is shown as part of \emph{Lemma B.3} in \cite{sela2021concurrentSize}, the proofs of both of these Lemmas remain the same.

\subsubsection{The linearization is legal}\label{linearization is legal}

In what follows, we denote the set's $i$-th successful \ins$(k)$ operation (by {\em $i$-th} we refer to the linearization order, namely, to the $i$-th successful \ins$(k)$ to be linearized) by \ins$_i(k)$, its linearization time by $t_{insert_i(k)}$, the time of its original linearization by $orig\_t_{insert_i(k)}$ and the time of its linearization point defined in \cite{sela2021concurrentSize} by $g\_t_{insert_i(k)}$ (this is defined only for slow operations).
We further denote the analogous \del{} operation and its related times by \del$_i(k)$, $t_{delete_i(k)}$, $orig\_t_{delete_i(k)}$, and $g\_t_{delete_i(k)}$.

\begin{claim}\label{claim:results comply}
Consider a sequential history formed by ordering an execution's operations (with their results) according to their linearization points defined in \Cref{handshakes:linearizability}. 
Then operation results in this history comply with the sequential specification of a set.
\end{claim}
\begin{proof}

The correctness of the results of successful update operations follows from \Cref{corollary:alternating ins and del}.
Now, let us examine the results of \contains{} operations and failing update operations.
Let $op$ be such an operation on a key $k$, and let the operation it depends on be $insert_i(k)$ for some $i\geq 1$, a similar proof can be made for a \del{} operation.
Because $insert_i(k)$ is an insertion then $op$ must be a failing insertion or a \contains{} returning true. 
Our goal is to show that \ins$_i(k)$ is indeed the last successful operation on key $k$ to be linearized before $op$.
Let $orig\_t_{op}$ be the original linearization moment of $op$ and $t_{op}$ its actual linearization point according to \Cref{handshakes:linearizability}.
If $op$ is a slow operation that is linearized immediately after $t_{insert_i(k)}$ then we are done.
Otherwise, $op$ must be linearized according to its original linearization point.
We begin by showing that $t_{insert_i(k)}<t_{op}$.
If $insert_i(k)$ is linearized according to its original linearization point then $t_{insert_i(k)}=orig\_t_{insert_i(k)} \overset{\text{the way we chose $i$}}{<} orig\_t_{op}=t_{op}$.
Otherwise, $t_{insert_i(k)}$ must be a slow operation and because $op$ is linearized at its original linearization point we know that by time $t_{op}$ it must be that $insert_i(k)$ is already linearized (this is due to the way we defined our linearization points in \Cref{handshakes:linearizability}).
Therefore, $t_{insert_i(k)} < t_{op}$.
We are left to show that if there is a successful delete operation linearized after \ins$_i(k)$ then $t_{op}<t_{delete_i(k)}$.
If there is no such operation then we are done, otherwise, from the way we chose $i$ and due to the linearizability of the original data structure it holds that $t_{op}=orig\_t_{op}<orig\_t_{delete_i(k)}\overset{\Cref{lemma: lin point of succ ins/del is after orig lin}}{\leq}t_{delete_i(k)}$.
\\ $ \Longrightarrow \text{In all cases it holds that } t_{insert_i(k)} < t_{op} < t_{delete_i(k)} \\ \text{ and we are done}$. \\

Finally, we analyze the linearization of a \size{} operation. Denote such an operation by $op$. From \Cref{corollary:compute only runs with slow op} we know there are no concurrent fast operations when \codestyle{fastCompute()} is executing thus the fast counter array values stay untouched when \codestyle{fastCompute()} is performed.
Therefore, any fast operation that $op$ has seen must be linearized before $op$ and any fast operation that $op$ didn't see must be linearized after $op$.

Moreover, from \Cref{corollary:compute only runs with slow op} all slow operations executed during \codestyle{compute()} have seen the second handshake phase and thus are linearized accordingly.
Let $j$ be the value that \term{determiningSize} obtained from the insertion counter of some thread $T$ in \term{countersSnapshot}.\codestyle{snapshot}.
We will prove that $T$'s $j$-th successful slow \ins{} is linearized before $op$ and $T$'s $(j+1)$-st successful slow \ins{} (if such operation occurs) is linearized after it (the proof for \del{} is the same). 
We will denote $T$'s $j$-th successful slow \ins{} by \ins$_{T,j}$, its linearization time as $t_{insert_{T,j}}$, its original linearization time as $orig\_t_{insert_{T,j}}$ and its linearization point according to \cite{sela2021concurrentSize} as $g\_t_{insert_{T,j}}$. 
The same notations are used accordingly for $T$'s $(j+1)$-st successful slow \ins{} as \ins$_{T,j+1}$, $t_{insert_{T,j+1}}$, $orig\_t_{insert_{T,j+1}}$, $g\_t_{insert_{T,j+1}}$. 

According to the linearizability proof in \emph{Claim B.1} in \cite{sela2021concurrentSize} we know that $g\_t_{insert_{T,j}}<t_{op}$ and together with \Cref{lemma:lin point is before lin point from paper} we get that $t_{insert_{T,j}}<t_{op}$.

If $T$'s $(j+1)$-st successful slow \ins{} has not seen a second handshake then it is linearized according to the original linearization point and from the linearizability proof in \emph{Claim B.1} in \cite{sela2021concurrentSize} we know that $t_{op}<g\_t_{insert_{T,j+1}}$.
Therefore, using \Cref{corollary:compute only runs with slow op} and since $g\_t_{insert_{T,j+1}}$ happens during $insert_{T,j+1}$'s execution of $slow\_op(k)$ we conclude that $t_{op}<t_{insert_{T,j+1}}$.

In the case where $insert_{T,j+1}$ had a concurrent fast operation which was linearized between $orig\_t_{insert_{T,j+1}}$ and $g\_t_{insert_{T,j+1}}$ then $insert_{T,j+1}$ is linearized at the latter of $orig\_t_{insert_{T,j+1}}$ and $sz\_t$ where $sz\_t$ is immediately after the linearization of the last concurrent \size{} operation that did not see $insert_{T,j+1}$.
From the linearizability proof in \emph{Claim B.1} in \cite{sela2021concurrentSize} we know that $t_{op}<g\_t_{insert_{T,j+1}}$.

If $insert_{T,j+1}$ and $op$ are not concurrent then from \Cref{claim:lin within op} and \emph{Claim 8.1} in \cite{sela2021concurrentSize} which states that $g\_t_{insert_{T,j+1}}$ happens within $insert_{T,j+1}$'s execution we conclude that $t_{op}<t_{insert_{T,j+1}}$.

If $insert_{T,j+1}$ and $op$ are concurrent then because we know $op$ does not see $insert_{T,j+1}$ and because $insert_{T,j+1}$ is linearized at the latter of 2 points one of which being immediately after the linearization of the last concurrent \size{} operation that did not see $insert_{T,j+1}$ we conclude that $insert_{T,j+1}$ must be linearized after $op$ thus $t_{op}<t_{insert_{T,j+1}}$.
\end{proof}

The proof of \Cref{claim:results comply} uses the following:

\begin{lemma}\label{lemma: lin point of succ ins/del is after orig lin}
The linearization point of a successful $insert_i(k)$ or $delete_i(k)$ happens in its original linearization point or after it.
\end{lemma}
\begin{proof}
For fast and slow operations which are linearized according to their original linearization point the lemma holds immediately.
Otherwise, the operation is either linearized at the linearization point from \cite{sela2021concurrentSize} which is shown to occur after the original linearization point as part of \emph{Lemma B.3} in \cite{sela2021concurrentSize}, or it is linearized at the latter of two options one of which being the original linearization point.
    
$\Rightarrow$ For each key $k$ and each $i\geq1:$ $orig\_t_{insert_i(k)} \leq t_{insert_i(k)}$ and $orig\_t_{delete_i(k)} \leq t_{delete_i(k)}$.
\end{proof}

\begin{observation}\label{observation:alternating original lps}
    The original linearization points of successful insertions and deletions of each key $k$ are alternating. 
\end{observation}
This follows from the linearizability of the original data structure and the sequential specification of a set.
    
\begin{lemma}\label{lemma:ins and del order}
For each key $k$ and each $i\geq 1$: 
\[
orig\_t_{insert_i(k)} \leq t_{insert_i(k)} < orig\_t_{delete_i(k)} \leq t_{delete_i(k)} < orig\_t_{insert_{i+1}(k)}
\]
\end{lemma}
\begin{proof}
    From \Cref{lemma: lin point of succ ins/del is after orig lin} we know that $orig\_t_{insert_i(k)} \leq t_{insert_i(k)}$ and $orig\_t_{delete_i(k)} \leq t_{delete_i(k)}$. Thus it is only left to show that $t_{insert_i(k)} < orig\_t_{delete_i(k)}$ and $t_{delete_(k)} < orig\_t_{insert_{i+1}(k)}$.

    $t_{insert_i(k)} < orig\_t_{delete_i(k)}$:\\
    If $insert_i(k)$ is a fast operation or a slow operation that has seen a 1st handshake, then it is linearized according to its original linearization point.
    Therefore, $t_{insert_i(k)}=orig\_t_{insert_i(k)}$ and since $orig\_t_{insert_i(k)}<orig\_t_{delete_i(k)}$ by \Cref{observation:alternating original lps}, then $t_{insert_i(k)}<orig\_t_{delete_i(k)}$.
    Otherwise, $insert_i(k)$ is a slow operation that has seen a second handshake and there are 2 cases:
    \begin{itemize}
        \item $delete_i(k)$ is a slow operation - from \emph{Lemma B.3} in \cite{sela2021concurrentSize} we know $g\_t_{insert_i(k)}<orig\_t_{delete_i(k)}$
        and together with \Cref{lemma:lin point is before lin point from paper} we get that $t_{insert_i(k)}\leq g\_t_{insert_i(k)}$. Therefore, $t_{insert_i(k)}<orig\_t_{delete_i(k)}$.
        \item $delete_i(k)$ is a fast operation - from \Cref{lemma: successful slow ins/del linearized when size runs orig_t and t after the other,observation:alternating original lps} we conclude that $t_{insert_i(k)} < orig\_t_{delete_i(k)}$.  
    \end{itemize}

    The proof for $t_{delete_(k)} < orig\_t_{insert_{i+1}(k)}$ is similar with minor changes.

\end{proof}

\begin{corollary}\label{corollary:alternating ins and del}
The linearization points of successful insertions and deletions of each key $k$ are alternating. 
\end{corollary}

\begin{lemma}\label{lemma:update preceding size} 
    Let \term{countersSnapshot} be a \codestyle{HandshakeCountersSnap\-shot} instance.
    Any non-\codestyle{INVALID} value written to a counter in the \term{countersSnapshot}.\codestyle{snapshot} array must have been written to the corresponding counter in the \codestyle{metadataCounters} array (of the \codestyle{SizeCalc\-ulator} instance held by the set) \emph{before} the \term{countersSnapshot}.\codestyle{collecting} field is set to \codestyle{false}.
    \end{lemma}
The proof of \Cref{lemma:update preceding size} remains the same as in \cite{sela2021concurrentSize}.

\begin{lemma} \label{lemma: successful slow ins/del linearized when size runs orig_t and t after the other}
    For any successful \ins{} or \del{} operation $op_s$ that has obtained \codestyle{sizePhase} $\equiv_{4}{2}$ in \Cref{trans_op: read size phase} and for any depending fast operation $op_f$ it holds that: If $orig\_t_{op_s} < orig\_t_{op_f}$ then $t_{op_s}<orig\_t_{op_f}$
\end{lemma}
\begin{proof}
    Let $op_s$ be some successful \ins{} or \del{} operation that has obtained \codestyle{sizePhase} $\equiv_{4}{2}$ in \Cref{trans_op: read size phase} and let $op_f$ be some fast operation such that $orig\_t_{op_s} < orig\_t_{op_f}$.
    If $op_s$ was linearized according to its original linearization point then we are done. 
    Otherwise, from the definition of our linearization points we know there are 2 cases:
    \begin{itemize}
        \item $op_s$ was linearized as in \cite{sela2021concurrentSize}.
        In this case, since we know that there is no depending fast operation whose original linearization point is between $orig\_t_{op_s}$ and $t_{op_s}$ and because every fast operation is linearized at its original linearization point. Then, $orig\_t_{op_s} \leq t_{op_s} < orig\_t_{op_f}$. 
        
        \item $op_s$ is linearized at $sz\_t$ where $sz\_t$ is immediately after the linearization of the last concurrent \size{} operation does not take $op_s$ into account and whose linearization point comes after $orig\_t_{op_s}$. 
        $op_s$ has obtained \codestyle{sizePhase} $\equiv_{4}{2}$ in \Cref{trans_op: read size phase}. Therefore, it must have executed \Cref{trans_op: read size phase} after some \size{} operation has executed \Cref{code: HandshakeSizeCalculator doFirstAndSecondHandshakes second inc} and because $orig\_t_{op_s}$ happens during \Cref{trans_op: execute slowop} then $orig\_t_{op_s}$ happened after that \size{} operation executed \Cref{code: HandshakeSizeCalculator doFirstAndSecondHandshakes second inc}.
        By definition, $sz\_t$ must have happened before that \size{} executed \Cref{code: HandshakeSizeCalculator compute finish handshake} (because \size{} is linearized at \codestyle{compute()}).
        Finally, because $orig\_t_{op_s}<sz\_t$ then both $orig\_t_{op_s}$ and $sz\_t$ happened when some \size{} has finished executing \Cref{code: HandshakeSizeCalculator doFirstAndSecondHandshakes first inc} and has yet to execute \Cref{code: HandshakeSizeCalculator compute finish handshake}. 
        Therefore, from \Cref{lemma:first handshake to last inc no fast op} we conclude that $op_f$ could not have been linearized between $orig\_t_{op_s}$ and $sz\_t=t_{op_s}$ and $orig\_t_{op_s} \leq t_{op_s} < orig\_t_{op_f}$.
    \end{itemize} 
\end{proof} 

In the next following lemmas and observations we will undertake a detailed analysis of the handshake mechanism which is essential for some of the lemmas and claims part of the linearization proof above.

\begin{observation}
    \label{observation: every countersSnapshot obj installed is diff}
    Every \codestyle{HandshakeCountersSnapshot} object is initialized in \Cref{code:HandshakeSizeCalculator init countersSnapshots} and is installed onto the \codestyle{HandshakeSizeCalc\-ulator.countersSnapshot} field exactly once in \Cref{code:HandshakeSizeCalculator compareAndExchange}.
\end{observation}

\begin{lemma}
    \label{lemma: no writes to countersSnapshot when between read and exchange}
    For any \size{} operation denoted as $op_{size}$ that had a successful \codestyle{compareAndExchange} in \Cref{code:HandshakeSizeCalculator compareAndExchange}: no successful writes to the \codestyle{HandshakeSizeCalculator.countersSnapshot} field by other \size{} operations were performed when $op_{size}$'s next line to execute pointed to one of \Crefrange{code: HandshakeSizeCalculator if collecting}{code:HandshakeSizeCalculator compareAndExchange}.
\end{lemma}
\begin{proof}
If any successful writes to the \codestyle{HandshakeSizeCalculat\-or.countersSnapshot} field by other \size{} operations were performed when $op_{size}$'s next line to execute pointed to one of \Crefrange{code: HandshakeSizeCalculator if collecting}{code:HandshakeSizeCalculator compareAndExchange} then from \Cref{observation: every countersSnapshot obj installed is diff} the value of the \codestyle{HandshakeSizeCal\-culator.countersSnapshot} field would have changed from the time $op_{size}$ has read it in \Cref{code: HandshakeSizeCalculator read counterSnapshot} and therefore $op_{size}$'s invocation of the \codestyle{compareAndExchange} operation in \Cref{code:HandshakeSizeCalculator compareAndExchange} would fail in contradiction to the way we defined $op_{size}$.
\end{proof}

\begin{corollary}
    \label{corollary: no interleaving read writes of successful compareandexchange sizes}
    From \Cref{lemma: no writes to countersSnapshot when between read and exchange} we conclude that no two \size{} operations that had a successful \codestyle{compareAndExchange} in \Cref{code:HandshakeSizeCalculator compareAndExchange} (i.e successfully wrote to \codestyle{HandshakeSizeCalculator.countersSnap\-shot}) can have their next line to execute point to \Crefrange{code: HandshakeSizeCalculator if collecting}{code:HandshakeSizeCalculator compareAndExchange} simultaneously.
    In other words, at a given moment there is at most one \size{} operation whose \codestyle{compareAndExchange} is about to succeed and whose next line to execute points to one of \Crefrange{code: HandshakeSizeCalculator if collecting}{code:HandshakeSizeCalculator compareAndExchange}.
\end{corollary}

\begin{lemma}
    \label{lemma: at most one size during doFirst until compute}
    At any given moment, at most one \size{} operation has their next line to execute point to one of \Crefrange{code:HandshakeSizeCalculator compute doFirstAndSecondHandshakes}{code: HandshakeSizeCalculator set collecting false}.
\end{lemma}
\begin{proof}
    We prove by contradiction.
    Let $op_1,op_2$ be two \size{} operations whose next operation to execute points to one of \Crefrange{code:HandshakeSizeCalculator compute doFirstAndSecondHandshakes}{code: HandshakeSizeCalculator set collecting false} at the same time. 
    w.l.o.g assume $op_1$ got to \Cref{code: HandshakeSizeCalculator if collecting} before $op_2$ (i.e $op_1$'s next line to execute pointed to \Cref{code: HandshakeSizeCalculator if collecting} before $op_2$'s next line to execute pointed to it) and we choose $op_2$ to be the first operation whose next line to execute reaches \Crefrange{code:HandshakeSizeCalculator compute doFirstAndSecondHandshakes}{code: HandshakeSizeCalculator set collecting false} after $op_1$.

    For a \size{} operation to reach \Crefrange{code:HandshakeSizeCalculator compute doFirstAndSecondHandshakes}{code: HandshakeSizeCalculator set collecting false} it must enter the \codestyle{if} clause in \Cref{code:HandshakeSizeCalculator if compareAndExchange was successful}. Therefore, its \codestyle{compareAndExchange} operation in \Cref{code:HandshakeSizeCalculator compareAndExchange} must have succeeded. 
    Therefore, from \Cref{corollary: no interleaving read writes of successful compareandexchange sizes} when $op_1$'s next line to execute pointed to one of \Crefrange{code: HandshakeSizeCalculator if collecting}{code:HandshakeSizeCalculator compareAndExchange} then $op_2$'s next line to execute did not point to any of \Crefrange{code: HandshakeSizeCalculator if collecting}{code:HandshakeSizeCalculator compareAndExchange}.
    
    Therefore $op_2$'s next line to execute reached \Cref{code: HandshakeSizeCalculator if collecting} only after $op_1$'s next operation to execute has pointed to \Cref{code:HandshakeSizeCalculator if compareAndExchange was successful} which is after $op_1$ executed its \codestyle{compareAndExchange} operation in \Cref{code:HandshakeSizeCalculator compareAndExchange} and successfully wrote to \codestyle{HandshakeSizeCalculator.countersSnap\-shot}. Therefore, $op_2$ must have read in \Cref{code: HandshakeSizeCalculator read counterSnapshot} the object $op_1$ wrote to \codestyle{HandshakeSizeCalculator.countersSnapshot}. Because we chose $op_2$ to be the first operation after $op_1$ to reach \Crefrange{code:HandshakeSizeCalculator compute doFirstAndSecondHandshakes}{code: HandshakeSizeCalculator set collecting false} and since $op_1$ could not have reached \Cref{code: HandshakeSizeCalculator set collecting false} (because its next line to execute needs to be one of \Crefrange{code:HandshakeSizeCalculator compute doFirstAndSecondHandshakes}{code: HandshakeSizeCalculator set collecting false} when $op_2$ reaches \Cref{code:HandshakeSizeCalculator compute doFirstAndSecondHandshakes}) then the value of the \codestyle{countersSnapshot.collecting} of the object $op_2$ read in \Cref{code: HandshakeSizeCalculator read counterSnapshot} must be true until $op_2$ reaches \Cref{code:HandshakeSizeCalculator compute doFirstAndSecondHandshakes}. Therefore, $op_2$ does enter the \codestyle{if} clause in \Cref{code: HandshakeSizeCalculator if collecting} which is in contradiction to the fact that $op_2$ reaches \Crefrange{code:HandshakeSizeCalculator compute doFirstAndSecondHandshakes}{code: HandshakeSizeCalculator set collecting false} during its execution.
\end{proof}

\begin{lemma} \label{lemma: no concurrent sizes from first inc to finish handshake}
    At any given moment, at most one \size{} operation has their next line to execute point to one of   \multicrefrange{code: HandshakeSizeCalculator doFirstAndSecondHandshakes first inc}{code: HandshakeSizeCalculator doFirstAndSecondHandshakes return}{code:HandshakeSizeCalculator compute collect}{code: HandshakeSizeCalculator compute finish handshake}.
\end{lemma}
\begin{proof}
    We prove by contradiction.
    Let $op_1,op_2$ be two \size{} operations whose next line to execute points to one of   \multicrefrange{code: HandshakeSizeCalculator doFirstAndSecondHandshakes first inc}{code: HandshakeSizeCalculator doFirstAndSecondHandshakes return}{code:HandshakeSizeCalculator compute collect}{code: HandshakeSizeCalculator compute finish handshake} at the same time. 
    w.l.o.g assume $op_1$ got to \Cref{code:HandshakeSizeCalculator compute doFirstAndSecondHandshakes} before $op_2$ (i.e $op_1$'s next line to execute pointed to \Cref{code:HandshakeSizeCalculator compute doFirstAndSecondHandshakes} before $op_2$'s next line to execute pointed to it) and we choose $op_2$ to be the first operation whose next line to execute reaches \Cref{code:HandshakeSizeCalculator compute doFirstAndSecondHandshakes} after $op_1$.

    From \Cref{lemma: at most one size during doFirst until compute} we determine that when $op_2$'s next line to execute points to \Crefrange{code: HandshakeSizeCalculator doFirstAndSecondHandshakes wait until}{code: HandshakeSizeCalculator compute finish handshake} then $op_1$'s next line to execute must point to one of \Crefrange{code:HandshakeSizeCalculator compute computeSize}{code: HandshakeSizeCalculator compute finish handshake}.
    From \Cref{lemma: at most one size during doFirst until compute} and because we chose $op_2$ to be the first operation whose next line to execute reaches \Cref{code:HandshakeSizeCalculator compute doFirstAndSecondHandshakes} after $op_1$ we conclude that from the moment $op_1$'s next line to execute reached \Cref{code:HandshakeSizeCalculator compute computeSize} and until $op_2$'s next to line to execute reached \Cref{code: HandshakeSizeCalculator doFirstAndSecondHandshakes wait until} no writes to \codestyle{sizePhase} were made. Therefore, the first value read by $op_2$ in \Cref{code: HandshakeSizeCalculator doFirstAndSecondHandshakes wait until} must be the value written by $op_1$ in \Cref{code: HandshakeSizeCalculator doFirstAndSecondHandshakes second inc}. Now, from \Cref{code: HandshakeSizeCalculator doFirstAndSecondHandshakes wait until} we know that the value of \codestyle{currSizePhase} in $op_1$'s execution of \Cref{code: HandshakeSizeCalculator doFirstAndSecondHandshakes second inc} must be $\equiv_4 0$ therefore the value written in \Cref{code: HandshakeSizeCalculator doFirstAndSecondHandshakes second inc} by $op_1$ must be $\equiv_4 2$. Therefore, 
    $op_2$ has read in \Cref{code: HandshakeSizeCalculator doFirstAndSecondHandshakes wait until} a value $\equiv_4 2$ and from \Cref{lemma: at most one size during doFirst until compute} we know that no \size{} operation other then $op_1$ and $op_2$ can  write to \codestyle{sizePhase} until $op_2$'s next line to execute reaches \Cref{code:HandshakeSizeCalculator compute computeSize}. Therefore, since $op_1$'s next line to execute must point to one of \Crefrange{code:HandshakeSizeCalculator compute computeSize}{code: HandshakeSizeCalculator compute finish handshake} until $op_2$'s next line to execute reaches \Cref{code: HandshakeSizeCalculator doFirstAndSecondHandshakes first inc} then no writes at all can be made to \codestyle{sizePhase} when $op_2$'s next line to execute points to \Cref{code: HandshakeSizeCalculator doFirstAndSecondHandshakes wait until} making it so that the value of \codestyle{sizePhase} will stay $\equiv_4 2$ and $op_2$ will never reach \Cref{code: HandshakeSizeCalculator doFirstAndSecondHandshakes first inc} in contradiction to the assumption that $op_1$'s and $op_2$'s next line to execute points to one of   \multicrefrange{code: HandshakeSizeCalculator doFirstAndSecondHandshakes first inc}{code: HandshakeSizeCalculator doFirstAndSecondHandshakes return}{code:HandshakeSizeCalculator compute collect}{code: HandshakeSizeCalculator compute finish handshake} at the same time.
\end{proof}

\begin{lemma}\label{lemma:first handshake size phase is 1}
While there is some \size{} operation whose next line to execute points to \Cref{code: HandshakeSizeCalculator doFirstAndSecondHandshakes second inc} then all concurrent \ins{} and \del{} operations whose next line to execute is one of \Crefrange{trans_op: start perform operation}{trans_op: set idle} must have obtained \codestyle{sizePhase} $\equiv_{4}{1}$ in \Cref{trans_op: read size phase}.
\end{lemma} 
\begin{proof} 
    Denote $op_{size}$ some \size{} operation whose next line to execute points to \Cref{code: HandshakeSizeCalculator doFirstAndSecondHandshakes second inc}.
    Let \codestyle{size\_ph} be the \codestyle{sizePhase} set by $op_{size}$ in \Cref{code: HandshakeSizeCalculator doFirstAndSecondHandshakes first inc}. From \Cref{code: HandshakeSizeCalculator doFirstAndSecondHandshakes wait until} we know we could only reach \Cref{code: HandshakeSizeCalculator doFirstAndSecondHandshakes first inc} if \codestyle{currSizePhase} $\equiv_{4}{0}$, therefore, \codestyle{size\_ph} $\equiv_{4}{1}$. From \Cref{lemma: no concurrent sizes from first inc to finish handshake} we know that there can not be any other \size{} operation whose next line to execute points to one of   \multicrefrange{code: HandshakeSizeCalculator doFirstAndSecondHandshakes first inc}{code: HandshakeSizeCalculator doFirstAndSecondHandshakes return}{code:HandshakeSizeCalculator compute collect}{code: HandshakeSizeCalculator compute finish handshake}. Therefore, $op_{size}$ is the only operation that can write to \codestyle{sizePhase}. 
    
    Let $op$ be some \ins{} or \del{} operation whose next line to execute points to \Crefrange{trans_op: start perform operation}{trans_op: set idle} while $op_{size}$'s next line to execute points to \Cref{code: HandshakeSizeCalculator doFirstAndSecondHandshakes second inc}.
    \begin{itemize}
        \item If $op$'s next line to execute has reached \Cref{trans_op: read size phase} before $op_{size}$'s has executed the iteration for $op$'s thread ID in \Cref{perform_hanshake: wait for size phase} when executing \Cref{code: HandshakeSizeCalculator first handshake} then $op$ has changed its relevant \codestyle{opPhase} entry to \codestyle{FAST\_PHASE}. Unless this \codestyle{opPhase} entry is changed again then when $op_{size}$ executes the iteration for $op$'s thread ID in \Cref{perform_hanshake: wait for size phase} it will read \codestyle{FAST\_PHASE} forever making it so that it never continues to the next iteration and so that $op_{size}$ never reaches \Cref{code: HandshakeSizeCalculator doFirstAndSecondHandshakes second inc} in contradiction to $op_{size}$'s definition. Therefore, $op$ must have changed its thread ID's entry in \codestyle{opPhase} to \codestyle{IDLE\_PHASE} or to some other value $x$ such that $x\geq \text{\codestyle{size\_ph}}$.
        If $op$ changed its \codestyle{op\_phase} to \codestyle{IDLE\_PHASE} then its next line to execute must have reached \Cref{trans_op: return} before $op_{size}$'s next line to execute reached \Cref{code: HandshakeSizeCalculator doFirstAndSecondHandshakes second inc} in contradiction to the way we defined $op$. Otherwise, $op$ must have changed its \codestyle{op\_phase} to some value $x$ such that $x\geq \text{\codestyle{size\_ph}}$. Then,
        because $op_{size}$ is the only operation that can write to \codestyle{sizePhase} while $op_{size}$'s next line to execute points to one of \Crefrange{code: HandshakeSizeCalculator doFirstAndSecondHandshakes first inc}{code: HandshakeSizeCalculator doFirstAndSecondHandshakes return} we conclude that $op$ must have read \codestyle{size\_ph} $\equiv_{4}{1}$ in \Cref{trans_op: read size phase}.

        \item Otherwise, if $op$'s next line to execute has reached \Cref{trans_op: read size phase} after $op_{size}$'s has executed the iteration for $op$'s thread ID in \Cref{perform_hanshake: wait for size phase} when executing \Cref{code: HandshakeSizeCalculator first handshake} then $op_{size}$'s next line to execute must have reached \Cref{code: HandshakeSizeCalculator first handshake} before $op$'s next line to execute reached \Cref{trans_op: start perform operation}. Therefore, because $op_{size}$ has executed the write in \Cref{code: HandshakeSizeCalculator doFirstAndSecondHandshakes first inc} before $op$'s next line to execute has reached \Cref{trans_op: start perform operation} then $op$'s read in \Cref{trans_op: read size phase} must have obtained \codestyle{size\_ph} $\equiv_{4}{1}$.
    \end{itemize}
\end{proof}

\begin{lemma}\label{lemma:second handshake size phase is 2}
While there is some \size{} operation whose next line to execute points to any of \Cref{code: HandshakeSizeCalculator doFirstAndSecondHandshakes return}, \Crefrange{code:HandshakeSizeCalculator compute collect}{code: HandshakeSizeCalculator compute finish handshake} then all concurrent \ins{} and \del{} operations whose next line to execute is one of \Crefrange{trans_op: start perform operation}{trans_op: set idle} must have obtained \codestyle{sizePhase} $\equiv_{4}{2}$ in \Cref{trans_op: read size phase}.
\end{lemma} 
\begin{proof} 
    Denote $op_{size}$ some \size{} whose next line to execute points to any of \Cref{code: HandshakeSizeCalculator doFirstAndSecondHandshakes return}, \Crefrange{code:HandshakeSizeCalculator compute collect}{code: HandshakeSizeCalculator compute finish handshake}. Let \codestyle{size\_ph} be the \codestyle{sizePhase} set by $op_{size}$ in \Cref{code: HandshakeSizeCalculator doFirstAndSecondHandshakes second inc}. From \Cref{code: HandshakeSizeCalculator doFirstAndSecondHandshakes wait until} we know we could only reach \Cref{code: HandshakeSizeCalculator doFirstAndSecondHandshakes second inc} if \codestyle{currSizePhase} $\equiv_{4}{0}$, therefore, \codestyle{size\_ph} $\equiv_{4}{2}$. From \Cref{lemma: no concurrent sizes from first inc to finish handshake} we know that there can not be any other \size{} operation whose next line to execute points to one of   \multicrefrange{code: HandshakeSizeCalculator doFirstAndSecondHandshakes first inc}{code: HandshakeSizeCalculator doFirstAndSecondHandshakes return}{code:HandshakeSizeCalculator compute collect}{code: HandshakeSizeCalculator compute finish handshake}. Therefore, $op_{size}$ is the only operation that can write to \codestyle{sizePhase}. 

    Let $op$ be some \ins{} or \del{} operation whose next line to execute points to \Crefrange{trans_op: start perform operation}{trans_op: set idle} while $op_{size}$'s next line to execute points to any of \Cref{code: HandshakeSizeCalculator doFirstAndSecondHandshakes return}, \Crefrange{code:HandshakeSizeCalculator compute collect}{code: HandshakeSizeCalculator compute finish handshake}.
    \begin{itemize}
        \item If $op$'s next line to execute has reached \Cref{trans_op: read size phase} before $op_{size}$'s has executed the iteration for $op$'s thread ID in \Cref{perform_hanshake: wait for size phase} when executing \Cref{code: HandshakeSizeCalculator second handshake} then $op$ has changed its relevant \codestyle{opPhase} entry to \codestyle{FAST\_PHASE}. Unless this \codestyle{opPhase} entry is changed again then when $op_{size}$ executes the iteration for $op$'s thread ID in \Cref{perform_hanshake: wait for size phase} it will read \codestyle{FAST\_PHASE} forever making it so that it never continues to the next iteration and so that $op_{size}$ never reaches \Cref{code: HandshakeSizeCalculator doFirstAndSecondHandshakes return} in contradiction to $op_{size}$'s definition. Therefore, $op$ must have changed its thread ID's entry in \codestyle{opPhase} to \codestyle{IDLE\_PHASE} or to some other value $x$ such that $x\geq \text{\codestyle{size\_ph}}$.
        If $op$ changed its \codestyle{op\_phase} to \codestyle{IDLE\_PHASE} then its next line to execute must have reached \Cref{trans_op: return} before $op_{size}$'s next line to execute reached \Cref{code: HandshakeSizeCalculator doFirstAndSecondHandshakes return} in contradiction to the way we defined $op$. Otherwise, $op$ must have changed its \codestyle{op\_phase} to some value $x$ such that $x\geq \text{\codestyle{size\_ph}}$. Then,
        because $op_{size}$ is the only operation that can write to \codestyle{sizePhase} while $op_{size}$'s next line to execute points to one of \Crefrange{code: HandshakeSizeCalculator doFirstAndSecondHandshakes first inc}{code: HandshakeSizeCalculator doFirstAndSecondHandshakes return} we conclude that $op$ must have read \codestyle{size\_ph} $\equiv_{4}{2}$ in \Cref{trans_op: read size phase}.

        \item Otherwise, if $op$'s next line to execute has reached \Cref{trans_op: read size phase} after $op_{size}$'s has executed the iteration for $op$'s thread ID in \Cref{perform_hanshake: wait for size phase} when executing \Cref{code: HandshakeSizeCalculator second handshake} then $op_{size}$'s next line to execute must have reached \Cref{code: HandshakeSizeCalculator second handshake} before $op$'s next line to execute reached \Cref{trans_op: start perform operation}. Therefore, because $op_{size}$ has executed the write in \Cref{code: HandshakeSizeCalculator doFirstAndSecondHandshakes second inc} before $op$'s next line to execute has reached \Cref{trans_op: start perform operation} and because $op_{size}$ does not execute any other writes to \codestyle{sizePhase} until its next line to execute reaches \Cref{code: HandshakeSizeCalculator return size} then $op$'s read in \Cref{trans_op: read size phase} must have obtained \codestyle{size\_ph} $\equiv_{4}{2}$.
    \end{itemize}
\end{proof}

\begin{corollary}\label{corollary:compute only runs with slow op}
From \Cref{lemma:second handshake size phase is 2} we conclude that when some \size{} operation's next line to execute points to any of \Cref{code: HandshakeSizeCalculator doFirstAndSecondHandshakes return}, \Crefrange{code:HandshakeSizeCalculator compute collect}{code: HandshakeSizeCalculator compute finish handshake}, then only slow operations that have seen the second handshake and are linearized accordingly can reach their operation execution (namely, execute $slow\_op$).
\end{corollary}

\begin{lemma}\label{lemma:first handshake to last inc no fast op}
While there is some \size{} operation whose next line to execute points to one of \multicrefrange{code: HandshakeSizeCalculator doFirstAndSecondHandshakes second inc}{code: HandshakeSizeCalculator doFirstAndSecondHandshakes return}{code:HandshakeSizeCalculator compute collect}{code: HandshakeSizeCalculator compute finish handshake} then all concurrent \ins{} and \del{} operations whose next line to execute is one of \Crefrange{trans_op: start perform operation}{trans_op: set idle} must have obtained \codestyle{sizePhase} $\not\equiv_{4}{0}$ in \Cref{trans_op: read size phase}.
\end{lemma} 
\begin{proof} 
    The only case that is not covered by \Cref{lemma:first handshake size phase is 1} and \Cref{lemma:second handshake size phase is 2} is that in which some \ins{} or \del{} operation $op$'s next line to execute points to \Cref{trans_op: start perform operation} and to \Cref{trans_op: set idle} while there is some \size{} operation (denote as $op_{size}$) whose next line to execute points to \Cref{code: HandshakeSizeCalculator second handshake}. For this to happen, it must be that $op$ starts and finishes executing \Crefrange{trans_op: read size phase}{trans_op: set idle} while $op_{size}$'s next line to execute points to \Cref{code: HandshakeSizeCalculator second handshake}. Therefore, from \Cref{lemma: no concurrent sizes from first inc to finish handshake} and because $op_{size}$ has already executed \Cref{code: HandshakeSizeCalculator doFirstAndSecondHandshakes first inc} by the time $op$ executed \Cref{trans_op: read size phase} we conclude that $op$ must have read the value written by $op_{size}$ in \Cref{code: HandshakeSizeCalculator doFirstAndSecondHandshakes first inc}. From \Cref{code: HandshakeSizeCalculator doFirstAndSecondHandshakes wait until} we know we could only reach \Cref{code: HandshakeSizeCalculator doFirstAndSecondHandshakes first inc} if \codestyle{currSizePhase} $\equiv_{4}{0}$, therefore, the value written by $op_{size}$ in \Cref{code: HandshakeSizeCalculator doFirstAndSecondHandshakes first inc} must be $\equiv_{4}{1}\not\equiv_{4}{0}$.
\end{proof}

\end{document}